\documentclass[11pt]{article}[11pt]
\usepackage{graphicx} 

\usepackage{amssymb}
\usepackage[dvipsnames]{xcolor}
\usepackage{amsmath}
\usepackage{hyperref}
\usepackage{bbm}
\usepackage{array}
\usepackage{float}
\usepackage{subcaption}
\usepackage{bbm}
\usepackage[nameinlink,noabbrev]{cleveref}
\usepackage{booktabs}
\usepackage[shortlabels]{enumitem}
\usepackage{mathtools}
\usepackage{upgreek}
\usepackage{changepage}
\usepackage{setspace}
\usepackage{multirow}

\usepackage[
backend=biber,
sorting=anyt,
defernumbers=true,
style=authoryear,
natbib,
uniquename=false,
useprefix=true,
maxbibnames = 50,
maxcitenames=3,
mincitenames = 1
]{biblatex}
\AtEveryBibitem{\clearfield{month}}

\DeclareFieldFormat[article]{volume}{#1} 
\DeclareFieldFormat[article]{number}{\mkbibparens{#1}}  
\DeclareFieldFormat[article]{pages}{#1}  

\DeclareFieldFormat[incollection]{pages}{#1}

\renewbibmacro*{volume+number+eid}{%
  \printfield{volume}
  \printfield[parens]{number}%
  \setunit{\addcomma}%
  \printfield{eid}}

\renewbibmacro{in:}{}

\DeclareFieldFormat[article,inbook,incollection,inproceedings,patent,thesis,unpublished]{title}{#1}
\DeclareSourcemap{
  \maps[datatype=bibtex]{
    \map[overwrite]{
      \step[fieldsource=doi, null]
      \step[fieldset=url, null]
      \step[fieldset=issn, null]
      \step[fieldset=doi, null]
      \step[fieldset=isbn, null]
    }  
  }
}
\AtEveryBibitem{%
  \clearfield{note} 
}
\DefineBibliographyStrings{english}{%
  page             = {},
  pages            = {},
}

\usepackage[left = 2.5cm , right = 2.5cm, bottom = 2.5cm, top = 2.5cm]{geometry}
\usepackage{titling}
\usepackage{amsthm}

\newtheoremstyle{space}
  {15pt} 
  {8pt} 
  {\itshape} 
  {} 
  {\bfseries} 
  {.} 
  {.5em} 
  {} 

\theoremstyle{space}

\newtheorem{proposition}{Proposition}
\newtheorem{theorem}{Theorem}
\theoremstyle{plain}
\newtheorem{lemma}{Lemma}
\newtheorem{assumption}{Assumption}
\theoremstyle{remark}
\newtheorem{rem}{Remark}[section]
\crefname{assumption}{assumption}{assumptions}
\Crefname{Assumption}{Assumption}{Assumptions}
\crefname{rem}{Remark}{Remarks}

\allowdisplaybreaks

\DeclareMathOperator*{\argmin}{argmin}

\newcommand{\bx}{\mathbf{x}}

\newcommand{\bX}{\mathbf{X}}

\newcommand{\bF}{\mathbf{F}}

\newcommand{\bv}{\mathbf{v}}
\newcommand{\bbv}{\boldsymbol{v}}

\newcommand{\bphi}{\boldsymbol\phi}
\newcommand{\bbeta}{\boldsymbol\beta}

\newcommand{\bLambda}{\boldsymbol\Lambda}
\newcommand{\blambda}{\boldsymbol\lambda}
\newcommand{\bxii}{\boldsymbol\xi}
\newcommand{\bSigma}{\boldsymbol\Sigma}

\newcommand{\bGamma}{\boldsymbol\Gamma}
\newcommand{\bgamma}{\boldsymbol\gamma}

\newcommand{\bchi}{\boldsymbol\chi}

\newcommand{\bO}{\boldsymbol{\mathcal{O}}}

\newcommand{\cred}{\textcolor{black}}

\title{Tail-robust estimation of factor-adjusted vector autoregressive models for high-dimensional time series}

\date{\today}
\author{Dylan Dijk \and Haeran Cho}

\begin{document}

\maketitle

\begin{abstract}
We study the problem of modelling high-dimensional, heavy-tailed time series data via a factor-adjusted vector autoregressive (VAR) model, which simultaneously accounts for pervasive co-movements of the variables by a handful of factors, as well as their remaining interconnectedness using a sparse VAR model.
\cred{To handle heavy tails, we propose an element-wise data truncation step followed by a two-stage estimation procedure for estimating the latent factors and the VAR parameter matrices. 
Assuming the existence of the $(2 + 2\epsilon)$-th moment only for some $\epsilon \in (0, 1)$, we derive the rates of estimation which, making explicit the effect of heavy tails through $\epsilon$, are comparable to the rates attainable in light-tailed settings as $\epsilon \to 1$.}
Numerically, we demonstrate the competitive performance of the proposed estimators on simulated datasets and in an application to forecasting macroeconomics indicators.
\end{abstract}

\vspace{0.2cm}

\textbf{Keywords:} Heavy tail, high-dimensional time series, factor modelling, vector autoregression, truncation

\section{Introduction} \label{sec:intro}

\cred{Vector autoregressive (VAR) models \citep{Lutkephol2005} are popularly adopted for capturing temporal and cross-sectional dependencies among multiple time series \citep{yilmaz2009,Gorrostieta2012,Bay2004}.}
High-dimensional time series are routinely collected in many areas such as finance, economics, medicine, engineering, natural and social sciences. As the dimension of the time series grows, the number of parameters to estimate in a VAR model grows quadratically, and its estimation quickly becomes a high-dimensional problem. Therefore, the assumption of sparsity is often made in conjunction with $\ell_1$-type regularisation-based methods in the literature \citep{Basu2015, Han2015}.

The sparsity assumption, however, may not be suitable for datasets typically observed in economics and finance, which exhibit strong cross-sectional correlations that cannot be accounted for by a small number of non-zero VAR coefficient values. In addition, theoretical investigation into the consistency of the VAR estimation often assumes that the largest eigenvalue of the spectral density matrix is uniformly bounded across the frequency, which is at odds with the behaviour of many real datasets where VAR modelling is popularly adopted. 
As an illustration, Figure~\ref{fig:eigen_fredmd} plots the largest eigenvalues of the sample covariance matrices of increasing dimensions obtained from the macroeconomic data example in Section~\ref{sec:real-data} where, with the number of cross-sections increasing, the largest eigenvalue of the thus-obtained sample covariance matrix grows linearly with the dimensionality.
Consequently, when modelling the dynamic dependence among a large number of time series, an increasingly popular approach is to adjust for the presence of latent factors driving the strong cross-sectional correlations, before a sparse model is fitted, which lends interpretability as well as better accounting for the characteristics of real-life datasets \citep{Fan2020,Fan2023,Fan2024,Barigozzi2024,Krampe2025}.

\begin{figure}[h!t!]
    \centering
    \includegraphics[width=0.4\textwidth]{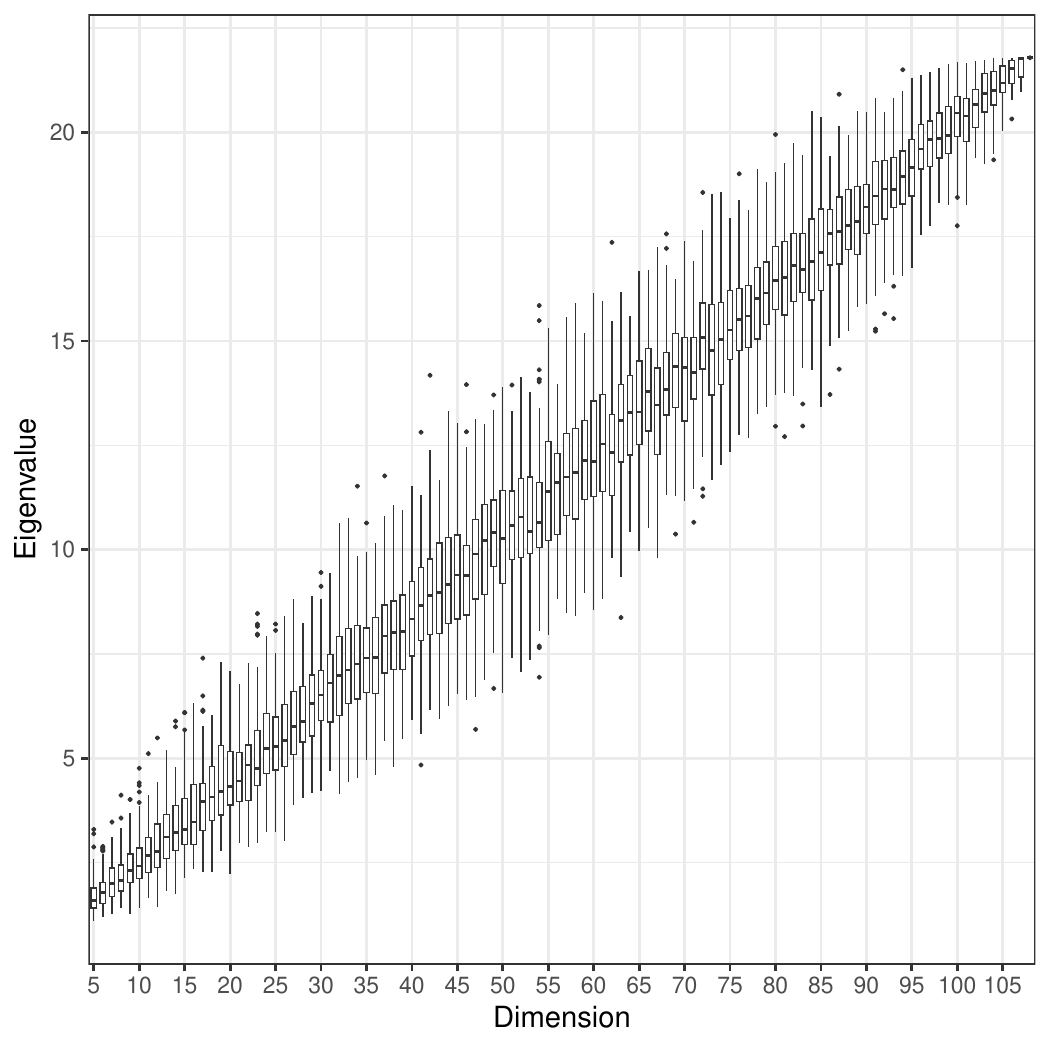}
    \caption{The largest eigenvalue ($y$-axis) of the covariance matrix estimated from the US macroeconomic dataset analysed in \Cref{sec:real-data} (February 1960 to November 2023, $n = 767$) with subsets of cross-sections randomly sampled $100$ times for each given dimension $p \in \{5, 6, \ldots, 108\}$ ($x$-axis).}
    \label{fig:eigen_fredmd}
\end{figure}

In addition to strong cross-sectional correlations, another common characteristic of high-dimensional time series is heavy-tailedness \citep{Fan2021}. It is well-documented that the performance of popular $\ell_1$-regularisation methods such as the Lasso, deteriorates in the presence of heavy tails, see e.g.\ \citet{Wu2016}. 
Similarly, principal component analysis, which is commonly adopted for time series factor modelling, is known to be sensitive to the presence of extreme observations due to heavy-tailedness \citep{Kristensen2014}.

To handle the issues arising from heavy tails, tail-robust estimation methods have been studied in the literature on factor modelling \citep{He2022,He2023,He2025}. \cred{We also mention \citet{Trucos2021} addressing the adjacent yet distinct problem of robust factor modelling and forecasting of high-dimensional time series in the presence of outliers; for the discussion on the distinction between the models for outliers and heavy tails, we refer to \citet{Raymaekers2024}.} 
In sparse high-dimensional linear modelling, the Huber loss has been shown to be robust against both heavy-tailedness and outliers \citep{Li2011, Fan2017, Wang2021}. 
For sparse VAR estimation, \citet{Barbaglia2020} investigate the maximum likelihood estimator assuming that the innovations follow Student's $t$-distributions, 
while \citet{Qiu2015} develop a quantile-based robust estimator under the assumption of an elliptical VAR model.
There are also methods based on data truncation, which has been utilised to mitigate the effects of heavy tails in high-dimensional covariance estimation \citep{Ke2019}, trace regression under the low-rank assumption \citep{Fan2021} or generalised linear modelling \citep{Han2023}.
In the context of VAR modelling, \citet{Liu2021} and \citet{Wang2023} investigate the truncation-based estimators via Lasso and constrained $\ell_1$-minimisation (a.k.a.\ Dantzig selector), under a weak assumption requiring the existence of a finite number of moments only.

\cred{In this paper, we study the problem of factor-adjusted VAR modelling under a weak moment condition, simultaneously addressing the strong cross-sectional dependence, heavy-tailedness and high dimensionality of modern data.
Methodologically, we propose to extend the two-stage estimation procedure of \citet{Barigozzi2024}, which builds upon the consistent estimation of (auto)covariance of the data for estimating the latent factor structure and VAR coefficients.
Specifically, for tail-robust estimation of the second-order properties, we incorporate a simple yet effective data truncation step for handling extreme observations attributed to heavy tails.}

\cred{Assuming the existence of $(2 + 2 \epsilon)$-th moment only, for some $\epsilon \in (0, 1)$, our theoretical investigation fully characterises the performance of the resultant estimators and the choice of the truncation parameter in terms of $\epsilon$.
We further show that as $\epsilon \to 1$, the rates attained by the proposed estimators are comparable to that of the state-of-the-art in time series factor modelling \citep{Bai2003, barigozzi2022asymptotic}, which requires the existence of the fourth moment or more, or in high-dimensional VAR modelling \citep{Basu2015} that calls for (sub-)Gaussianity. 
Additionally, our investigation into the second-stage Lasso estimation of the VAR model is of interest on its own;
in the special case of no factors, we derive a rate for the tail-robust $\ell_1$-regularised estimator which is sharper than those derived in the existing literature \citep{Wang2023, Liu2021}.}

The rest of the paper is organised as below.
\Cref{sec:model_est} introduces the factor-adjusted VAR model. 
\Cref{sec:method} introduces the two-stage estimation method for the factor structure and the VAR model, and \Cref{sec:theory} establishes their theoretical properties under a weak moment assumption.
\Cref{sec:simulations} demonstrates the good performance of the proposed estimators on simulated datasets, and \Cref{sec:real-data} presents an application to a macroeconomic dataset which illustrates the efficacy of the tail-robust method. Proofs of the theoretical results are given in the Appendix, and all the codes used to produce the numerical results are available at: \url{https://github.com/DylanDijk/truncFVAR}.

\paragraph{Notations.} 
By \textbf{I}, \textbf{O}, and $\mathbf{0}$, we denote an identity matrix, a matrix of zeros, and a vector of zeros, whose dimensions depend on the context. For any matrix $\bX  = [x_{ij}] \in \mathbb{R}^{n \times m}$, we denote the $i$-th 
column vector by $\bX_{(i)}$, and the $i$-th 
row vector by $\bX_{i}$ so that we can write $\bX = [\bX_{(1)}, \dots, \bX_{(m)}]  = [\bX_1, \dots, \bX_n]^\top$. 
When $n = m$, we denote the largest and smallest eigenvalues of $\bX$ in moduli by $\Lambda_{\max}(\bX)$ and $\Lambda_{\min}(\bX)$, respectively. The trace of a square matrix $\bX$ is denoted by $\text{Tr}(\bX) = \sum_{i=1}^n x_{ii}$. For a complex-valued matrix $\bX$, $\bX^*$ denotes its conjugate transpose.
For a random variable $X$, its $\nu$-th moment is denoted by $\|X\|_\nu = \mathbb{E}[|X|^{\nu}]^{1/\nu}$ for $\nu \ge 1$. The element-wise matrix $q$-norm is denoted by $|\bX|_q = {\left(\sum _{i=1}^{n}\sum _{j=1}^{m}|x_{ij}|^{q}\right)}^{1/q}$. 
Hence, $|\bX|_2 = \|\bX\|_F$ and  $|\bX|_{\infty}$, denote the Frobenius norm and entry-wise max norm respectively. 
The matrix norms induced by vector $q$-norms are denoted by $\|\bX\|_q$; e.g.\ $\|\bX\|_2$ denotes the spectral norm, and $\|\bX\|_\infty = \max_{1 \le i \le n} \sum_{j = 1}^m \vert x_{ij} \vert$. For any two real-valued sequences $x_k$ and $y_k$, we write $x_k \gtrsim y_k$ if there exists a $C > 0$ such that $x_k \geq Cy_k$ for all $k$. We write $x_k \asymp y_k$ if $x_k \gtrsim y_k$ and $y_k \gtrsim x_k$. For any two real numbers $x$ and $y$, let $x \wedge y$ denote their minimum, and $x \vee y$ denote their maximum. For sequences of random variables $\{X_n\}$ and positive constants $\{a_n\}$, we write $X_n = O_P(a_n)$ if $X_n/a_n$ is stochastically bounded.


\section{Factor-adjusted vector autoregressive model}\label{sec:model_est}


A zero-mean, $p$-variate process $\{\bxii_t\}_{t \in \mathbb{Z}}$ follows a VAR($d$) model if it
satisfies 
\begin{gather} \label{eq:VAR_model}
     \bxii_t = \mathbf{A}_1\bxii_{t-1} + \ldots + \mathbf{A}_d\bxii_{t-d} + \boldsymbol{\varepsilon}_t \, ,
\end{gather}
where $\mathbf{A}_\ell \in \mathbb{R}^{p \times p}$, $1 \leq \ell \leq d$, determine how the values of the vector time series depend on its past. The innovations $\boldsymbol{\varepsilon}_t$ are modelled as a zero-mean white noise process so that $\mathbb{E}(\boldsymbol{\varepsilon}_t) = \mathbf{0}$ and $\mathbb{E}(\boldsymbol{\varepsilon}_t\boldsymbol{\varepsilon}_s^\top) = \bSigma_{\boldsymbol{\varepsilon}} \cdot \mathbb{I}_{\{s = t\}}$ for some positive definite matrix $\bSigma_{\boldsymbol\varepsilon} \in \mathbb{R}^{p \times p}$. 
As the number of parameters in the VAR model grows quadratically with $p$, the assumption of sparsity is commonly found in the literature on high-dimensional VAR modelling.
As noted in the Introduction, however, such an assumption may be inadequate in modelling large time series data exhibiting pervasive cross-sectional correlations, and thus limit the applicability of the model to large time series data arising in economic and financial applications. 


Several approaches have emerged where strong cross-sectional and/or serial correlations are captured by a handful of latent factors, a.k.a.\ factor-adjusted VAR modelling.
In this framework, we observe $\bX_t = (X_{1t}, \dots, X_{pt})^\top$, $1 \leq t \leq n$, that admits the following decomposition
\begin{align}
\mathbf{X}_t = \bLambda \bF_t + \bxii_t = \bchi_t + \bxii_t,
\label{eq:fac_adj_model_def}
\end{align}
where $\bchi_t$ and $\bxii_t$ are latent with $\bxii_t$ following the VAR model as in~\eqref{eq:VAR_model} and $\bchi_t$ is a factor-driven component.
Following \citet{Fan2023} and \citet{Krampe2025}, we adopt a static factor modelling approach where $\bLambda \in \mathbb{R}^{p \times r}$
denotes the matrix of factor loadings, and $\{\bF_t\}_{t \in \mathbb{Z}}$ an $r$-variate vector time series of latent factors which is not necessarily second-order stationary, with the number of factors $r$ fixed. 
We mention that \citet{Barigozzi2024} explore factor-adjusted VAR modelling based on the generalised dynamic factor model (Forni et al. \cite*{Forni2000},\cite*{Forni2015}).
Throughout we assume that the VAR process $\{\bxii_t\}_{t \in \mathbb{Z}}$ is stable, i.e.\ 
all the solutions of the equation \( |\mathbf{\mathcal{A}}(z)| = 0 \) are outside the unit circle where \( \mathbf{\mathcal{A}}(z) = \mathbf{I}_p - \mathbf{A}_1 z - \cdots - \mathbf{A}_d z^d \) is referred to as the AR matrix polynomial in \( z \) \citep{Lutkephol2005}.



The loading matrix and factors are generally not identifiable since for any invertible matrix $\mathbf{R} \in \mathbb{R}^{r \times r}$, we have $\bchi_t = \bLambda \mathbf{R} \mathbf{R}^{-1} \bF_t$. 
Also, neither $\bchi_t$ nor $\bxii_t$ are observable, which calls for a condition that ensures their (asymptotic) identifiability.
To this end, we make the following assumptions. 

\begin{assumption}\label{assump:fac_adj_fac_loading}
\begin{enumerate}[topsep=0pt, label = (\roman*)]
    \item \label{loading_assum} 
    For all $p \ge 1$, $p^{-1} \bLambda^\top \bLambda$ is a diagonal matrix with positive and distinct diagonal entries, and $\mathrm{Cov}(\bF_t) = \mathbf{I}$ for all $t$.
    \item \label{finite_loadings} There exists a constant $M>0$ such that $|\bLambda|_{\infty} \leq M$.
\end{enumerate}
\end{assumption}

\Cref{assump:fac_adj_fac_loading}~\ref{loading_assum} implies that the eigenvalues of $\mathrm{Cov}(\bchi_t)$ diverge linearly in $p$, which in turn 
indicates that the factors are pervasive across the cross-sections of $\bX_t$. 
The identification condition is comparable to that found in \citet{POETFan2011}, and is slightly stronger than those found in \citet{Bai2003} or \citet{stock2002forecasting}. \Cref{assump:fac_adj_fac_loading}~\ref{finite_loadings} is a standard condition, see Assumption~B in \citet{Bai2003}.



\begin{assumption}
\begin{enumerate}[topsep=0pt, label = (\roman*)]\label{assump:VAR}
    \item The smallest and largest eigenvalues of $\mathrm{Cov}(\boldsymbol{\varepsilon}_t) = \bSigma_{\boldsymbol{\varepsilon}}$ are bounded as $0 < \Lambda_{\min}(\bSigma_{\boldsymbol{\varepsilon}}) \le  \Lambda_{\max}(\bSigma_{\boldsymbol{\varepsilon}}) < \infty$.
    \item Denoting $\mu_{\min}(\mathbf{\mathcal{A}}) = \min_{|z| = 1} \Lambda_{\max}(\mathbf{\mathcal{A}}^*(z)\mathbf{\mathcal{A}}(z))$ and $\mu_{\max}(\mathbf{\mathcal{A}}) = \max_{|z| = 1} \Lambda_{\max}(\mathbf{\mathcal{A}}^*(z)\mathbf{\mathcal{A}}(z))$, we have
    $0 < \mu_{\min}(\mathbf{\mathcal{A}}) \le  \mu_{\max}(\mathbf{\mathcal{A}})<\infty$.
\end{enumerate}
\end{assumption}

\citet[][Proposition 2.2]{Basu2015} provide lower and upper bounds for $\mu_{\min}(\mathbf{\mathcal{A}})$ and $\mu_{\max}(\mathbf{\mathcal{A}})$ in terms of the coefficients matrices $\mathbf{A}_{\ell}$'s.
\Cref{assump:VAR} ensures that the eigenvalues of the spectral density matrix of $\{\bxii_t\}_{t \in \mathbb{Z}}$ are bounded away from zero and finite uniformly in the frequency.
Specifically, let us denote the autocovariance matrix of $\{\boldsymbol{\xi}_t\}_{t \in \mathbb{Z}}$ at some lag $h \in \mathbb{Z}$ by $\boldsymbol{\Gamma}_{\boldsymbol{\xi}}(h) = \mathbb{E}(\boldsymbol{\xi}_t \boldsymbol{\xi}_{t - h}^\top)$, with $\bGamma_{\bxii}(-h) = \bGamma_{\bxii}(h)^\top$.
Then, the spectral density matrix of $\{\bxii_t\}_{t \in \mathbb{Z}}$ satisfies
\begin{align}
\boldsymbol{\Sigma}_{\bxii}(\omega) = \frac{1}{2\pi} \sum_{h \in \mathbb{Z}} \bGamma_{\boldsymbol{\xi}}(h) e^{\iota \omega h}, \quad \omega \in[-\pi, \pi],
\end{align}
where $\iota = \sqrt{-1}$.
Further, let us write the smallest and largest eigenvalues of $\boldsymbol{\Sigma}_{\bxii}(\omega)$ across all frequencies as
\begin{gather*}
\mathfrak{m}\left(\boldsymbol{\Sigma}_{{\bxii}}\right)=\underset{\omega \in[-\pi, \pi]}{\operatorname{essinf}} \Lambda_{\min }\left(\boldsymbol{\Sigma}_{\bxii}(\omega)\right) \text{ \ and \ }
\mathcal{M}\left(\boldsymbol{\Sigma}_{{\bxii}}\right)=\underset{\omega \in[-\pi, \pi]}{\operatorname{esssup}} \Lambda_{\max }\left(\boldsymbol{\Sigma}_{\bxii}(\omega)\right) \, .
\end{gather*}
Then, \Cref{assump:VAR} implies that there exists some constants $0 < m_{\bxii} \leq M_{\bxii}$ such that 
\begin{align}
\label{eq:spec_bounds}
  m_{\bxii} \leq \frac{1}{2 \pi} \frac{\Lambda_{\min }\left(\boldsymbol{\Sigma}_{\varepsilon}\right)}{\mu_{\max }(\mathcal{A})} \leq \mathfrak{m}\left(\boldsymbol{\Sigma}_{\bxii}\right)  \leq \mathcal{M}\left(\boldsymbol{\Sigma}_{\bxii}\right) & \leq \frac{1}{2 \pi} \frac{\Lambda_{\max }\left(\boldsymbol{\Sigma}_{\varepsilon}\right)}{\mu_{\min }(\mathcal{A})} \leq M_{\bxii} \, ,
\end{align}
see \citet[][Equation (2.6)]{Basu2015}.
In particular, this ensures that there exists a constant $C>0$ satisfying
\begin{align}
\label{eq:gamma:xi:bound}
 \|\bGamma_{\bxii}(h)\|_2 \leq C \text{ \ for all \ } h \in \mathbb{Z}. 
\end{align}
Altogether, jointly under \Cref{assump:fac_adj_fac_loading} and \Cref{assump:VAR}, the leading $r$ eigenvalues of $\bGamma_\bx = n^{-1}\sum_{t=1}^n \mathbb{E}[ \bX_t \bX_{t}^\top ]$,
are diverging linearly in $p$, while the remaining eigenvalues are bounded for all $p \ge 1$, by Weyl's inequality. 
Consequently, $\bchi_t$ and $\bxii_t$ are asymptotically identifiable as $p \to \infty$, through the presence of a gap in the eigenvalues of $\bGamma_\bx$. 
Then, imposing the sparsity on the VAR coefficients is justifiable under the proposed model~\eqref{eq:fac_adj_model_def}, since pervasive cross-sectional correlations are accounted for by the factors and only the remaining idiosyncratic dynamic dependence is captured by the non-zero coefficients of $\mathbf{A}_\ell, \, 1 \le \ell \le d$.

For the investigation of the tail-robust estimation method to be proposed in the next section, we make the following condition:
\begin{assumption}
\label{assum:tail}
There exist some constants $\epsilon \in (0, 1)$ and $\cred{M_{\epsilon}} > 0$ 
such that 
\begin{align*}
\max_{1 \le t \le n} \left\{ \max_{1 \le i \le p} \mathbb{E}(|\xi_{it}|^{2 + 2 \epsilon}), \, \max_{1 \le j \le r} \mathbb{E}(|F_{jt}|^{2 + 2 \epsilon}) \right\} \leq \cred{M_{\epsilon}}. 
\end{align*}
\end{assumption}
Together with 
Assumption~\ref{assump:fac_adj_fac_loading}~\ref{finite_loadings}, Assumption~\ref{assum:tail} implies that $\mathbb{E}[ \vert X_{it} \vert^{2 + 2\epsilon} ] \lesssim \cred{M_{\epsilon}}$ with $2 + 2\epsilon \in (2, 4)$. 
This permits the data to be considerably heavier-tailed compared to the commonly found conditions in the factor modelling literature, which require the existence of the fourth moment \citep{Bai2003} or even sub-Weibull-type tail behaviour \citep{POETFan2011}. In the high-dimensional VAR modelling literature, Gaussianity of the time series is often assumed \citep{Basu2015, Han2015} with the exception of \citet{Wang2023} and \citet{Liu2021} where conditions analogous to Assumption~\ref{assum:tail} are made. 

\begin{rem}
\label{rem:tail}
\cred{We highlight that our goal is to perform tail-robust estimation of the VAR parameters in the model~\eqref{eq:VAR_model} under the weak moment condition in Assumption~\ref{assum:tail}, which involves estimating
the (auto)covariance of $\{\mathbf{X}_t\}_{t \in [n]}$ and (unobserved) $\{\bxii_t\}_{t \in [n]}$. 
As the quantities of interest are governed by the second moments of the time series, the method we propose in Section~\ref{sec:method} employs truncating off the impact of extreme observations attributed to heavy tails, an approach successfully applied to related problems in high-dimensional time series analysis, such as large covariance matrix estimation \citep{Ke2019}, VAR estimation \citep{Wang2023, Liu2021} and time series tensor factor modelling \citep{barigozzi2024tail}.
Compared to \citet{Wang2023, Liu2021}, the factor-adjusted VAR model in~\eqref{eq:fac_adj_model_def} provides a more realistic approach for modelling strong co-movements of high-dimensional time series as argued in the Introduction, while posing additional challenges from the latency of $\{\bxii_t\}_{t = 1}^n$; also, unlike \citet{barigozzi2024tail}, we permit the factor time series to be stochastic rather than deterministic, and focus on learning the the dynamic dependence structure of the (unobservable) idiosyncratic component after removing latent factors.}

\cred{There are adjacent yet distinct problems; for example, when it is of 
interest to model the extreme outcomes or tail risks directly, one may perform quantile VAR modelling \citep{White2015,ando2022,Carriero2024,Huang2025}.
More generally, the literature on robust statistics is vast and various notions exist to characterise the robustness against perturbations and outliers \citep{huber_robust_1981, hampel_robust_2005}; we refer to \citet{loh2025theoretical} for a review of the recent developments in high-dimensional robust statistics. }
\end{rem}

\section{Methodology} \label{sec:method}

Under the factor-adjusted VAR model in~\eqref{eq:fac_adj_model_def}, to estimate the coefficient matrices $\mathbf{A}_\ell, \, 1 \le \ell \le d$, driving the dynamics in the VAR process $\{\bxii_t\}_{t \mathbb{Z}}$, we first estimate the factor-driven common component~$\bchi_t, \, 1 \le t \le n$. 
To handle heavy tails, we propose to truncate the data prior to carrying out the principal component-based estimation of the factor space. 
The estimated common component is then subtracted from the truncated data to yield the estimator of $\bxii_t, \, 1 \le t \le n$, from which the VAR coefficients are estimated via Lasso. 
Sections~\ref{sec:truncation_and_mod_est} and~\ref{sec:var_est} describe these steps while regarding the number of factors $r$ and the VAR order $d$ as known, and Section~\ref{sec:CV_tune_tau} discusses the selection of $r$, $d$ and other tuning parameters including the level of truncation.

\subsection{Data truncation and factor structure estimation} \label{sec:truncation_and_mod_est}

Given a truncation parameter $\tau > 0$, we perform the element-wise truncation to obtain
\begin{align*} 
\bX_t(\tau) = (X_{1t}(\tau), \dots, X_{pt}(\tau))^\top, \text{ \ where \ }
X_{it}(\tau) = \text{sign}(X_{it}) \cdot (\tau \wedge |X_{it}|),
\end{align*}
with which we estimate $\bGamma_\bx = n^{-1} \sum_{t = 1}^n \mathbb{E}( \bX_t \bX_t^\top )$ by
\begin{align} 
\widehat{\bGamma}_\bx(\tau) = \frac{1}{n} \sum_{t = 1}^n \bX_t(\tau) \bX_t(\tau)^\top \, .\nonumber
\end{align}
\cred{We present the choice of the truncation parameter that balances between the bias (increases as $\tau$ decreases) and the variance (increases with $\tau$) of $\widehat{\bGamma}_\bx(\tau)$, in Equation~\eqref{eq:tau_prop} below, which depends the tail behaviour through $\epsilon$ in \Cref{assum:tail}.}
In practice, we propose to select $\tau$ via cross validation, see \Cref{sec:CV_tune_tau}.
\cred{
Alternative data transformation strategies exist for robust estimation of covariance against outliers \citep{Raymaekers2021}, which has been employed by \citet{Trucos2021} for dynamic factor modelling; we compare our approach to such a method on simulated datasets in \Cref{sec:results}.}

Next, we estimate the loading matrix $\bLambda$ and the factor time series $\bF_t$ under the identifiability condition made in \Cref{assump:fac_adj_fac_loading}~\ref{loading_assum}, as
\begin{gather} \label{eq:factor_pca_solutions}
\widehat{\bLambda}(\tau) = \widehat{\mathbf{E}}_{\bx}(\tau) \widehat{\mathbf{M}}_{\bx}^{1/2} (\tau)\text{ \ and \ }
\widehat{\bF}_t(\tau) = \widehat{\mathbf{M}}_{\bx}^{-1/2}(\tau) (\widehat{\mathbf{E}}_{\bx}(\tau))^\top  \bX_t(\tau), 
\end{gather}
where 
$\widehat{\mathbf{M}}_{\bx}(\tau)$ denotes the diagonal matrix containing the leading $r$ eigenvalues of $\widehat{\bGamma}_\bx(\tau)$, and $\widehat{\mathbf{E}}_{\bx}(\tau)$ the matrix containing the corresponding eigenvectors.  
Then the factor-driven component is estimated by 
$
\widehat{\bchi}_t(\tau) = \widehat{\bLambda}(\tau) \widehat{\bF}_t(\tau) = \widehat{\mathbf{E}}_{\bx}(\tau) ( \widehat{\mathbf{E}}_{\bx}(\tau) )^\top \bX_t(\tau),
$
from which an estimator of the idiosyncratic component is obtained as $\widehat{\bxii}_t(\tau) = \bX_t(\tau) - \widehat{\bchi}_t(\tau)$.

\subsection{VAR parameter estimation}\label{sec:var_est}

Let us denote the matrix collecting all the VAR coefficients for the idiosyncratic VAR process, by $\mathbb{A} = [\mathbf{A}_1, \dots, \mathbf{A}_d] \in \mathbb{R}^{p \times pd}$.
We can re-write the VAR model in~\eqref{eq:VAR_model} as
\begin{align} \underbrace{\left[\begin{array}{c}
\bxii^{\top}_n \\ \vdots \\ \bxii^{\top}_{d+1}\end{array}\right]}_{\boldsymbol{\mathcal{Y}}} & =\underbrace{\left[\begin{array}{ccc}\bxii^{\top}_{n-1} & \cdots & \bxii^{\top}_{n-d} \\ \vdots & \ddots & \vdots \\ \bxii^{\top}_{d} & \cdots & \bxii^{\top}_1\end{array}\right]}_{\boldsymbol{\mathcal{X}}} \underbrace{\left[\begin{array}{c}\mathbf{A}_1^{\top} \\ \vdots \\ \mathbf{A}_d^{\top}\end{array}\right]}_{\mathbb{A}^\top}+\left[\begin{array}{c} \boldsymbol{\varepsilon}^{\top}_n \\ \vdots \\ \boldsymbol{\varepsilon}^{\top}_d\end{array}\right]. \label{eq:stacked_var_regression_form}
\end{align}
While $\boldsymbol{\mathcal{Y}}$ and $\boldsymbol{\mathcal{X}}$ are not observed, with the latent idiosyncratic component estimated by $\widehat{\bxii}_t(\tau)$ as described in Section~\ref{sec:truncation_and_mod_est}, we derive their estimated counterparts as
$\widehat{\boldsymbol{\mathcal{Y}}}(\tau)$ and $\widehat{\boldsymbol{\mathcal{X}}}(\tau)$, respectively.

Then, we propose to estimate $\mathbb{A}$ via $\ell_1$-penalised least squares estimation as
\begin{align}
(\widehat{\mathbb{A}}(\tau))^\top &=
\argmin_{\mathbf{M} \in \mathbb{R}^{pd \times p}} \frac{1}{N} \left\vert \widehat{\boldsymbol{\mathcal{Y}}}(\tau) - \widehat{\boldsymbol{\mathcal{X}}}(\tau) \mathbf{M} \right\vert_2^2 + \lambda \vert \mathbf{M} \vert_1
\nonumber \\
&= \argmin _{\mathbf{M} \in \mathbb{R}^{pd \times p}} \text{Tr}\left( \mathbf{M}^\top \widehat{\bGamma}(\tau) \mathbf{M} - 2\mathbf{M}^\top \widehat{\bgamma}(\tau) \right) + \lambda |\mathbf{M}|_1
\label{eq:idio_objective}
\end{align}
where $\lambda > 0$ is a penalty parameter, $N = n - d$, and the matrices $\widehat{\boldsymbol{\Gamma}}(\tau)$ and $\widehat{\boldsymbol{\gamma}}(\tau)$ are defined as
\begin{align}
\label{eq:gammas}
\widehat{\bGamma}(\tau) = \frac{1}{N} \widehat{\boldsymbol{\mathcal{X}}}(\tau)^\top \widehat{\boldsymbol{\mathcal{X}}}(\tau)  
\text{ \ and \ }
\widehat{\bgamma}(\tau) = \frac{1}{N} \widehat{\boldsymbol{\mathcal{X}}}(\tau)^\top \widehat{\boldsymbol{\mathcal{Y}}}(\tau)   \, .
\end{align}
The matrices $\widehat{\boldsymbol{\Gamma}}(\tau)$ and $\widehat{\boldsymbol{\gamma}}(\tau)$ can be regarded as block matrices composed of sample autocovariance matrices of the estimated idiosyncratic process. We denote their respective estimands by 
\begin{align} \label{eq:gammas_pop}
\bGamma = \begin{bmatrix} \bGamma_{\bxii}(0) & \bGamma_{\bxii}(1) & \dots & \bGamma_{{\bxii}}(d-1) \\
\bGamma_{{\bxii}}(1)^\top & \bGamma_{{\bxii}}(0) & \dots & \bGamma_{{\bxii}}(d-2) \\
\vdots & & \ddots & \vdots \\
\bGamma_{{\bxii}}(d - 1)^\top & \bGamma_{{\bxii}}(d - 2)^\top & \dots & \bGamma_{{\bxii}}(0) \end{bmatrix} \text{ \ and \ }
\bgamma = \begin{bmatrix}
\bGamma_{{\bxii}}(1) \\
\bGamma_{{\bxii}}(2) \\
\vdots \\
\bGamma_{{\bxii}}(d)
\end{bmatrix} \, .  
\end{align}

The computation of~\eqref{eq:idio_objective} can be facilitated by solving the $p$ separate sub-problems in parallel, as
\begin{gather}
\label{eq:robust_var_lasso}
\widehat{\mathbb{A}}(\tau) = \left[\widehat{\bbeta}_{1}(\tau), \dots, \widehat{\bbeta}_{p}(\tau)\right]^\top
\text{ \ with \ }
    \widehat{\bbeta}_{i}(\tau) = \argmin_{\bbeta \in \mathbb{R}^{pd}} \bbeta^\top \widehat{\bGamma}(\tau) \bbeta -2 \bbeta^\top \widehat{\bgamma}(\tau)_{(i)} + \lambda |\bbeta|_1, \, 1 \le i \le p  \nonumber \, .
\end{gather}

\subsection{Tuning parameter selection}\label{sec:CV_tune_tau}

\paragraph{Selection of $\tau$.}
We propose to perform a cross validation (CV) procedure to select the truncation parameter $\tau$.
We denote the set of equi-distanced candidate 
values by $\mathcal{G} = \{\tau_j, \, 1 \le j \le J: \, \tau_1 < \ldots < \tau_J \}$, where $\tau_1 = \text{median}_{i,t} |X_{it}/\widehat{\sigma}_i|$, and $\tau_J = \max_{i,t} |X_{it}/\widehat{\sigma}_i|$, and $\widehat{\sigma}_i$ denotes the median absolute deviation (MAD) of $\{X_{it}\}_{t = 1}^n$ adopted for variable standardisation. 
\cred{For data generated from heavy-tailed distributions, MAD provides a simple yet stable measure of scale that is less affected by extreme realisations \citep{Hampel01061974,Rousseeuw1993}.
While scaling is not a theoretical requirement (cf.\ $M_\epsilon$ in \Cref{assum:tail} is placed on the maximum $(2 + 2\epsilon)$-th moment of $\xi_{it}$ and $F_{jt}$ over $i$ and $j$), 
scaling accounts for differences in scale across the individual series.}
In producing all the numerical results reported in this paper, we set $J = 60$. 

We split the data into two folds, the first and second half, with their sets of indices denoted by $\mathcal{I}_1 = \{1, \dots , \lfloor n/2 \rfloor\}$, and $\mathcal{I}_2 = \{\lfloor n/2 \rfloor + 1, \dots , n\}$. Then, for each $\tau \in \mathcal{G}$, the data are truncated based on a scaled version of $\tau$, as
\begin{gather*}
\bX_t(\tau) = (X_{1t}(\tau_1), \dots, X_{pt}(\tau_p))^\top \text{ \ with \ } \tau_i = \widehat{\sigma}_i \cdot \tau \, ,
\end{gather*}
and the following cross validation metric is obtained:
\begin{gather} \label{eq:CV_measure}
\text{CV}(\tau) = \max_{0 \leq h \leq d} \left\{ \left|\widehat{\bGamma}_{\bx, \mathcal{I}_1}(\tau, h) - \widehat{\bGamma}_{\bx, \mathcal{I}_2}(\infty, h) \right|_{\infty}  \, + \, \left|\widehat{\bGamma}_{\bx, \mathcal{I}_2}(\tau, h) - \widehat{\bGamma}_{\bx, \mathcal{I}_1}(\infty, h)\right|_{\infty} \right\} \, ,
\end{gather}
where $\widehat{\bGamma}_{\bx, \mathcal{I}}(\tau, h)$ denotes the sample autocovariance matrix at lag $h$ obtained from $\{\mathbf{X}_t(\tau), \, t \in \mathcal{I} \}$, with $\tau = \infty$ indicating that no truncation is performed. 
As shown in \Cref{sec:theory}, the theoretical properties of the proposed estimator $\widehat{\mathbb{A}}(\tau)$ are inherited from those of $\widehat{\bGamma}_{\bx}(\tau, h)$, which motivates the CV measure in~\eqref{eq:CV_measure}.
The final choice of $\tau$ is then given by the minimiser of the cross validation measure score, as $\widehat{\tau} = \argmin_{\tau \in \mathcal{G}} \text{CV}(\tau)$. \cred{In Appendix~\ref{sec:additional_simul}, we explore the performance of the proposed CV procedure, see \Cref{fig:tau_vs_p_and_n_auto} which plots the values of $\tau$ selected by the CV procedure over varying $p$ and $n$, alongside the average proportion of data points that are less than the chosen $\tau$ (in absolute value). 
The figure demonstrates that the $\widehat{\tau}$ is generally in agreement with the theoretical choice given in Equation~\eqref{eq:tau_prop} below, and they are chosen as extremal quantiles of $\vert X_{it} / \widehat\sigma_i \vert$.} 

\cred{The computational cost of the CV procedure scales with the grid length $|\mathcal{G}|$, dimension $p$ and sample size $n$, as $O(d|\mathcal{G}|np^2)$. In \Cref{sec:comp_cost} we compare numerically the computational cost of our proposed method to that of a regularised Huber loss method for high-dimensional sparse linear regression \citep{Wang2021}, which complements the estimation performance comparison made in \Cref{sec:results} and demonstrates the competitiveness of our proposal.}

\paragraph{Selection of the VAR order $d$.}
When $p$ is finite, the Akaike information criterion \citep{Akaike1998} and Bayesian information criterion \citep{SchwarzBIC} are frequently adopted for the selection of the VAR order $d$, and there also exists likelihood ratio tests \citep{Tiao01121981}. 
The extended Bayesian information criterion of \citet{ChenEBIC} explores the applicability of an information criterion to high dimensions where $p \to \infty$.
In the context of factor-adjusted VAR modelling, \citet{Barigozzi2024} and \citet{Krampe2025} propose to adopt CV for joint selection of the number of factors, together with the VAR order.
While these approaches have achieved some success, their theoretical and numerical performance is still to be fully understood in high dimensions.
In our simulation studies we assume that the VAR order $d$ is known, focusing on the investigation into the proposed tail-robust estimation strategy. 
In our data application in \Cref{sec:real-data}, we consider a range of VAR order values when performing a forecasting exercise and investigate the impact of its choice.
\cred{In \Cref{sec:comp_cost}, we numerically investigate the performance of the proposed estimator when the VAR order is over-specified. In such a case, $\widehat{\mathbf{A}}_\ell(\tau), \, \ell \ge d + 1$, are close to the matrix of zeros as desired, and the data truncation continues to benefit the estimation performance when compared to an approach without data truncation.}

\paragraph{Selection of the factor number $r$.}
Ihe information criterion-based approach of \citet{Bai2002} is popularly adopted for the selection of the factor number.
We also refer to \citet{Onatski2010}, \citet{Alessi2010}, \citet{Ahn2013} and \citet{Fan2022_fac_number_est} for alternative factor number estimators. 
In Section~\ref{sec:simulations}, we treat $r$ as known, to separate the issues arising from the factor number selection from the investigation into the usefulness of the proposed truncation-based method. In a forecasting exercise performed on macroeconomic data in \Cref{sec:real-data}, we utilise the estimator of \citet{Bai2002} following the practice of \citet{McCracken2015} analysing the same dataset.

\paragraph{Selection of the penalty parameter $\lambda$.}
For the selection of $\lambda$ in \eqref{eq:idio_objective}, we use the CV-based approach which is implemented in the R package \texttt{glmnet} \citep{Friedman2010}.

\section{Theoretical results}\label{sec:theory}



In addition to the assumptions made in \Cref{sec:model_est} for model identifiability and characterising the tail behaviour, we make the following assumptions on the serial dependence in $\{\bX_t\}_{t \in \mathbb{Z}}$.
For any stochastic process $\left\{\mathbf{Y}_t\right\}_{t \in \mathbb{Z}}$, we denote the $\alpha$-mixing coefficient as $\alpha^{\mspace{-1mu} \mathbf{Y}}\mspace{-2mu}(\ell)
=\sup _{s \in \mathbb{Z}} \, \alpha\left(\left\{\mathbf{Y}_t\right\}_{t=-\infty}^s,\left\{\mathbf{Y}_t\right\}_{t=s+\ell}^{\infty}\right)$, where $\alpha\left(\left\{\mathbf{Y}_t\right\}_{t=-\infty}^s,\left\{\mathbf{Y}_t\right\}_{t=r}^{\infty}\right)=\sup_{A, B} |\mathbb{P}(A \cap B)-\mathbb{P}(A) \mathbb{P}(B)|$
with the supremum taken over all events $A \in \sigma\left(\left\{\mathbf{Y}_t\right\}_{t=-\infty}^s\right)$ and $B \in \sigma\left(\left\{\mathbf{Y}_t\right\}_{t=r}^{\infty}\right)$, and $\sigma(\cdot)$ is the sigma field generated by the process. 


\begin{assumption}
\label{assump:fac_adj_fac_idio}


There exists some constant $c \in (0, \infty)$ such that:
\begin{enumerate}[topsep=0pt, label = (\roman*)]
\item $\{\mathbf{F}_t\}_{t \in \mathbb{Z}}$ is $\alpha$-mixing with the coefficients decaying as $\alpha^{\mspace{-1mu} \mathbf{F}}\mspace{-2mu}(m) \leq \exp(-2cm)$. 
\item For any $1 \leq i,j \leq p$, the bivariate process $\{(\xi_{it}, \xi_{jt})\}_{t \in \mathbb{Z}}$ is $\alpha$-mixing with the coefficients decaying as $\max_{1 \le i, j \le p} \alpha^{\mspace{-1mu} \mathbf{\bxii}}_{ij}\mspace{-2mu}(m) \leq \exp(-2cm)$. \label{assump:mix_xii}
\end{enumerate}
\end{assumption}

Mixing-type short range dependence conditions such as \Cref{assump:fac_adj_fac_idio} are found both in the factor modelling \citep{POETFan2011} and the tail-robust VAR modelling \citep{Wang2023} literature. 
%
\Cref{assump:fac_adj_fac_idio}~\ref{assump:mix_xii} places the $\alpha$-mixing condition on the bivariate process $\{(\xi_{it}, \xi_{jt})\}_{t \in \mathbb{Z}}$, which follows if $\alpha$-mixing is assumed on the entire vector process $\{\bxii_t\}_{t \in \mathbb{Z}}$. 
We take the former approach as the latter requires care in high dimensions with $p \to \infty$ \citep{Han2023prob}. 
Finally, we assume:
\begin{assumption}
    \label{assump_fac_adj:8} The processes $\{\bF_t\}_{t \in \mathbb{Z}}$ and $\{\bxii_t\}_{t \in \mathbb{Z}}$ are independent.
\end{assumption}
Jointly, \Cref{assump:fac_adj_fac_idio} and \Cref{assump_fac_adj:8} imply that the bivariate processes formed by the truncated data, $\{(X_{it}(\tau), X_{\cred{j}t}(\tau))\}_{t \in \mathbb{Z}}$, are also strongly mixing. 



The following \Cref{prop:peligrad_inequality_trunc_autocov_error} and \Cref{prop:fac_adj_idio_rate} show that the estimators of the (auto)covariance matrix of $\mathbf{X}_t$ obtained with the truncated data and subsequently, those of the latent $\bxii_t$, are consistent even in ultra-high dimensions, under the weak conditions requiring the existence of the $(2+2\epsilon)$-th moment only (\Cref{assum:tail}). Let us write for some $h \ge 0$,
\begin{align*}
\bGamma_\bx(h) = \frac{1}{n-h} \sum_{t = h + 1}^n \mathbb{E}\left( \bX_t \bX_{t - h}^\top \right) 
\text{ \ and \ }
\widehat{\bGamma}_\bx(\tau, h) = \frac{1}{n-h} \sum_{t = h + 1}^n \bX_t(\tau) \bX_{t - h}^\top(\tau) \, .
\end{align*}
\begin{proposition} \label{prop:peligrad_inequality_trunc_autocov_error}
Suppose that \Cref{assump:fac_adj_fac_loading}, \ref{assum:tail}, \ref{assump:fac_adj_fac_idio} and \ref{assump_fac_adj:8}, hold, and let us set
\begin{align} \label{eq:tau_prop}
   \tau \asymp  \left[\frac{n \cred{\widetilde{M}_{\epsilon}}}{ \log(p) \log^2(n)}\right]^{\frac{1}{2 + 2\epsilon}}  \, 
\end{align}
with $\cred{\widetilde{M}_{\epsilon}} \in (0, \infty)$ that depends only on $r$, $M$ and $\cred{M_{\epsilon}}$, where $M$ and $\cred{M_{\epsilon}}$ are defined in \Cref{assump:fac_adj_fac_loading}~\ref{finite_loadings} and \Cref{assum:tail}, respectively.
Then for any fixed $h \geq 0$, we have
\begin{align} 
\left|\widehat{\bGamma}_\bx(\tau, h) - \bGamma_\bx(h)\right|_\infty &= O_P\left(\left(\frac{\log(p) \log^2(n) \cred{\widetilde{M}_{\epsilon}}^\frac{1}{\epsilon} }{n}\right)^{\frac{\epsilon}{1+\epsilon}}\right) . \nonumber
\end{align}
\end{proposition}

\cred{The choice of $\tau$ fulfilling Equation~\eqref{eq:tau_prop}, is derived by balancing the bias $\vert \mathbb{E}(\widehat{\bGamma}_\bx(\tau, h)) - \bGamma_\bx(h) \vert_\infty$, which increases as $\tau$ decreases, and the deviation $\vert \widehat{\bGamma}_\bx(\tau, h) - \mathbb{E}(\widehat{\bGamma}_\bx(\tau, h)) \vert_\infty$, which grows with increasing $\tau$, and explicitly depends on $\epsilon \in (0, 1)$ from \Cref{assum:tail}.}
 
\begin{rem} \label{rem:alt_robust_method}
\cred{
The subsequent theoretical results rely on the error bound for the
(auto)covariance estimator established in \Cref{prop:peligrad_inequality_trunc_autocov_error}. If an alternative
tail-robust procedure is used to estimate the (auto)covariance matrix,
then analogous results would follow, provided that a comparable error bound
can be established. Element-wise data truncation allows us to directly control the errors of estimating (auto)covariance matrices in the $\ell_\infty$-norm, which is key to establishing the deviation bound in Equation~\eqref{eq:db_bound} used within the proof of \Cref{prop:FNETS_trunc_consistency}.}
\end{rem}

\begin{proposition} \label{prop:fac_adj_idio_rate}
Suppose that \Cref{assump:fac_adj_fac_loading,assump:VAR,assum:tail,assump:fac_adj_fac_idio,assump_fac_adj:8} hold, and
recall $\widehat{\bGamma}(\tau)$ and $\widehat{\bgamma}(\tau)$ in~\eqref{eq:gammas} and their population counterparts $\bGamma$ and $\bgamma$ in~\eqref{eq:gammas_pop}.
Then, $\mathbb{P}(\mathcal{E}_{n, p}) \rightarrow 1$ as $n,p \rightarrow \infty$, where
\begin{align*} 
\mathcal{E}_{n, p}=\left\{\max \left\{\left| \widehat{\bGamma}(\tau) - \bGamma\right|_\infty ,  \left|\widehat{\bgamma}(\tau) - \bgamma\right|_\infty \right\}\le C\left(\left(\frac{\log(p)\log^2(n) \cred{\widetilde{M}_{\epsilon}}^\frac{1}{\epsilon} }{n}\right)^{\frac{\epsilon}{1+\epsilon}} \, \vee \frac{1}{\sqrt{p}}\right) 
 \right\} \; ,
\end{align*}
for some large enough constant $C > 0$, and with $\tau$ chosen as in \eqref{eq:tau_prop}. 
\end{proposition} 

Before presenting the consistency of the VAR parameter estimation, let us introduce some measures of the sparsity of $\mathbb{A} = [\mathbf{A}_1, \dots, \mathbf{A}_d] \in \mathbb{R}^{p \times pd}$, namely $s_{0, j}=\left|\mathbb{A}_{j}\right|_0$, $s_{\text {\upshape max}}=\max _{1 \leq j \leq p} s_{0, j}$ and $s_0=\sum_{j=1}^p s_{0, j}$, representing the number of non-zero coefficients in the $j$-th row of $\mathbb{A}$, their maximum over $j$ and the total number of non-zero coefficients in $\mathbb{A}$, respectively.
Together with \Cref{prop:fac_adj_idio_rate}, the following \Cref{prop:FNETS_trunc_consistency} establishes the consistency of $\widehat{\mathbb{A}}(\tau)$ in~\eqref{eq:idio_objective}.

\begin{proposition} \label{prop:FNETS_trunc_consistency}
Suppose that \Cref{assump:VAR} holds. In addition, assume that
\begin{align}
\label{eq:sparsity}
C s_{\text {\upshape max}} \left\{ \left( \frac{\log(p)\log^2(n)\, \cred{\widetilde{M}_{\epsilon}}^\frac{1}{\epsilon}}{n} \right)^{\frac{\epsilon}{1+\epsilon}} \vee \frac{1}{\sqrt p} \right\} \leq \frac{\pi m_{\bxii}}{16} \, ,
\end{align}
where $C$ is the same constant as that in 
\Cref{prop:fac_adj_idio_rate}, and $m_{\bxii}$ is defined in~\eqref{eq:spec_bounds}.
Let us set
$$
\lambda \asymp (\|\mathbb{A}\|_{\infty} + 1)\left( \left( \frac{\log(p)\log^2(n)\, \cred{\widetilde{M}_{\epsilon}}^\frac{1}{\epsilon}}{ n} \right)^{\frac{\epsilon}{1+\epsilon}} \vee \frac{1}{\sqrt{p}} \right)
$$
with $\|\mathbb{A}\|_{\infty} = \max_{1 \le i \le p} \sum_{j = 1}^p \sum_{\ell = 1}^d \vert A_{\ell, ij} \vert$, and \cred{suppose that $\tau$ is chosen as in \eqref{eq:tau_prop}}.
Then on $\mathcal{E}_{n, p}$, the following statements hold: 
\begin{align}
|\widehat{\mathbb{A}}(\tau) - \mathbb{A}|_2 &\lesssim  \sqrt{s_{0}}(\|\mathbb{A}\|_{\infty} + 1)\left(\left(\frac{\log(p)\log^2(n) \cred{\widetilde{M}_{\epsilon}}^\frac{1}{\epsilon} }{n}\right)^{\frac{\epsilon}{1+\epsilon}} \, \vee \frac{1}{\sqrt{p}}\right) \, , \label{eq:l_2_error_est}\\   |\widehat{\mathbb{A}}(\tau) - \mathbb{A}|_1 &\lesssim s_0 (\|\mathbb{A}\|_{\infty} + 1)\left(\left(\frac{\log(p)\log^2(n) \cred{\widetilde{M}_{\epsilon}}^\frac{1}{\epsilon} }{n}\right)^{\frac{\epsilon}{1+\epsilon}} \, \vee \frac{1}{\sqrt{p}}\right)   \, , \nonumber \\
\|\widehat{\mathbb{A}}(\tau) - \mathbb{A}\|_{\infty} &\lesssim  s_{\text {\upshape max}} (\|\mathbb{A}\|_{\infty} + 1)\left(\left(\frac{\log(p)\log^2(n) \cred{\widetilde{M}_{\epsilon}}^\frac{1}{\epsilon} }{n}\right)^{\frac{\epsilon}{1+\epsilon}} \, \vee \frac{1}{\sqrt{p}}\right)  \, . \nonumber
\end{align}
\end{proposition}
The proof of \Cref{prop:FNETS_trunc_consistency} proceeds by showing that on the event $\mathcal{E}_{n, p}$, the matrix $\widehat{\bGamma}(\tau)$ in \eqref{eq:gammas}
satisfies the so-called restricted eigenvalue condition under~\eqref{eq:sparsity}, along with a deterministic bound on the deviation  $|\widehat{\bgamma}(\tau)_{(j)} - \widehat{\bGamma}(\tau)\mathbb{A}_{j}|_{\infty}$ for all $1 \leq j \leq p$. 
The condition~\eqref{eq:sparsity} on the sparsity requires that for each row equation in~\eqref{eq:VAR_model}, the number of non-zero coefficients is bounded as $O\big((\tfrac{n}{\log(p)\log^2(n)})^{\frac{\epsilon}{1 + \epsilon}} \wedge \sqrt{p} \big)$.

\begin{rem} \label{rem:comparing_VAR_rates}
Under the model \eqref{eq:fac_adj_model_def}, if there is no factor in the model, i.e.\ $\bchi_t = \textbf{0}$, then the rate in  \Cref{prop:fac_adj_idio_rate} improves 
to $\big( \log(p) \log^2(n) \widetilde{M}_{2 + 2\epsilon}^{1/\epsilon} /n\big)^{\frac{\epsilon}{1 + \epsilon}}$. Consequently, 
the rate of estimation for $\widehat{\mathbb{A}}(\tau)$ reported in \Cref{prop:FNETS_trunc_consistency}, also improves analogously, e.g.\ the $\ell_2$-estimation error in~\eqref{eq:l_2_error_est} becomes 
$$
|\widehat{\mathbb{A}}(\tau) - \mathbb{A}|_2 \lesssim  \sqrt{s_{0}}(\|\mathbb{A}\|_{\infty} + 1) \left(\frac{\log(p)\log^2(n) \cred{\widetilde{M}_{\epsilon}}^\frac{1}{\epsilon} }{n}\right)^{\frac{\epsilon}{1+\epsilon}}
$$
which, as $\epsilon \to 1$, becomes comparable (up to log factors) to the rate available under Gaussianity, see e.g.\ \citet{Basu2015} and \citet{Han2015}. 
Combined with data truncation, the constrained $\ell_1$-minimisation estimator of \citet{Wang2023} attains the rate 
$\sqrt{s_0} \Vert \bGamma^{-1} \Vert_\infty \big( \log(p) \log^2(n) \cred{M_{\epsilon}}^{1/\epsilon} / n\big)^{\frac{\epsilon}{1 + \epsilon}}$, under a comparable moment condition while assuming the boundedness of $\|\mathbb{A}\|_{\infty}$ and the strong mixingness imposed on the $p$-variate process $\{\bxii_t\}_{t \in \mathbb{Z}}$.
We note that recalling the matrix $\bGamma \in \mathbb{R}^{pd \times pd}$ defined in~\eqref{eq:gammas_pop}, the term $ \Vert \bGamma^{-1} \Vert_\infty$ may grow with $p$.
Taking an alternative approach to characterising the decay of the serial dependence via a spectral decay index, 
\citet{Liu2021} derive the rate
$\sqrt{s_0} \cred{M_{\epsilon}}^{\frac{1}{2+2\epsilon}} (\Vert \mathbb{A} \Vert_\infty + 1) \big( \log(p)/ n \big)^{\frac{\epsilon}{1 + 2\epsilon}}$;
by comparison, our analysis returns an improved rate of estimation in its dependence on the index $\epsilon$.
\end{rem}

\section{Simulation studies} \label{sec:simulations}



\subsection{Set-up}

We consider the following factor-adjusted VAR process of order one:
\begin{gather*}
    \bX_t = \bLambda \bF_t + \bxii_t
    \text{ \ with \ } 
    \bxii_t = \mathbf{A}\bxii_{t-1} + \boldsymbol{\varepsilon}_t \, .
\end{gather*}
For each of the settings described below,
we generate 200 realisations and vary the dimensions of the data as $(n, p) \in \{(50,100),(50,200),(100,200), (200, 50),(200, 100)\}$. 

\paragraph{Generation of $\bxii_t$.} 
We consider the following settings for the VAR coefficients $\mathbf{A} = [A_{ij}]$.
\begin{enumerate}[label = (V\arabic*), wide, labelindent=0pt]
\item \label{eq:tri_diag_A_sparse_var} \textbf{Banded.} 
$A_{ij} = 0.5 \cdot \mathbb{I}_{\{i = j \}} + \text{\upshape sign}(i - j) \cdot 0.4 \cdot \mathbb{I}_{\{\vert i - j \vert = 1 \}}$. Taken from \citet{Wang2023}. 

\item \label{eq:renyi_A} \textbf{Erd\"{o}s-R\'{e}nyi random graph.} We generate $\mathcal{E}$ as the edge set of a randomly generated directed graph with the link probability $p^{-1}$, and set $\mathbf{A}_0 = [A_{0, ij}]$ with $A_{0, ij} = 0.275 \cdot \mathbb{I}_{\{(i, j) \in \mathcal{E}\}}$.
Then we obtain $\mathbf{A} = \mathbf{A}_0 / \|\mathbf{A}_0\|_2$.
Taken from \citet{Barigozzi2024}.
\end{enumerate}

For the innovation, we set $\boldsymbol{\varepsilon}_t = \bSigma_{\boldsymbol{\varepsilon}}^{1/2} \mathbf{v}_t$ where the entries of $\mathbf{v}_t = (v_{1t}, \ldots, v_{pt})^\top$ are i.i.d.\ random variables with zero mean and unit variance generated as below.
\begin{enumerate}[label = (I\arabic*), wide, labelindent=0pt]
    \item \label{innov_dist_gauss} $v_{it} \overset{\text{iid}}{\sim} \mathcal{N}(0, 1)$.

    \item \label{innov_dist_t} $\sqrt{\tfrac{\nu}{\nu - 2}} \cdot v_{it} \overset{\text{iid}}{\sim} t_\nu$ with $\nu \in \{2.1, 3, 4\}$.

    \item \label{innov_dist_log_norm} $\sqrt{\exp(2) - \exp(1)} \cdot (v_{it} + \exp(1/2)) \overset{\text{iid}}{\sim} \text{lognormal}(0,1)$
\end{enumerate}
 
With~\ref{innov_dist_t} and~\ref{innov_dist_log_norm}, we explore the non-Gaussian scenarios. 
Under~\ref{innov_dist_t}, the random variables $v_{it}$ are symmetric heavy-tailed random variables with finite $k$-th moments only for $k < \nu$, and hence do not have a finite fourth moment.
In the main text, we focus on the case of $\bSigma_{\boldsymbol{\varepsilon}} = \mathbf{I}$ while in Appendix~\ref{sec:additional_simul}, we consider alternative scenarios.

\paragraph{Generation of $\bchi_t$.} We consider the following settings:
\begin{enumerate}[label= (F\arabic*), wide, labelindent=0pt]
\item \label{item:fac_1} Setting $r = 3$, the loadings are generated as $\lambda_{ij} \overset{\text{iid}}{\sim} \mathcal{N}(0, 1)$, where $\lambda_{ij}$ denotes the $(i,j)$-th entry of the loading matrix $\bLambda$. 
The factor time series is generated as a VAR($1$) process $\bF_t = \mathbf{D} \bF_{t-1} + \mathbf{u}_t$, with $\mathbf{D} = 0.7 \cdot \mathbf{D}_0/\Lambda_{\max}(\mathbf{D}_0)$ where the diagonal entries of $\mathbf{D}_0$ are i.i.d.\ $\mathcal{U}[0.5, 0.8]$, and the off-diagonal ones are from $\mathcal{U}[0, 0.3]$. 
The distribution of $\mathbf{u}_t$ follows that of $\boldsymbol{\varepsilon}_t$ in~\ref{innov_dist_gauss}--\ref{innov_dist_log_norm}.
Once $\chi^0_{it} = \boldsymbol{\lambda}_i^\top \mathbf{F}_t$ is generated, we re-scale it such that empirical variance of $\chi_{it}$ matches that of $\xi_{it}$.
Motivated by \citet{POETFan2011} and \citet{Barigozzi2024}.
\item \label{item:fac_0} \textbf{No common component.} $\bchi_t = \mathbf{0}$. 
\end{enumerate}




\subsection{Results} \label{sec:results}

As noted in Section~\ref{sec:CV_tune_tau}, we assume that the true number of factors $r$ and the VAR order $d$ are known when estimating both the factor model and the VAR coefficient matrix. This allows us to separate issues of factor and order selection from the investigation into the performance of the proposed tail-robust estimation method in combination with the CV-based selection of $\tau$. \cred{We investigate the impact of VAR order mis-specification in \Cref{sec:order_missspec}}.
In the main text, we focus on the results on estimating the VAR coefficients, while in Appendices~\ref{sec:cov_sim}--\ref{sec:acv_sim}, we consider the performance of the intermediate estimators of the (auto)covariance of $\bX_t$, which is of independent interest. \cred{We also examine the computational cost of our method in \Cref{sec:comp_cost} in comparison with the $\ell_1$-regularised Huber regression estimator of \citet{adaHuber}.
Then in \Cref{sec:forc_simul}, we assess the forecasting performance on simulated data.}

We measure the error in estimating the VAR coefficient matrix $\mathbb{A}$ by the following maximum column Euclidean norm:
\begin{align}
\label{eq:err:measure}
|(\widehat{\mathbb{A}} - \mathbb{A})^\top|_{2,\infty} =  \max_{1 \le i \le p}  \left| \widehat{\bbeta}_{i} - \bbeta^{*}_{i}\right|_2,
\end{align}
where $\bbeta^{*}_{i}$ denotes the $i$-th row of the true parameter matrix $\mathbb{A} = [\mathbf{A}_1, \ldots, \mathbf{A}_d] = \left[\bbeta^{*}_{1}, \dots, \bbeta^{*}_{p}\right]^\top$.

Firstly, we consider the case where the factor-driven common component is present under~\ref{item:fac_1}. \cred{In \Cref{fig:fac_var_banded_boxplot_cellwise,fig:fac_var_renyi_boxplot}, we compare the proposed method with data truncation to two alternatives: the approach without truncation (`NotTruncated', i.e.\ $\tau = \infty$), and the approach based on applying the `wrapping function' of \citet{Raymaekers2021} to the data (`Wrapping'), while the remaining steps of the two-stage procedure described in Section~\ref{sec:method} are analogously applied.
The latter has been adopted by \citet{Trucos2021} for gaining robustness against outliers in generalised dynamic factor modelling and forecasting, and is implemented in the R package \texttt{cellWise} \citep{cellWise2026}.}
Next, we consider the scenarios where the factor-driven common component is absent under~\ref{item:fac_0}. 
Here, we include the performance from the $\ell_1$-regularised Huber regression estimator (`adaHuber') proposed by \citet{Wang2021} and implemented in the R package \texttt{adaHuber} \citep{adaHuber}. 
While this estimator provides robustness against heavy-tailedness or outliers in the errors, it may be susceptible to heavy-tailed regressors, and can be viewed as an intermediate one between our proposal and the standard, non-truncated counterpart.
We note that the estimators discussed in \Cref{rem:comparing_VAR_rates}, namely those proposed by \citet{Wang2023} and \citet{Liu2021}, are methodologically close to, or even identical to our proposal in the absence of the factor-driven component. 

\begin{figure}[h!t!]
    \centering
    \begin{minipage}[b]{0.32\textwidth}
        \includegraphics[width=\textwidth]{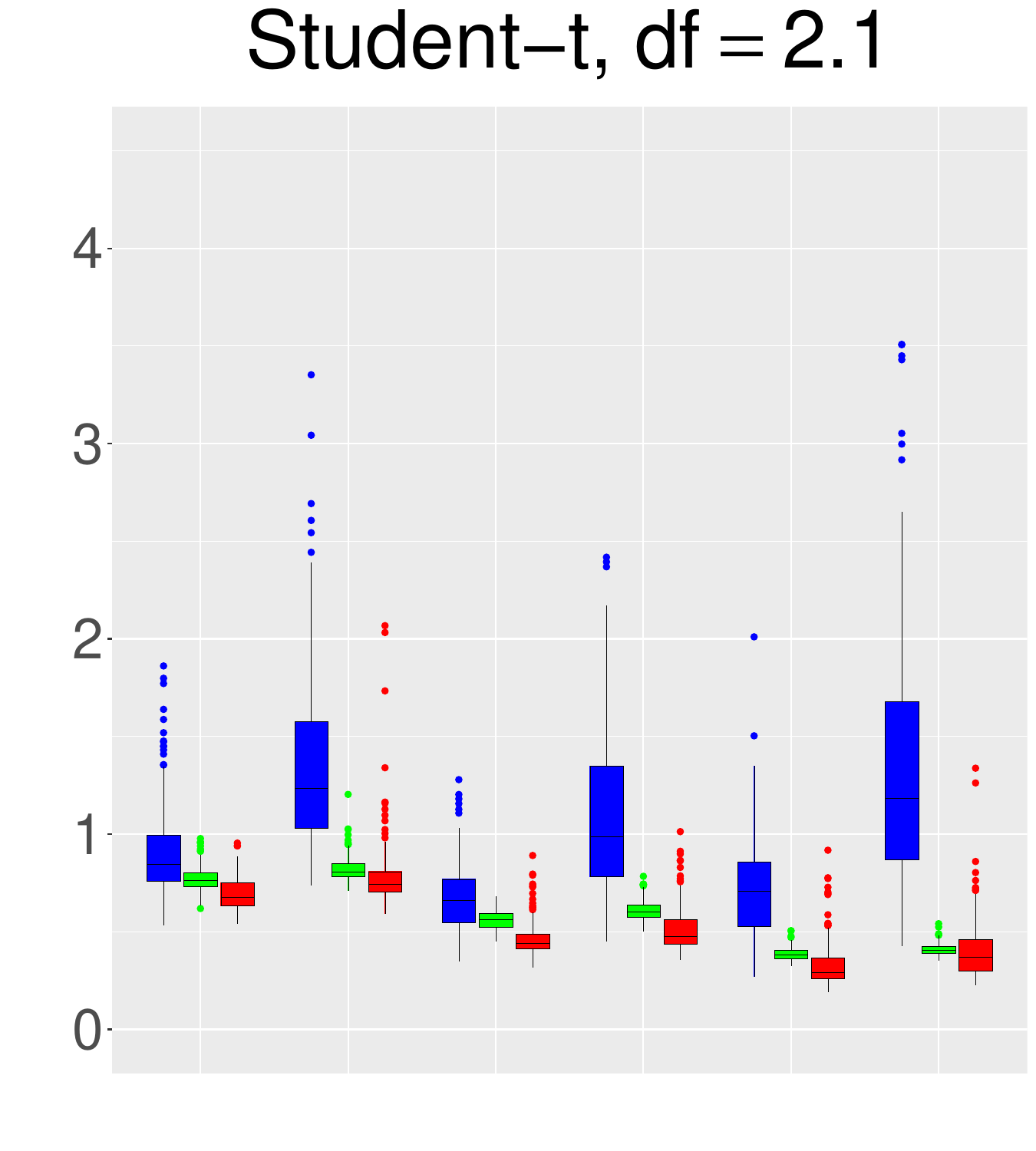}
    \end{minipage}
    \begin{minipage}[b]{0.32\textwidth}
        \includegraphics[width=\textwidth]{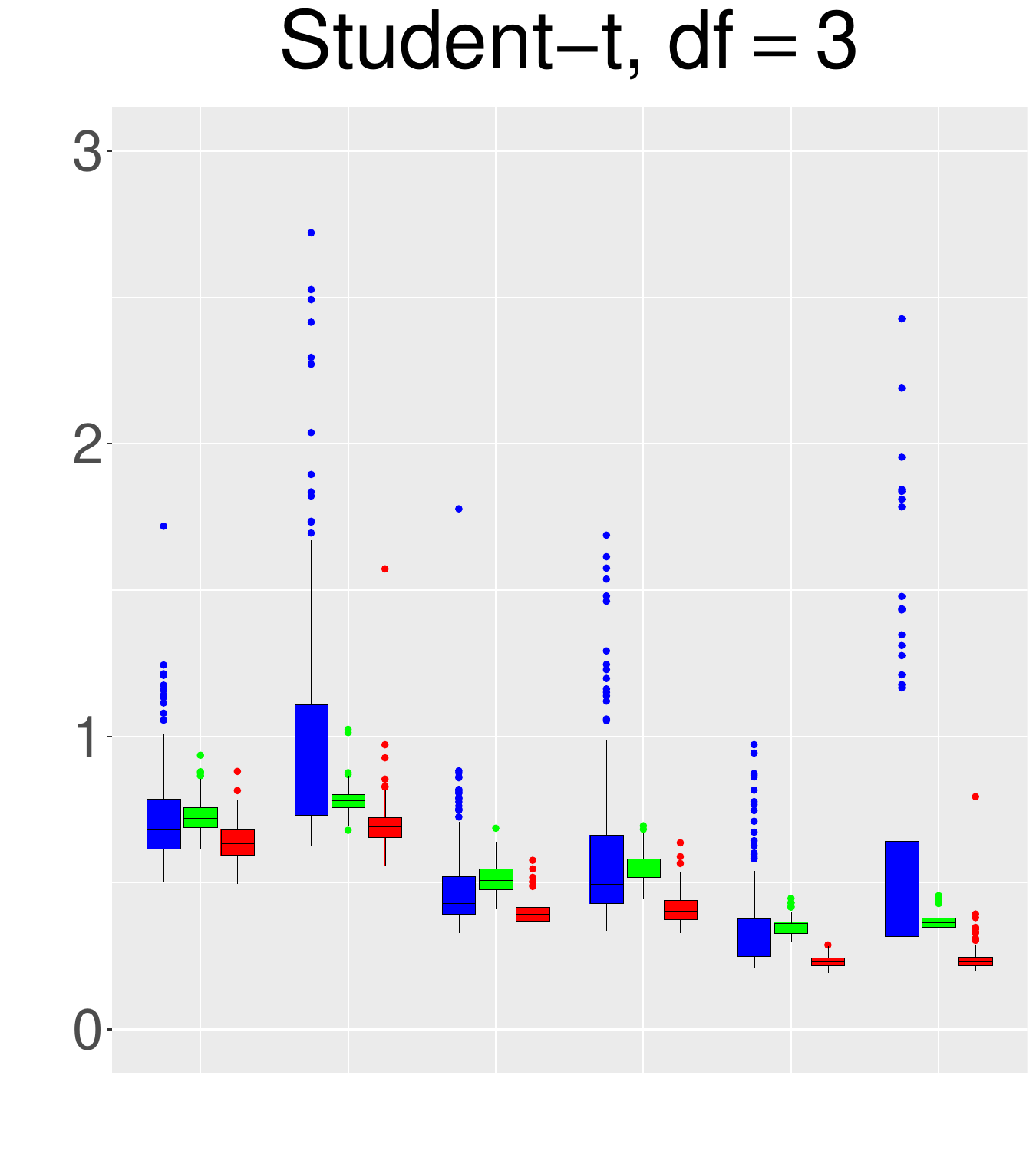}
    \end{minipage}
    \begin{minipage}[b]{0.32\textwidth}
        \includegraphics[width=\textwidth]{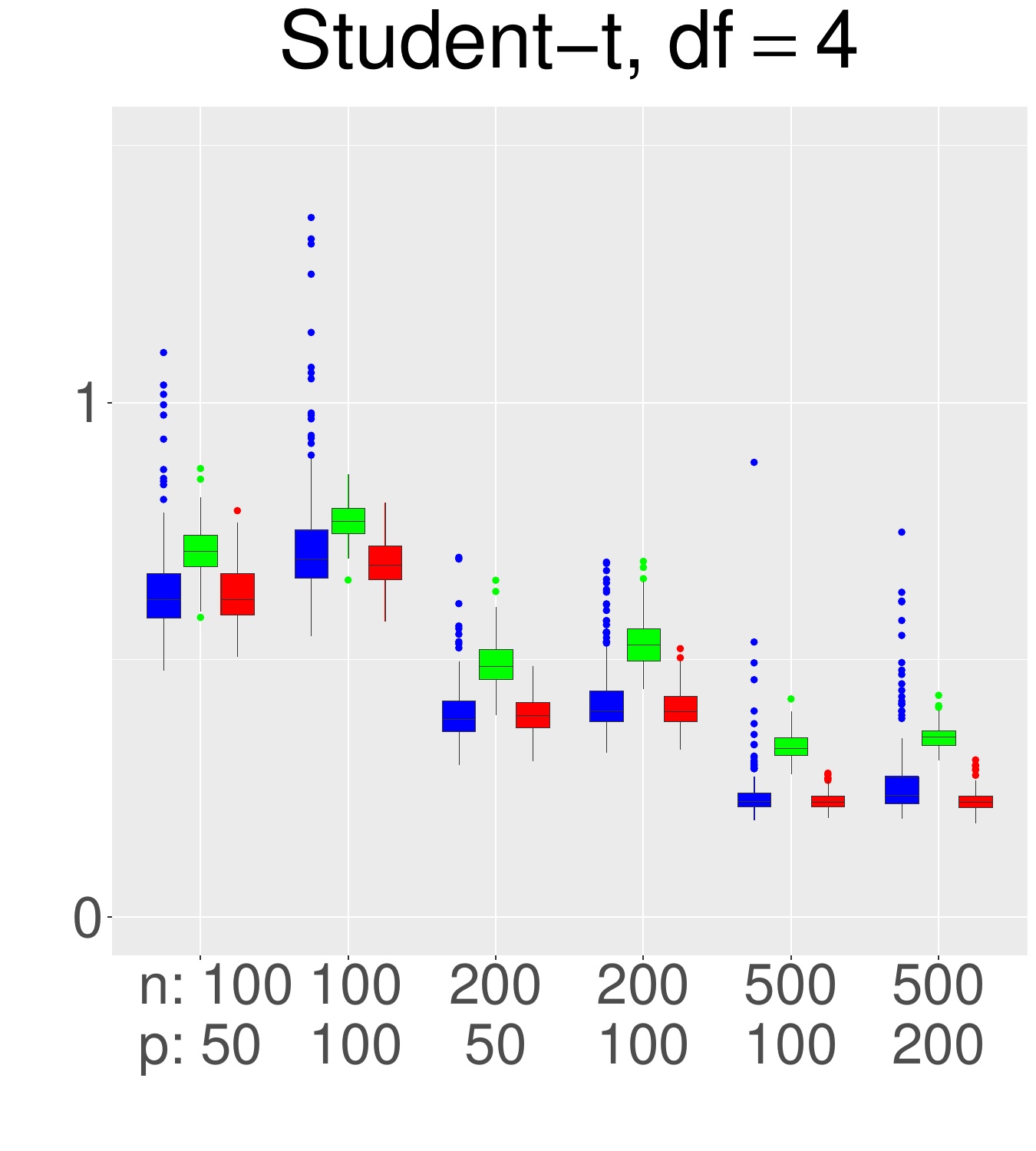}
    \end{minipage}  
    
    \begin{minipage}[b]{0.32\textwidth}
        \includegraphics[width=\textwidth]{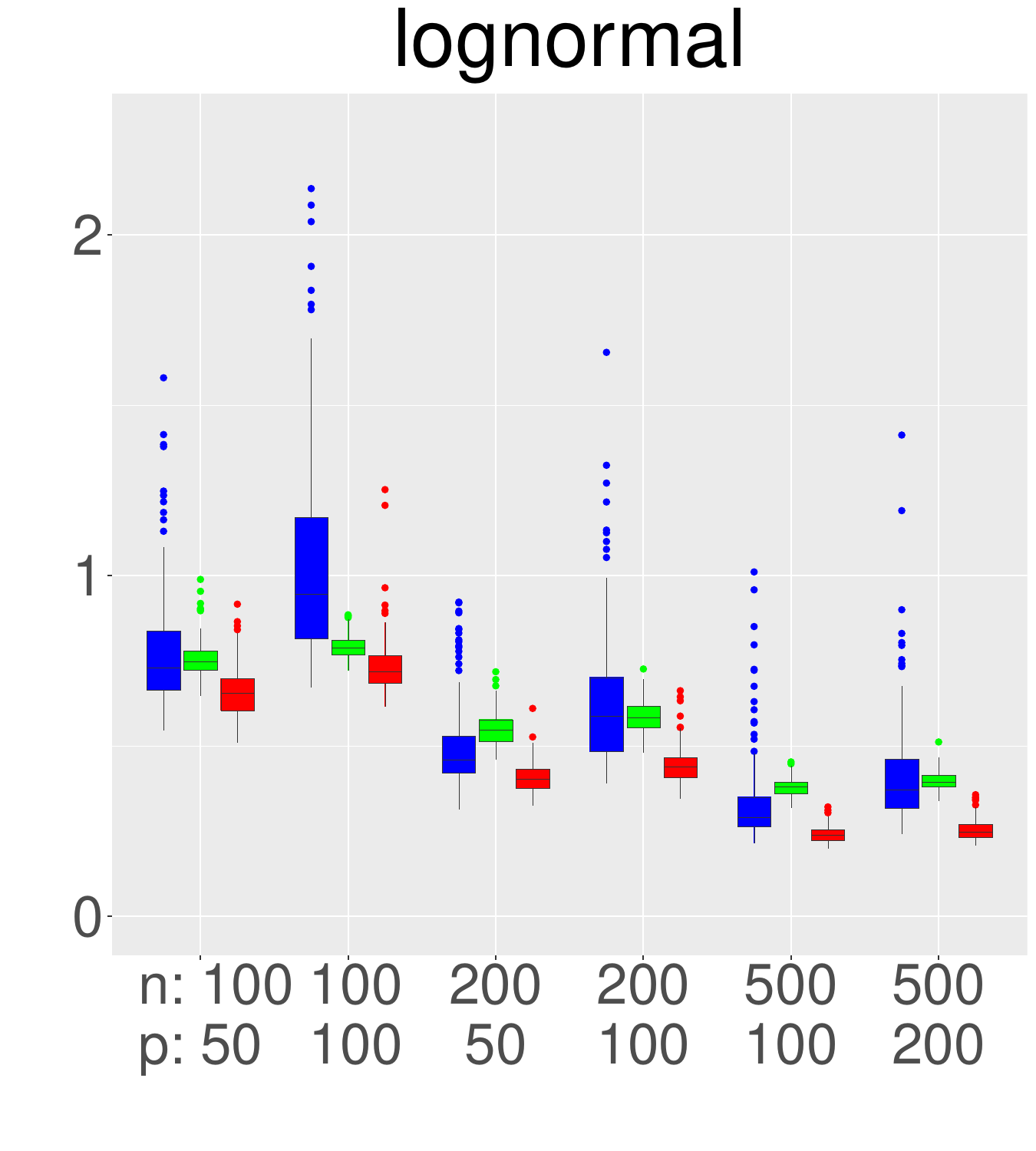}
    \end{minipage}
    \begin{minipage}[b]{0.32\textwidth}
        \includegraphics[width=\textwidth]{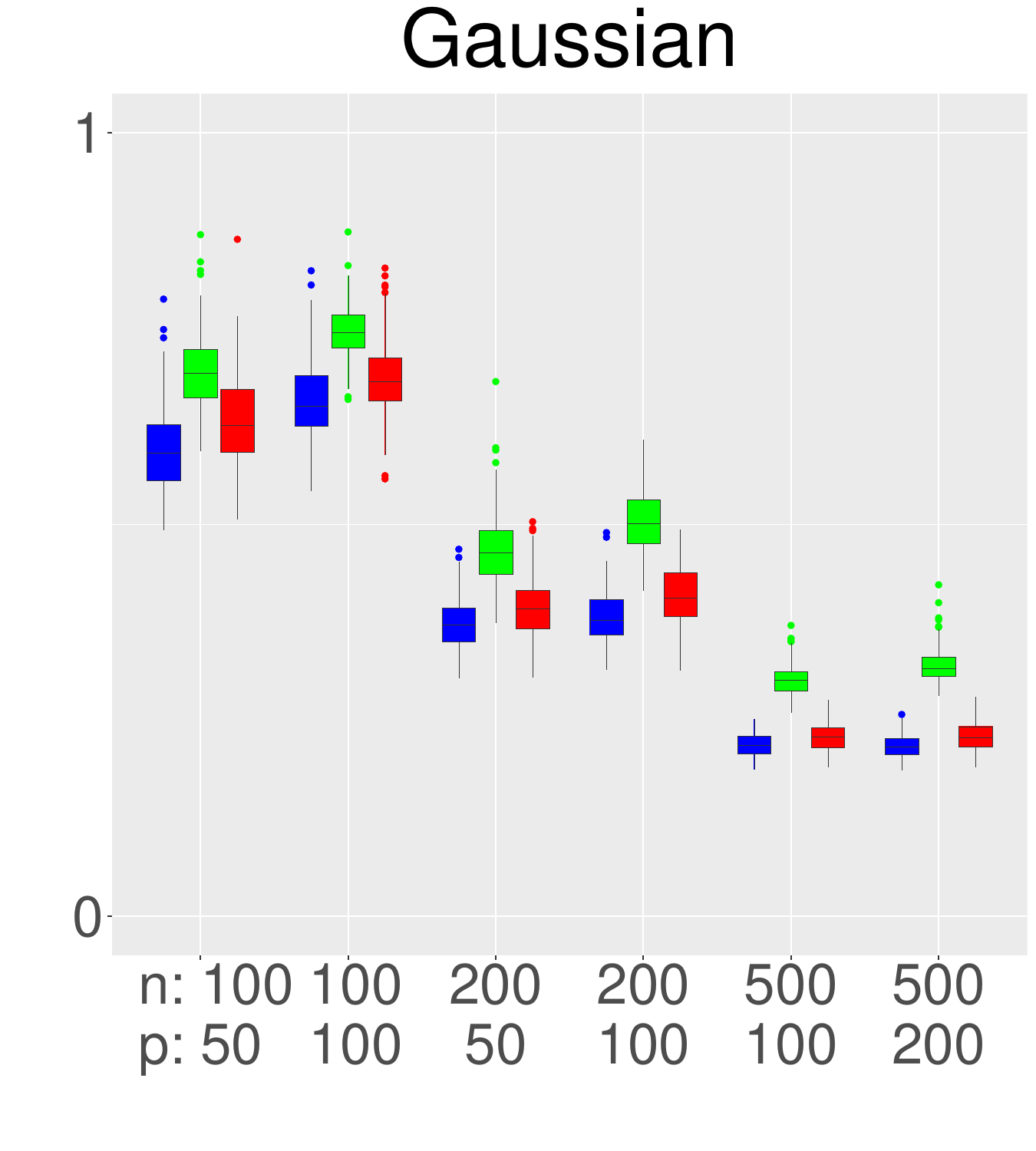}
    \end{minipage}
    \begin{minipage}[c]{0.32\textwidth}
    \vspace{-6cm}
        \includegraphics[width=\textwidth]{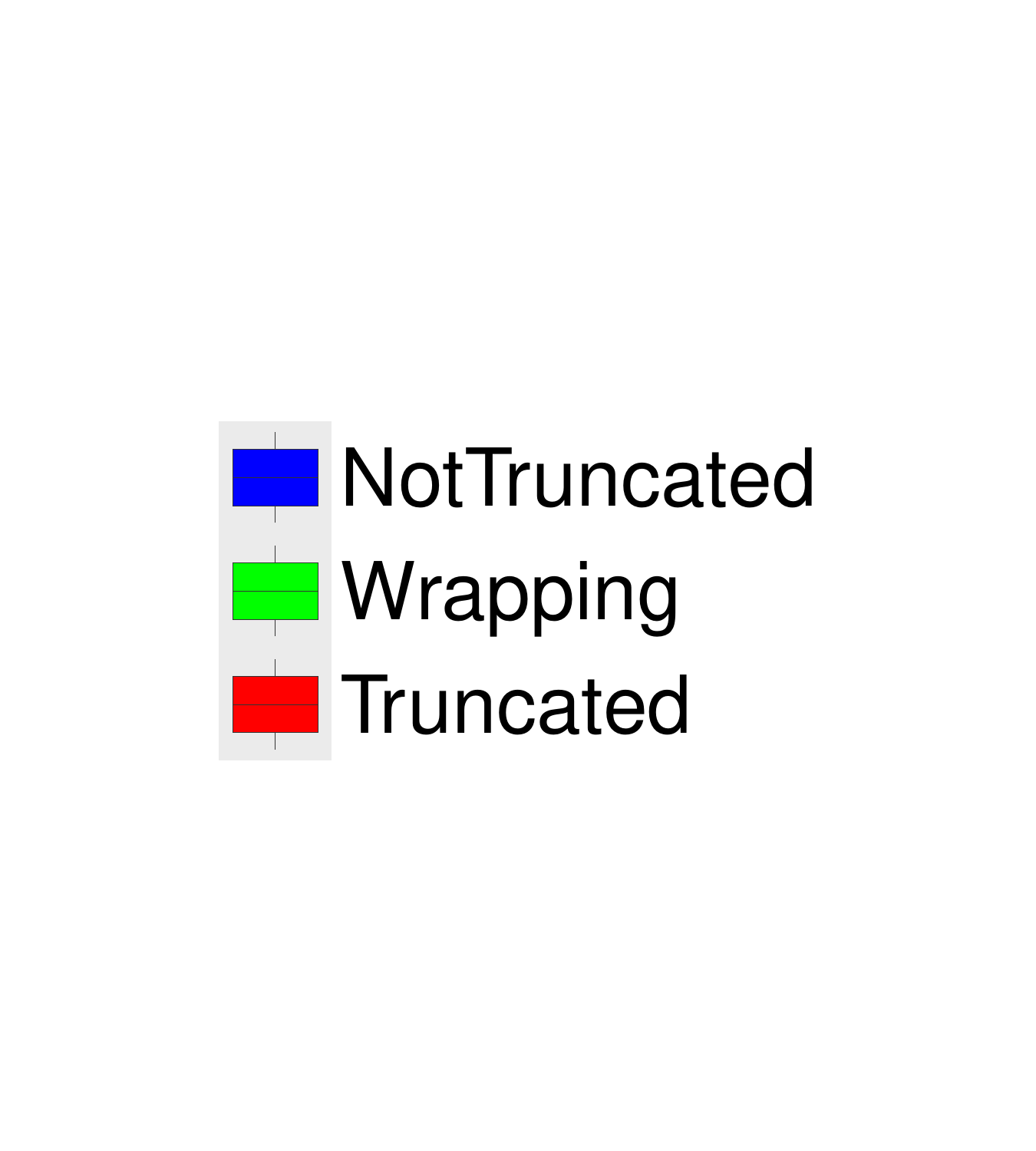}
    \end{minipage}  

    \caption{Boxplots of estimation errors measured as in~\eqref{eq:err:measure} over $200$ realisations as $(n, p)$ and the innovation distribution vary (see~\ref{innov_dist_gauss}--\ref{innov_dist_log_norm}), when the data are generated as in \ref{item:fac_1} and \ref{eq:tri_diag_A_sparse_var}.}
    \label{fig:fac_var_banded_boxplot_cellwise}
\end{figure}

\begin{figure}[h!t!]
    \centering
    \begin{minipage}[b]{0.32\textwidth}
        \includegraphics[width=\textwidth]{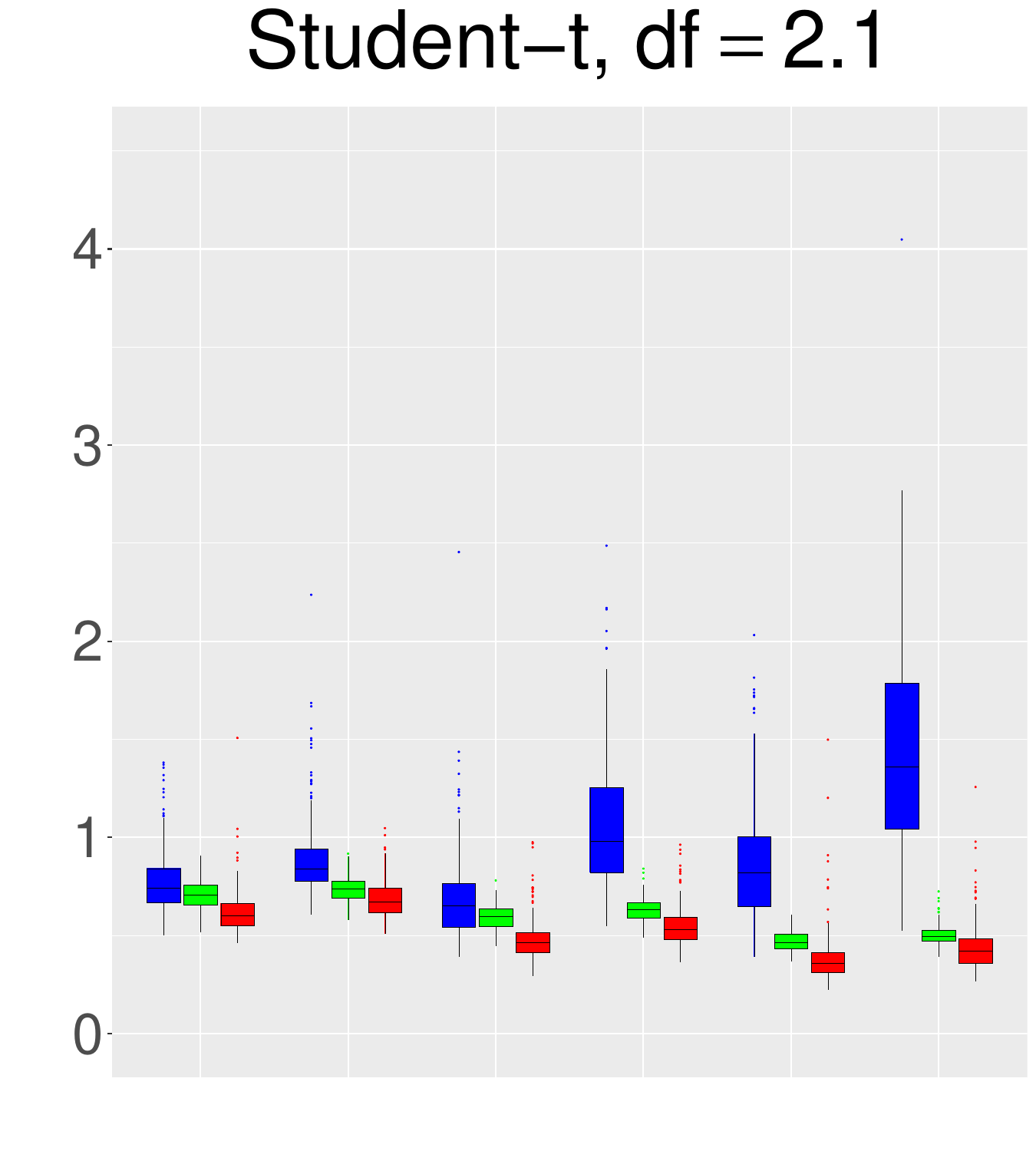}
    \end{minipage}
    \begin{minipage}[b]{0.32\textwidth}
        \includegraphics[width=\textwidth]{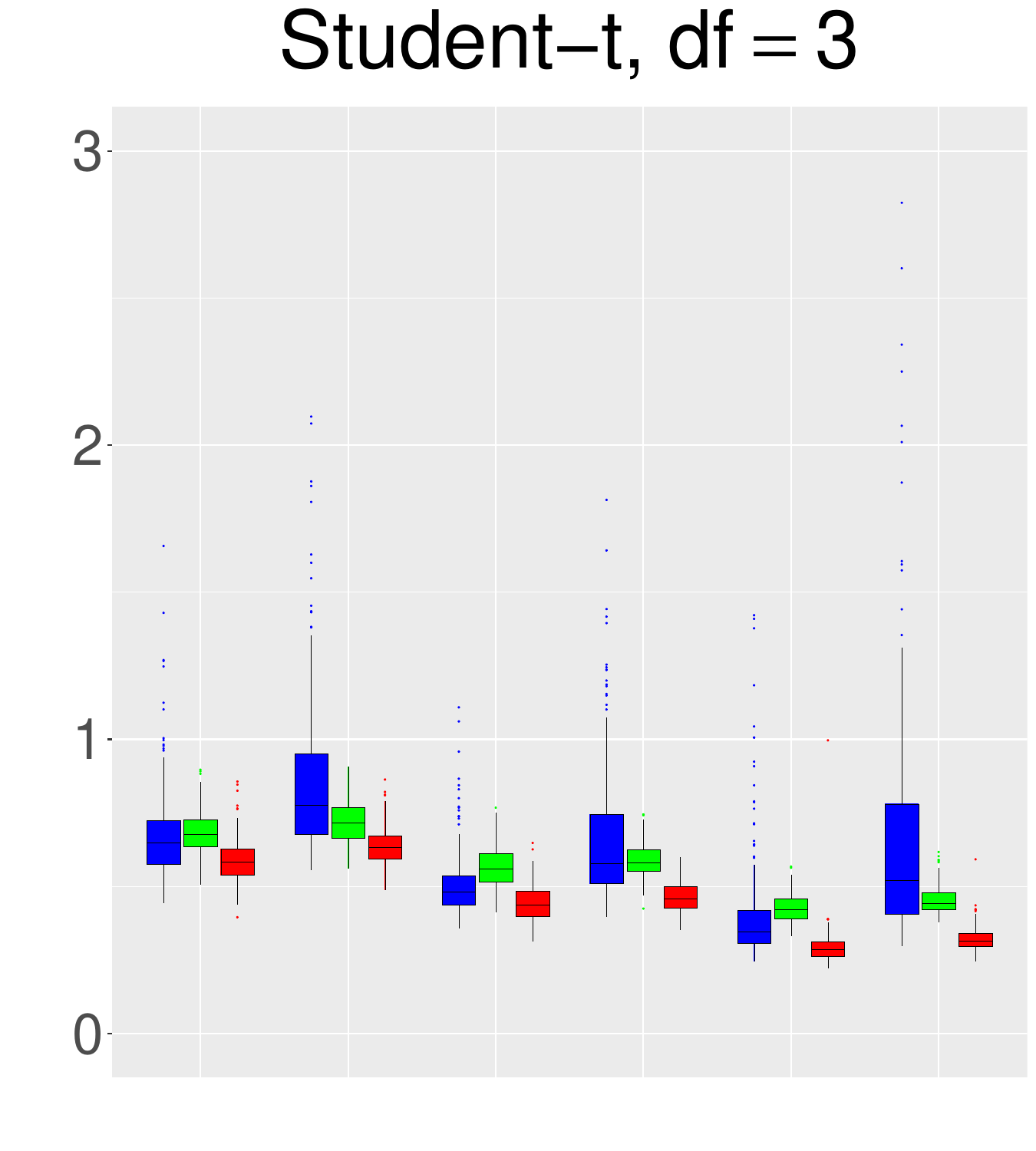}
    \end{minipage}
    \begin{minipage}[b]{0.32\textwidth}
        \includegraphics[width=\textwidth]{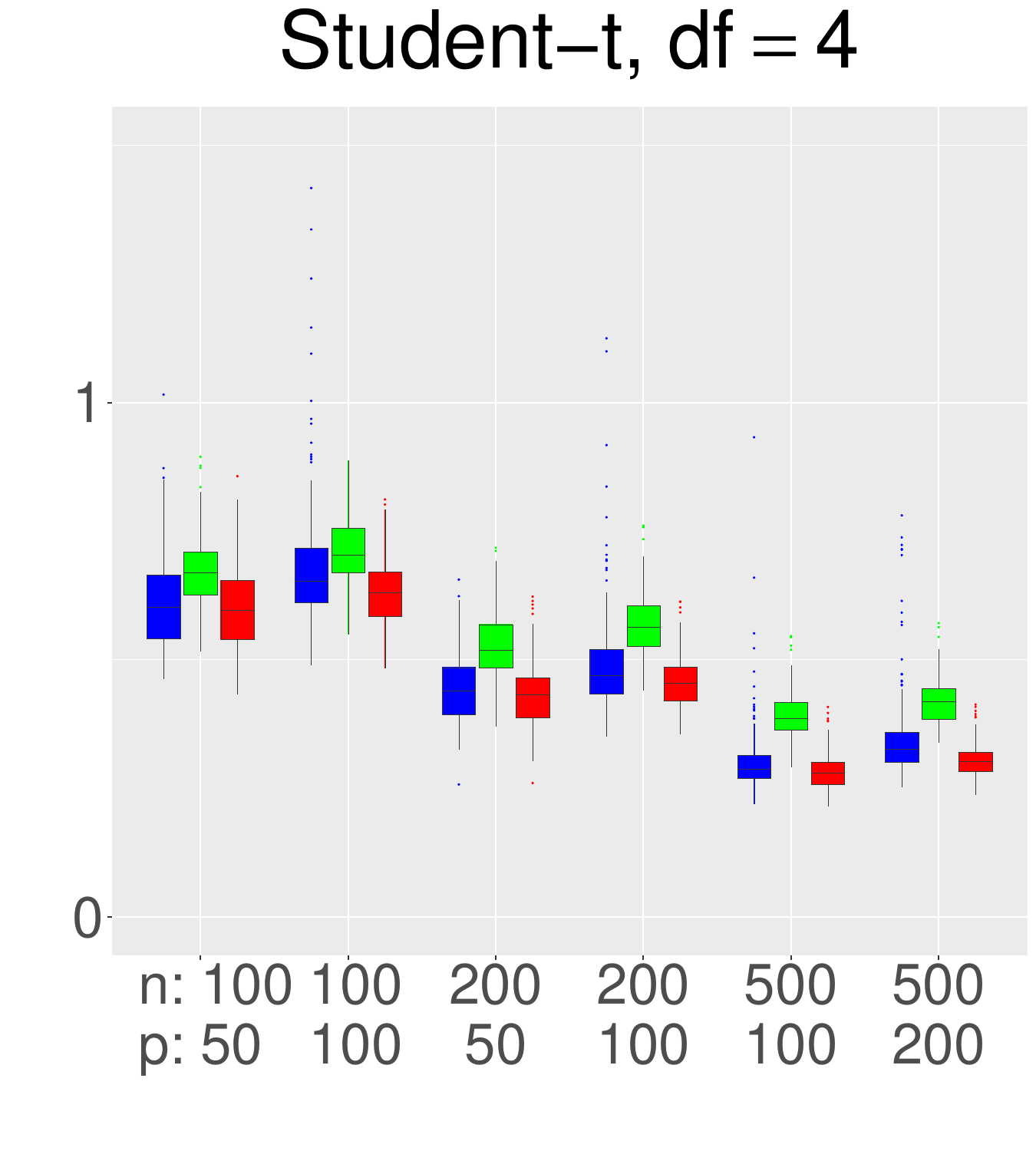}
    \end{minipage}  
    
    \begin{minipage}[b]{0.32\textwidth}
        \includegraphics[width=\textwidth]{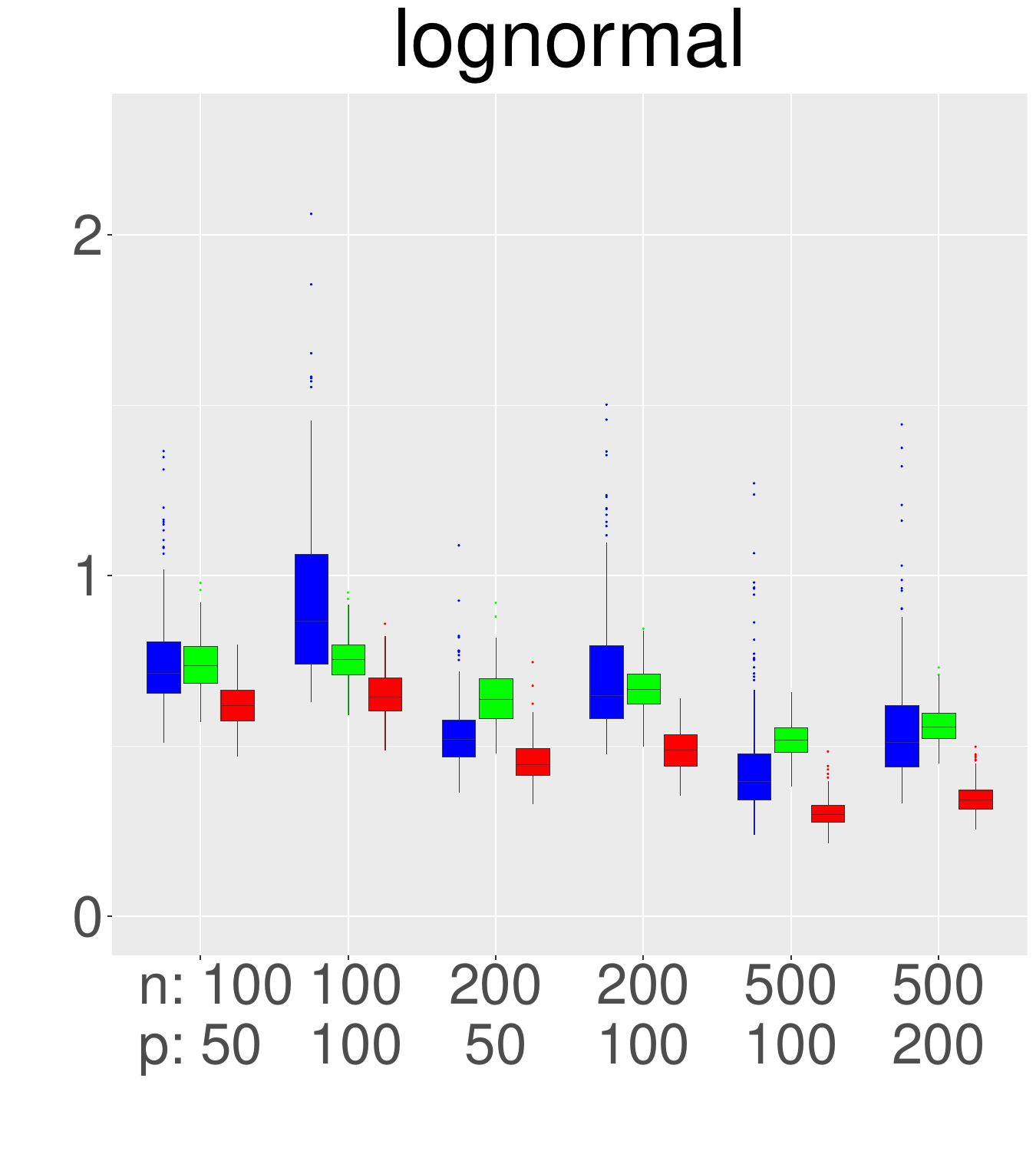}
    \end{minipage}
    \begin{minipage}[b]{0.32\textwidth}
        \includegraphics[width=\textwidth]{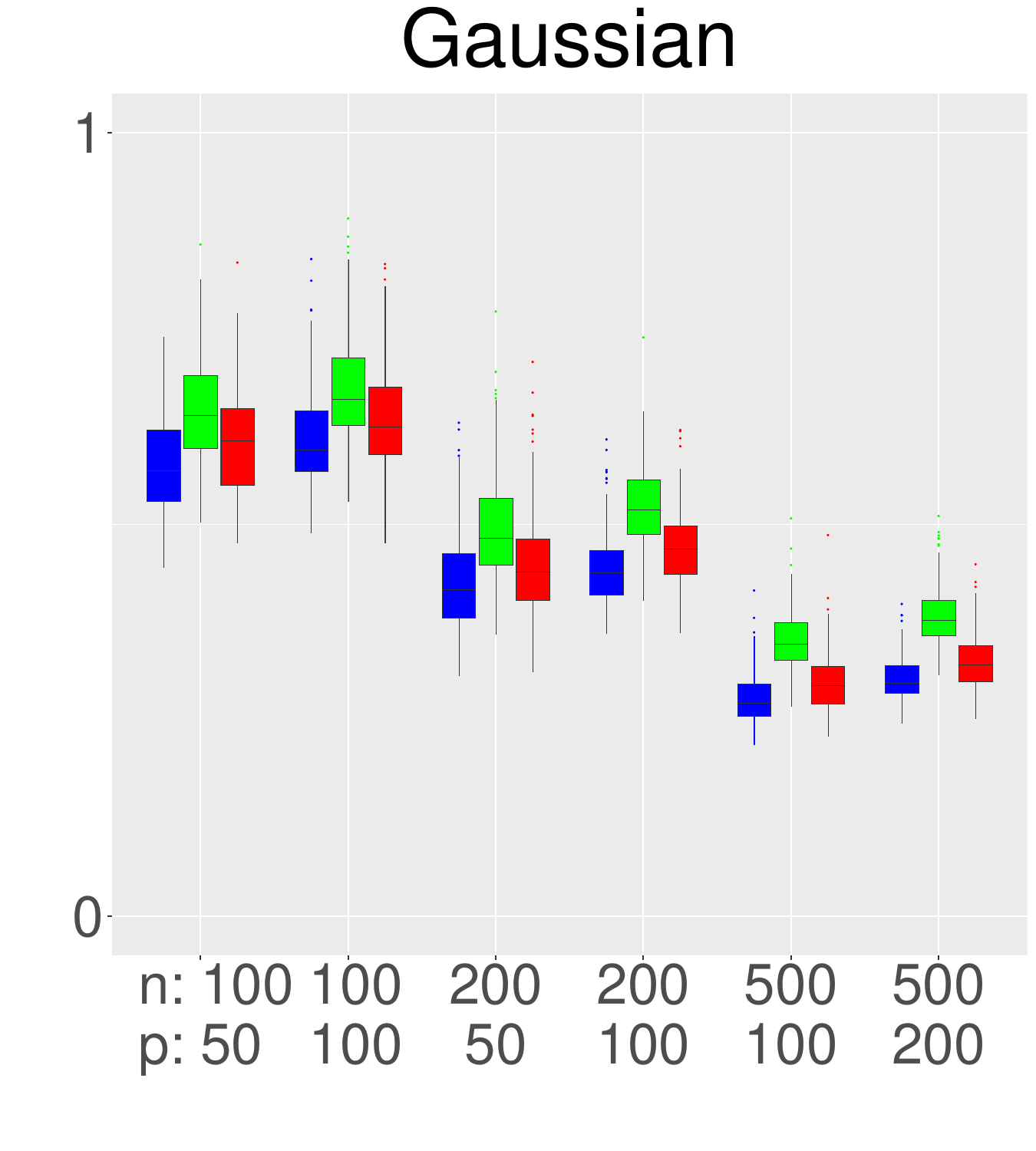}
    \end{minipage}
    \begin{minipage}[c]{0.32\textwidth}
    \vspace{-6cm}
        \includegraphics[width=\textwidth]{legend_wrapping.pdf}
    \end{minipage}

    \caption{Boxplots of estimation errors measured as in~\eqref{eq:err:measure} over $200$ realisations as $(n, p)$ and the innovation distribution vary (see~\ref{innov_dist_gauss}--\ref{innov_dist_log_norm}), when the data are generated as in \ref{item:fac_1} and \ref{eq:renyi_A}.}
    \label{fig:fac_var_renyi_boxplot}
\end{figure}

Overall, \Cref{fig:fac_var_banded_boxplot_cellwise},~\ref{fig:fac_var_renyi_boxplot} and~\ref{fig:comparing_adaHuber_lasso_vs_p_boxplot_outliers} show that as we depart
away from Gaussianity, the relative performance of the proposed data truncation-based method, is dramatically better than the standard or the adaHuber method; under Gaussianity, it is only marginally worse than the standard approach.
In fact, the truncation-based estimator attains the errors comparable to those
attained under Gaussianity, regardless of the innovation distribution. 
It is of interest to note that adaHuber performs substantially worse than our method or the standard one in Gaussian scenarios.
The results also show that the gap between our method and the competitors increases as the dimension to sample size ratio increases. 

\begin{figure}[h!t!]
    \centering
    \begin{minipage}[b]{0.32\textwidth}
        \includegraphics[width=\textwidth]{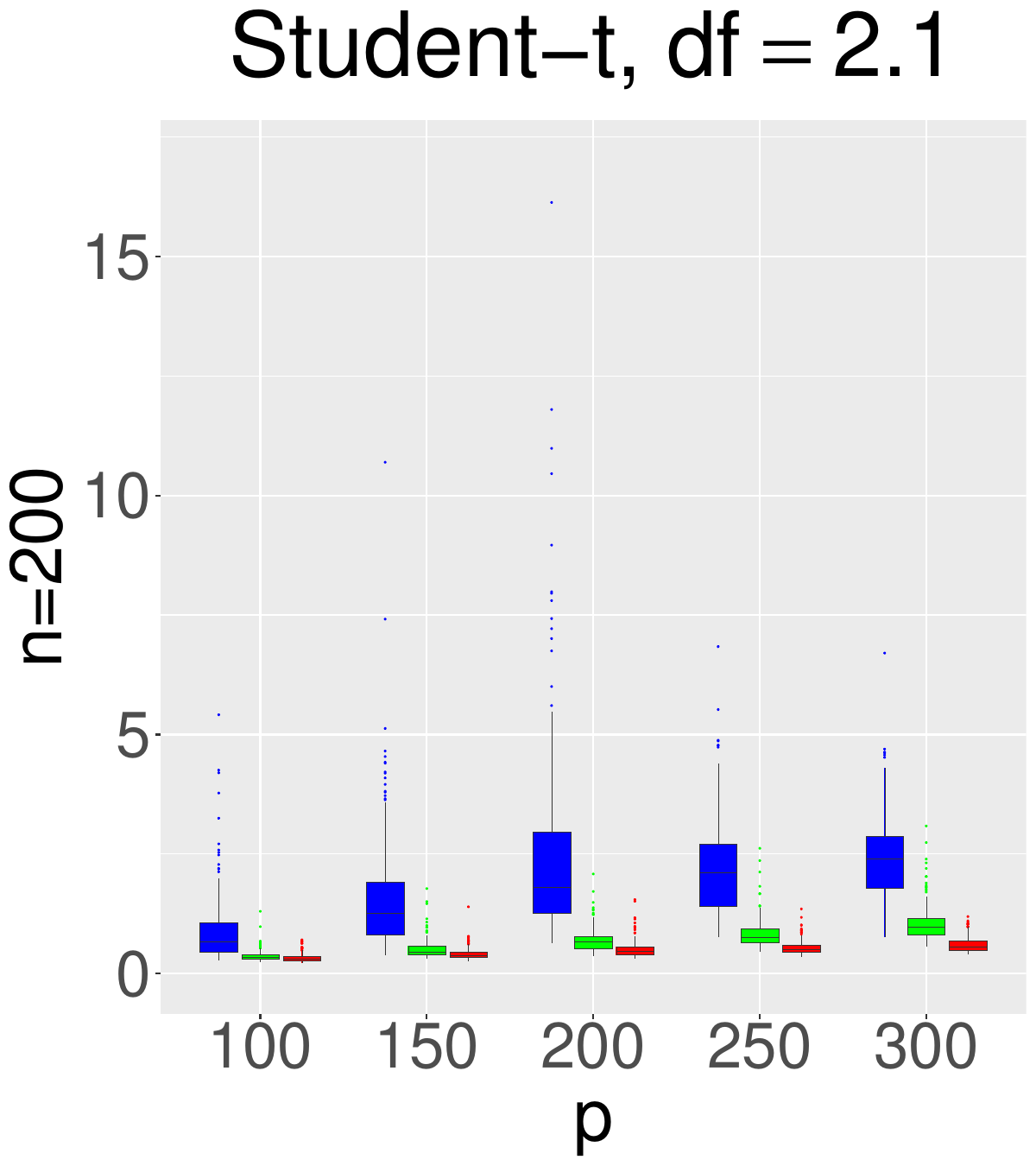}
    \end{minipage}
    \begin{minipage}[b]{0.32\textwidth}
        \includegraphics[width=\textwidth]{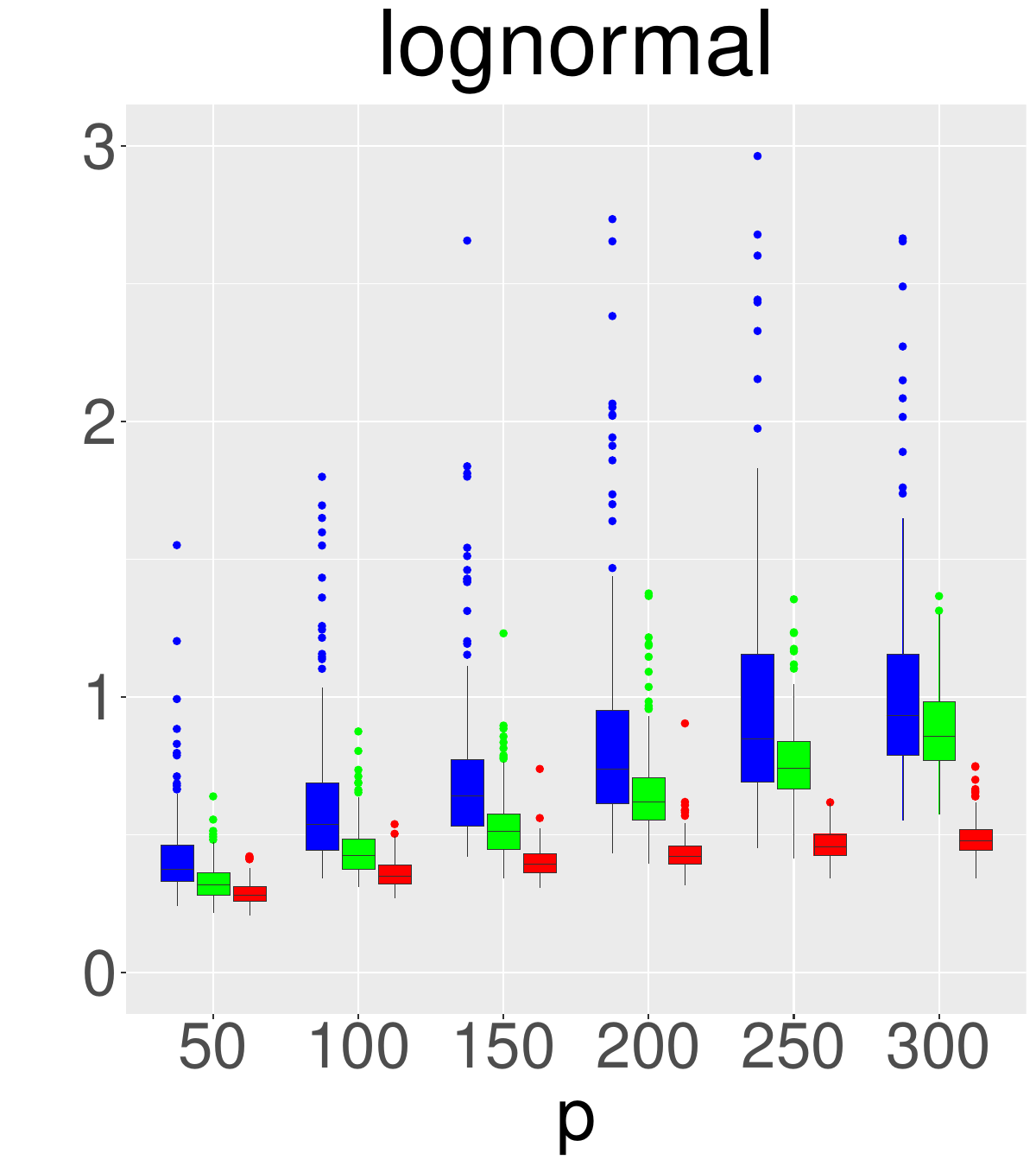}
    \end{minipage}
    \begin{minipage}[b]{0.32\textwidth}
        \includegraphics[width=\textwidth]{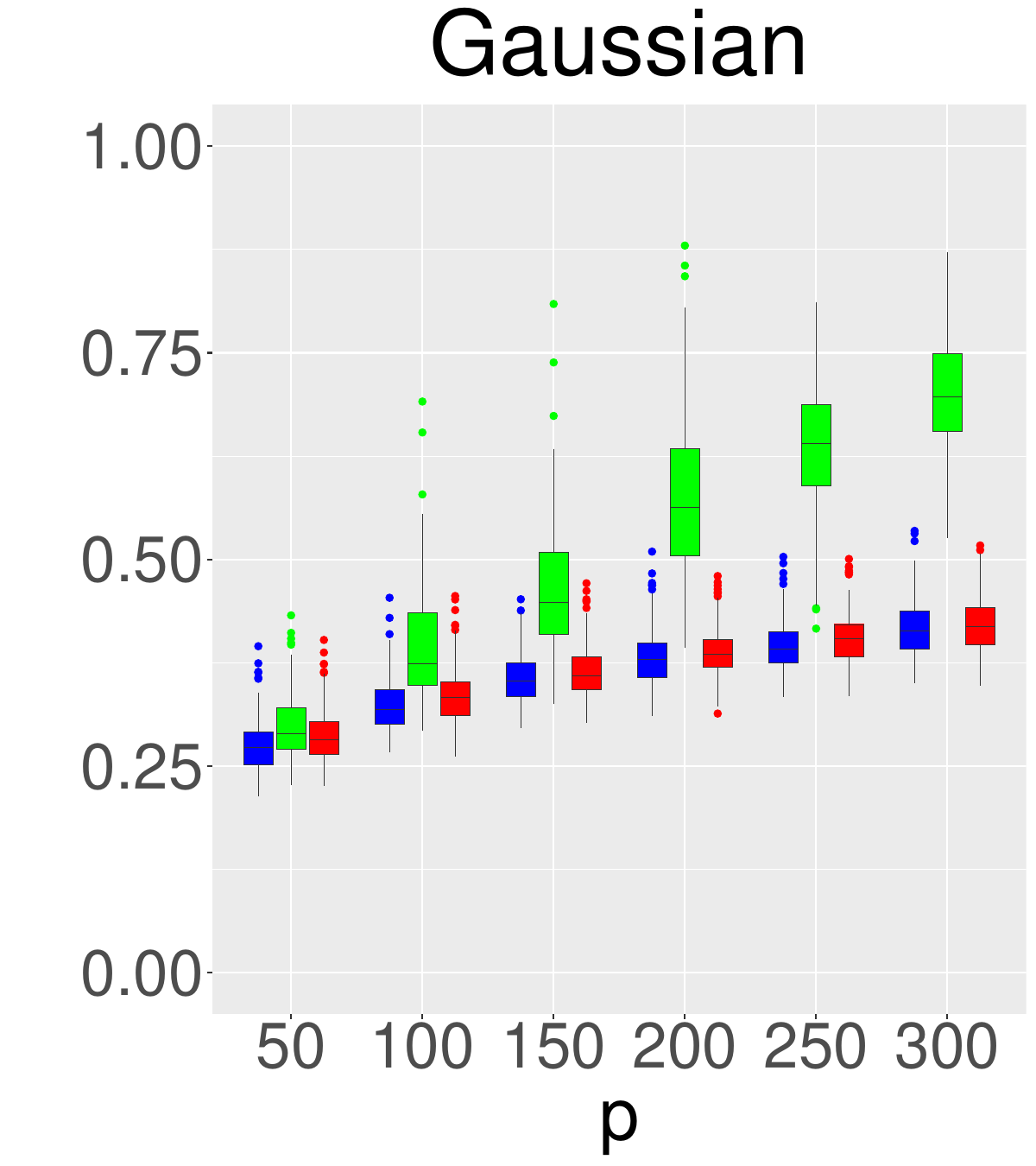}
    \end{minipage}  
    
    \begin{minipage}[t]{\textwidth}
        \centering
        \includegraphics[width=8cm]{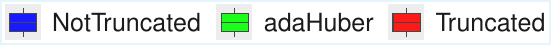} 
    \end{minipage}
    
    \begin{minipage}[b]{0.32\textwidth}
        \includegraphics[width=\textwidth]{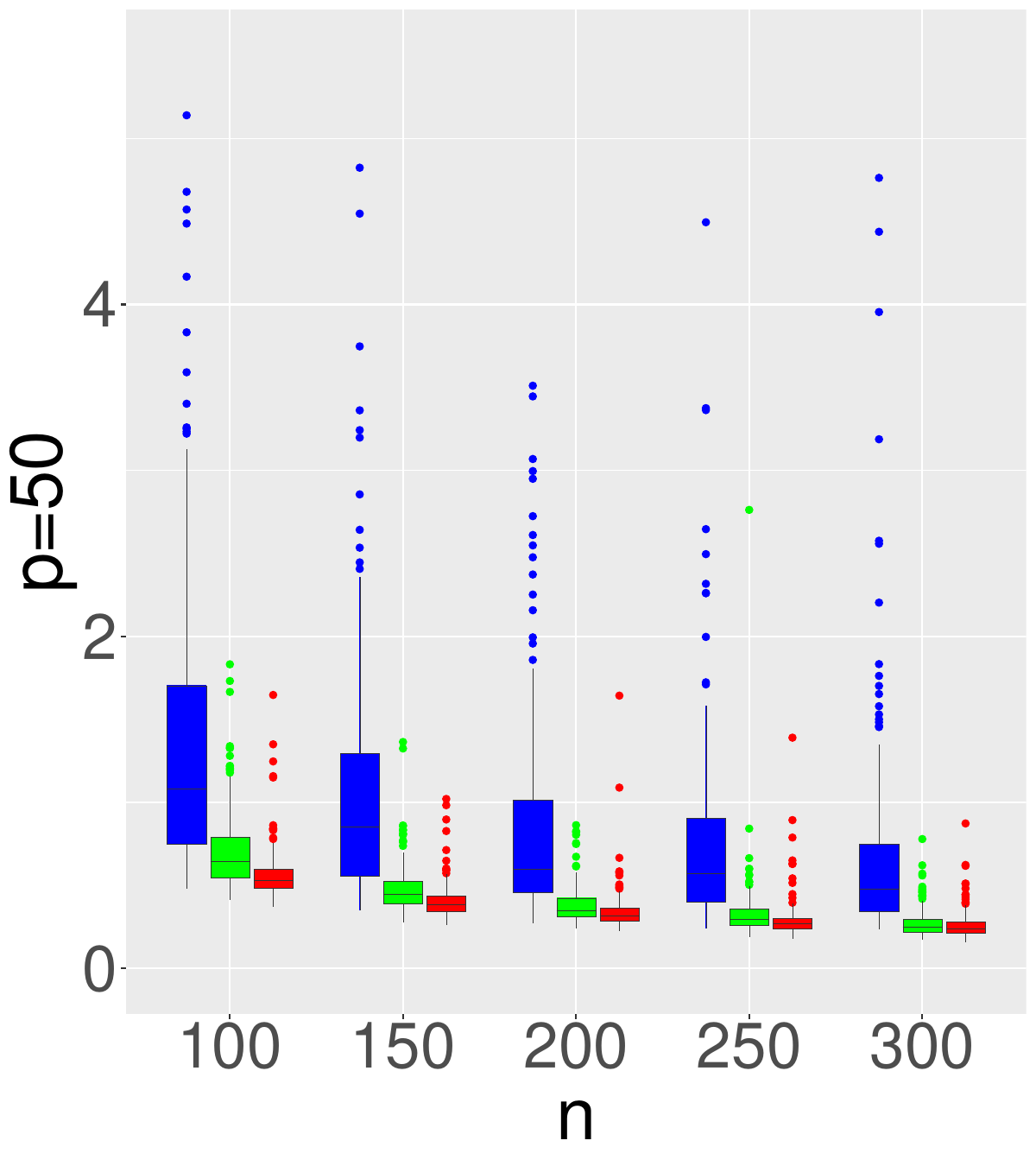}
    \end{minipage}
    \begin{minipage}[b]{0.32\textwidth}
        \includegraphics[width=\textwidth]{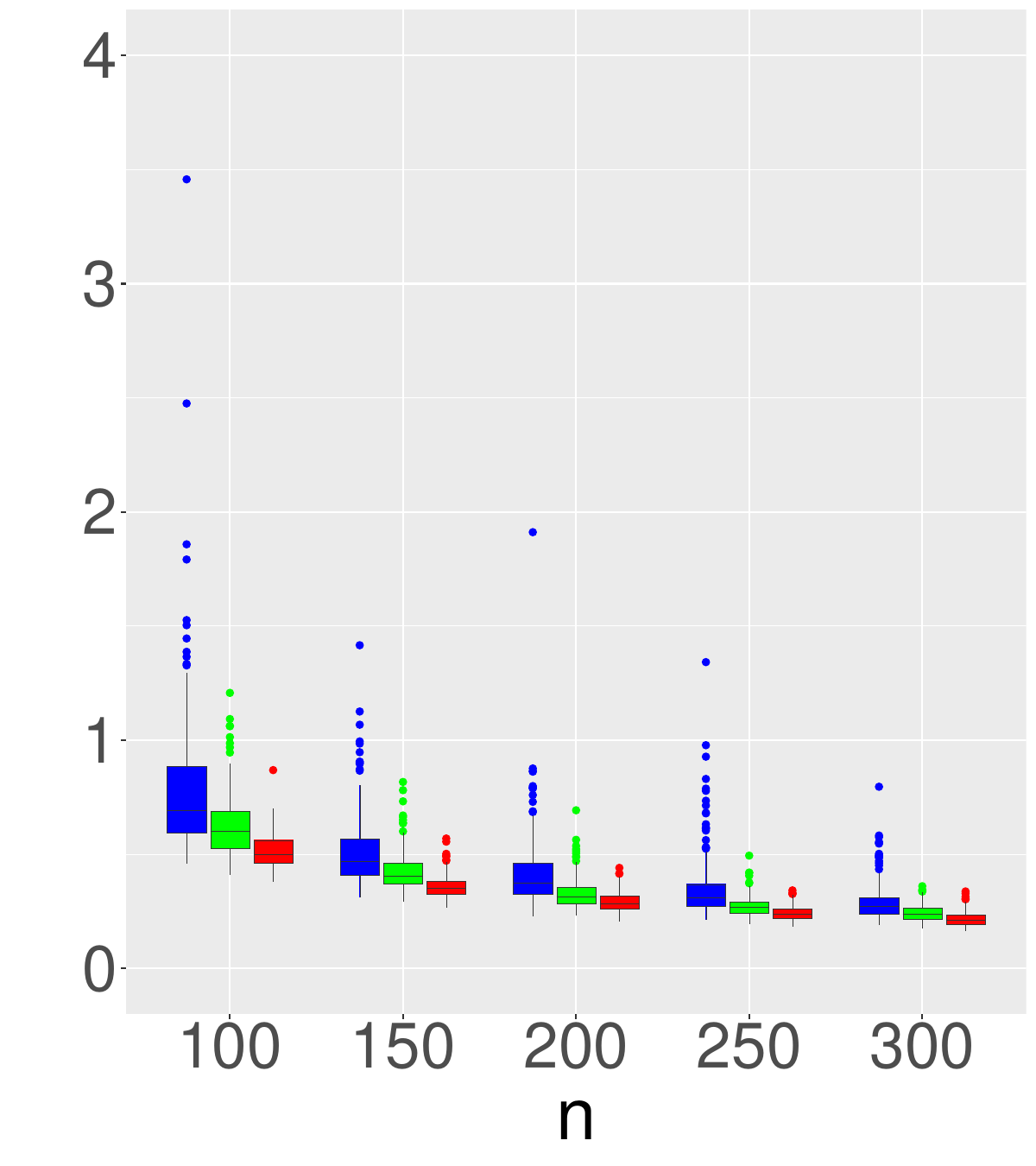}
    \end{minipage}
    \begin{minipage}[b]{0.32\textwidth}
        \includegraphics[width=\textwidth]{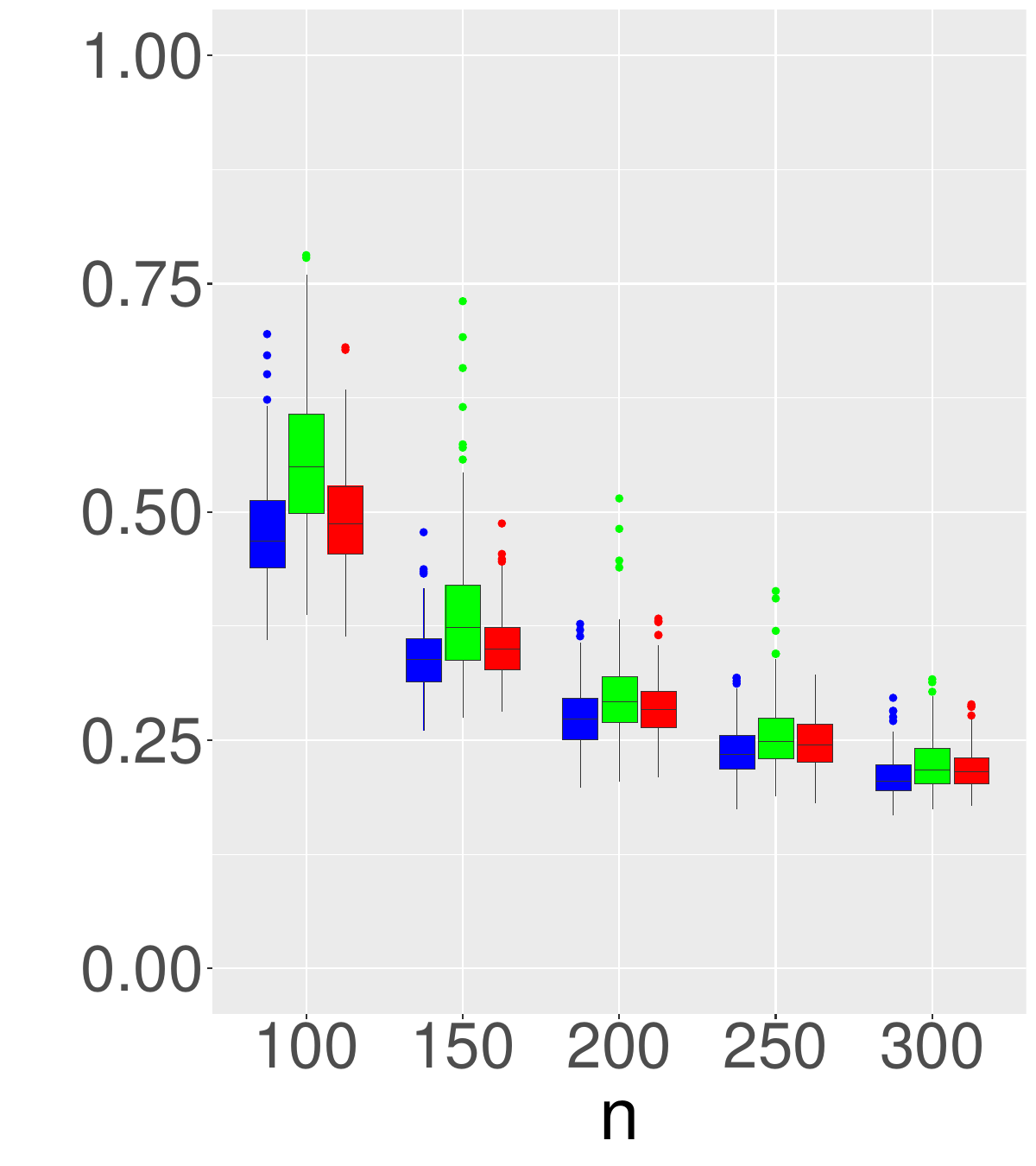}
    \end{minipage}  

    \caption{
    Box plots of estimation errors measured as in~\eqref{eq:err:measure} over $200$ realisations as $(n, p)$ and the innovation distribution vary (see~\ref{innov_dist_gauss}--\ref{innov_dist_log_norm}), when the data are generated as in \ref{item:fac_0} and \ref{eq:tri_diag_A_sparse_var}.
    }
    \label{fig:comparing_adaHuber_lasso_vs_p_boxplot_outliers}
\end{figure}

We additionally present a summary of the relative performance of the truncation estimator against that of the standard, no-truncation-based approach (with $\tau = \infty$), in \Cref{tab:auto_var_banded_lasso_diag,tab:no_stand_auto_var_renyi_lasso_diag}, measured by the relative mean errors (`RME') defined by:
\begin{equation} \label{eq:RME}
\mathrm{RME}_{\mathbb{A}}=\frac{\sum_{i=1}^{200}\left|\widehat{\mathbb{A}}^{(i)}(\tau) - \mathbb{A}\right|_{m}}{\sum_{i=1}^{200}\left|\widehat{\mathbb{A}}^{(i)} -\mathbb{A}\right|_{m}} \, ,
\end{equation}
where $\widehat{\mathbb{A}}^{(i)}(\tau)$ and $\widehat{\mathbb{A}}^{(i)}$ are the estimates from the $i$-th realisation for $1 \le i \le 200$.
The subscript $m$ denotes the matrix norm used to compute the $\mathrm{RME}_{\mathbb{A}}$ for which we consider the element-wise maximum norm (`max') and the maximum column Euclidean norm ($\ell_{2,\infty}$). 
Smaller $\mathrm{RME}_{\mathbb{A}}$ indicates a larger improvement of the proposed tail-robust method, whilst a value larger than one signifies the standard no-truncation estimator is better. 
As expected, we observe that the largest improvement is in the heaviest tailed case with the innovations generated from the \cred{Student-$t$ distribution with $2.1$ degrees of freedom.}

\begin{table}[h!t!] 
    \centering
    \caption{VAR coefficient matrix estimation errors measured by $\mathrm{RME}_{\mathbb{A}}$ in~\eqref{eq:RME} with different matrix norms from 200 realisations, when the data are generated as in \ref{item:fac_0} and \ref{eq:tri_diag_A_sparse_var}, as $(n, p)$ and the innovation distribution vary (see \ref{innov_dist_gauss}--\ref{innov_dist_log_norm}).
    }
    \begin{tabular}{c cc cc cc cc cc}
        \toprule & 
                 \multicolumn{2}{c}{log-normal} &
                 \multicolumn{2}{c}{$t_{2.1}$} & \multicolumn{2}{c}{$t_{3}$} & \multicolumn{2}{c}{$t_{4}$} & 
                 \multicolumn{2}{c}{Normal}\\
        \cmidrule(l{2pt}r{2pt}){2-3}
        \cmidrule(l{2pt}r{2pt}){4-5}
        \cmidrule(l{2pt}r{2pt}){6-7}
        \cmidrule(l{2pt}r{2pt}){8-9}
        \cmidrule(l{2pt}r{2pt}){10-11} $(n,p)$
        &  max   & $\ell_{2,\infty}$ & max & $\ell_{2,\infty}$ & max & $\ell_{2,\infty}$ & max & $\ell_{2,\infty}$ & max & $\ell_{2,\infty}$\\
        \midrule 
(100,50) & 0.738 & 0.653 & 0.465 & 0.408 & 0.809 & 0.691 & 0.944 & 0.881 & 1.017 & 1.032 \\ 
  (100,100) & 0.669 & 0.580 & 0.348 & 0.323 & 0.701 & 0.548 & 0.940 & 0.839 & 1.002 & 1.014 \\ 
  (200,50) & 0.777 & 0.711 & 0.462 & 0.388 & 0.783 & 0.665 & 0.961 & 0.921 & 1.028 & 1.040 \\ 
  (200,100) & 0.695 & 0.579 & 0.304 & 0.259 & 0.693 & 0.531 & 0.948 & 0.866 & 1.019 & 1.024 \\ 
  (500,100) & 0.774 & 0.657 & 0.292 & 0.236 & 0.776 & 0.616 & 0.951 & 0.880 & 1.026 & 1.029 \\ 
  (500,200) & 0.670 & 0.540 & 0.200 & 0.160 & 0.573 & 0.368 & 0.935 & 0.811 & 1.026 & 1.020 \\
        \bottomrule
    \end{tabular}
    \label{tab:auto_var_banded_lasso_diag}
\end{table}

\begin{table}[h!t!] 
    \centering
    \caption{VAR coefficient matrix estimation errors measured by $\mathrm{RME}_{\mathbb{A}}$ in~\eqref{eq:RME} with different matrix norms from 200 realisations, when the data are generated as in \ref{item:fac_0} and \ref{eq:renyi_A}, as $(n, p)$ and the innovation distribution vary (see~\ref{innov_dist_gauss}--\ref{innov_dist_log_norm}).}
    \begin{tabular}{c cc cc cc cc cc}
        \toprule & 
                 \multicolumn{2}{c}{log-normal} &
                 \multicolumn{2}{c}{$t_{2.1}$} & \multicolumn{2}{c}{$t_{3}$} & \multicolumn{2}{c}{$t_{4}$} & 
                 \multicolumn{2}{c}{Normal}\\
        \cmidrule(l{2pt}r{2pt}){2-3}
        \cmidrule(l{2pt}r{2pt}){4-5}
        \cmidrule(l{2pt}r{2pt}){6-7}
        \cmidrule(l{2pt}r{2pt}){8-9}
        \cmidrule(l{2pt}r{2pt}){10-11} $(n,p)$
        &  max   & $\ell_{2,\infty}$ & max & $\ell_{2,\infty}$ & max & $\ell_{2,\infty}$ & max & $\ell_{2,\infty}$ & max & $\ell_{2,\infty}$\\
        \midrule 
(100,50) & 0.650 & 0.648 & 0.563 & 0.539 & 0.765 & 0.702 & 0.942 & 0.884 & 1.007 & 1.035 \\ 
  (100,100) & 0.554 & 0.624 & 0.689 & 0.647 & 0.655 & 0.628 & 0.921 & 0.872 & 1.000 & 1.035 \\ 
  (200,50) & 0.716 & 0.685 & 0.418 & 0.386 & 0.759 & 0.691 & 0.928 & 0.890 & 1.026 & 1.049 \\ 
  (200,100) & 0.617 & 0.582 & 0.405 & 0.375 & 0.704 & 0.619 & 0.924 & 0.864 & 1.008 & 1.042 \\ 
  (500,100) & 0.595 & 0.576 & 0.326 & 0.281 & 0.665 & 0.597 & 0.936 & 0.904 & 1.017 & 1.058 \\ 
  (500,200) & 0.602 & 0.566 & 0.304 & 0.255 & 0.495 & 0.395 & 0.925 & 0.846 & 1.010 & 1.049 \\ 
        \bottomrule
    \end{tabular}
\label{tab:no_stand_auto_var_renyi_lasso_diag}
\end{table}

\section{Real data application} \label{sec:real-data}


We use the proposed tail-robust estimation method for the factor-adjusted VAR model, 
to forecast a large set of major US macroeconomic variables taken from the Federal Reserve Economic Data monthly database (FRED-MD).  
FRED-MD is maintained by the St.\ Louis Federal Reserve Bank, with monthly updates of the data released.
\cred{Following \citet{McCracken2015} we perform variable-specific transformations for stationarity, which involve differencing or logarithmic transforms and combinations therefore; we refer to the paper for their detail.} 
We consider the variables with no missing values which results in the dataset containing the observations from 108 variables ($p = 108$) between 1960-02 and 2023-11 ($n = 767$).
FRED-MD has popularly been analysed in the factor modelling literature as a benchmark dataset, and the presence of a handful of factors driving the pervasive cross-sectional correlations is evident from \Cref{fig:eigen_fredmd}.
\cred{In some prior analyses, extreme observations were removed manually \citep{Kock2024}, which motivates the use of our proposed tail-robust method which handles the extreme observations due to heavy tails in a data-driven manner, without any manual removal.} 
We select the number of factors $r$ using the three information criteria of \citet{Bai2002}, all of which return $7$ as the factor number. 

We perform a rolling window-based forecasting exercise with the window size $T = 120$ (corresponding to $10$ years of data).
At each $t \ge T$, using the preceding $T$ observations, we standardise each of the series with the median absolute deviation and then estimate the model in~\eqref{eq:fac_adj_model_def} as described in Section~\ref{sec:method} where, for the VAR order, we consider $d \in \{1, \ldots, 5\}$; we denote all such estimators with the subscripts $t$ and $T$.
Then, we generate the $h$-step ahead forecast of $\bX_{t + h}$ following \citet{stock2002forecasting}, see also \citet{Barigozzi2024}. 
Specifically, for the common component, we estimate the best linear predictor of $\bchi_{t + h}$ given the preceding~$T$ observations, as 
\begin{align*}
    \widehat{\boldsymbol{\bchi}}_{t+h \mid T}(\tau) = \widehat{\bGamma}_{\bchi,t\mid T}(\tau,h) \widehat{\mathbf{E}}_{t\mid T}(\tau) (\widehat{\mathbf{M}}_{t\mid T}(\tau))^{-1} (\widehat{\mathbf{E}}_{t\mid T}(\tau))^{\top} \bX_t \, ,
\end{align*}
where $\widehat{\bGamma}_{\bchi,t\mid T}(\tau,h) = T^{-1} \sum_{u = t - T + h + 1}^t \widehat{\bchi}_{u, t\mid T}(\tau)   \widehat{\bchi}_{u-h, t\mid T}(\tau)^\top$ denotes an estimator of the autocovariance of $\bchi_t$ at lag $h$, the matrix $\widehat{\mathbf{M}}_{t\mid T}(\tau)$ contains the leading $r$ eigenvalues of $\widehat{\bGamma}_{\bx,t\mid T}(\tau, 0)$ on its diagonal, $\widehat{\mathbf{E}}_{t\mid T}(\tau)$ contains the corresponding eigenvectors, and
$\widehat{\bchi}_{u, t\mid T}(\tau) = \widehat{\mathbf{E}}_{t\mid T}(\tau) ( \widehat{\mathbf{E}}_{t\mid T}(\tau) )^\top \bX_u(\tau)$ for $u \in \{t - T + 1, \ldots, t\}$. 
The forecast of the idiosyncratic component is computed as \begin{align} 
\widehat{\boldsymbol{\xi}}_{t+h \mid T}(\tau)=\sum_{\ell=1}^{d} \widehat{\mathbf{A}}_{\ell, t\mid T}(\tau) \widehat{\boldsymbol{\xi}}_{t+h-\ell, t \mid T}(\tau), 
\nonumber
\end{align}
where $\widehat{\boldsymbol{\xi}}_{t+h-\ell, t \mid T}(\tau) = \bX_{t + h - \ell} - \widehat{\bchi}_{t + h - \ell, t\mid T}(\tau)$ for $\ell \ge h$, and $\widehat{\boldsymbol{\xi}}_{t+h-\ell, t \mid T}(\tau) = \widehat{\boldsymbol{\xi}}_{t+h - \ell \mid T}(\tau)$ otherwise.
Then, the combined forecast is given by
\begin{gather*}
 \widehat{\bX}_{t+h \mid T}(\tau) = \widehat{\boldsymbol{\bchi}}_{t+h \mid T}(\tau) + \widehat{\boldsymbol{\xi}}_{t+h \mid T}(\tau) \, .
\end{gather*}
We carry out the exercise with the forecasting horizon $h = 1$, and compute the absolute error of the forecasts for each individual series. 
For comparison, we consider the standard approach without any truncation (i.e.\ $\tau = \infty$). Denoting the forecast errors for the $i$-th series from the two approaches, by
\begin{gather*}
    \text{FE}_{it}(\tau) = \left|\widehat{X}_{i,t+1 \mid T}(\tau) - X_{i,t+1}\right| \text{ \ and \ } \text{FE}_{it}(\infty) = \left|\widehat{X}_{i,t+1 \mid T}(\infty) - X_{i,t+1}\right| \, ,
\end{gather*}
respectively,
we investigate whether they perform significantly differently at any point by means of the fluctuation test of \citet{Giacomini2010} (implemented in the R package \texttt{murphydiagram}, \citeauthor{murphydiagram}, \citeyear{murphydiagram}). This test addresses the fact that one forecasting method may outperform another only on certain intervals, which global performance measures can miss. Formally, it is a two-sided test with the null hypothesis that the expected losses of the competing forecasts are equal on all of the windows. A series of test statistics is computed as a scaled rolling average of the process $\text{FE}_{it}(\tau) - \text{FE}_{it}(\infty)$ over the window of length $\lfloor \mu (n - T) \rfloor$ for some $\mu \in (0, 1)$. 
Following the authors' recommendation of using $\mu \ge 0.2$, we set $\mu = 0.3$. 

We observe that irrespective of the VAR order number $d$, significant improvements are observed for a large number of variables when the data are suitably truncated; 
for simplicity, we present the results for $d = 1$. 
Then, for the 54 out of the 108 variables, the fluctuation test favours the forecasting performance from the proposed tail-robust approach, while for three variables, the opposite is the case; for the remaining variables, no significant difference is observed.

\begin{figure}[h!t!]
    \centering
    \begin{minipage}[b]{0.49\textwidth}
    \includegraphics[width=\textwidth]{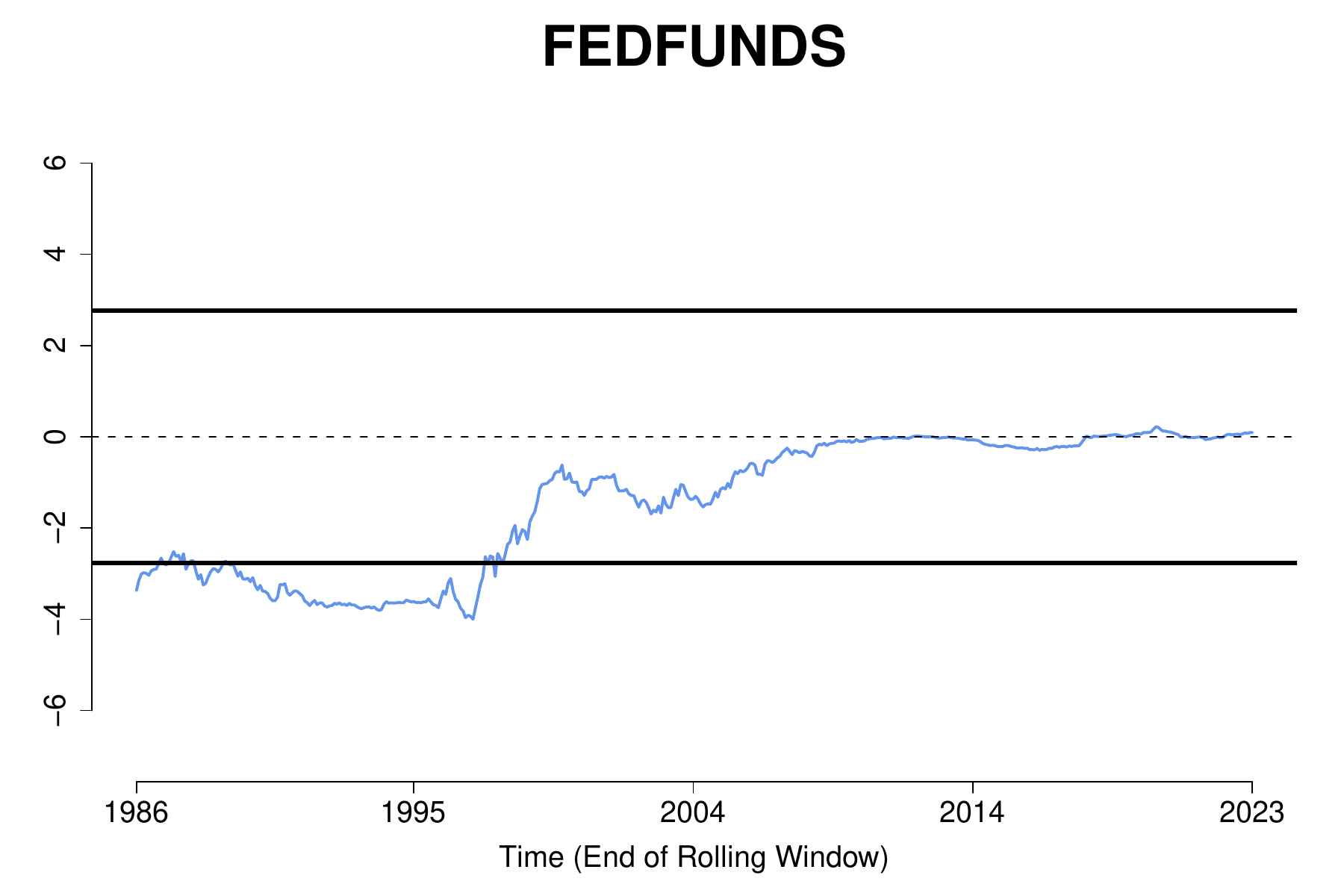}
    \end{minipage}
  \begin{minipage}[b]{0.49\textwidth}
  
    \includegraphics[width=\textwidth]{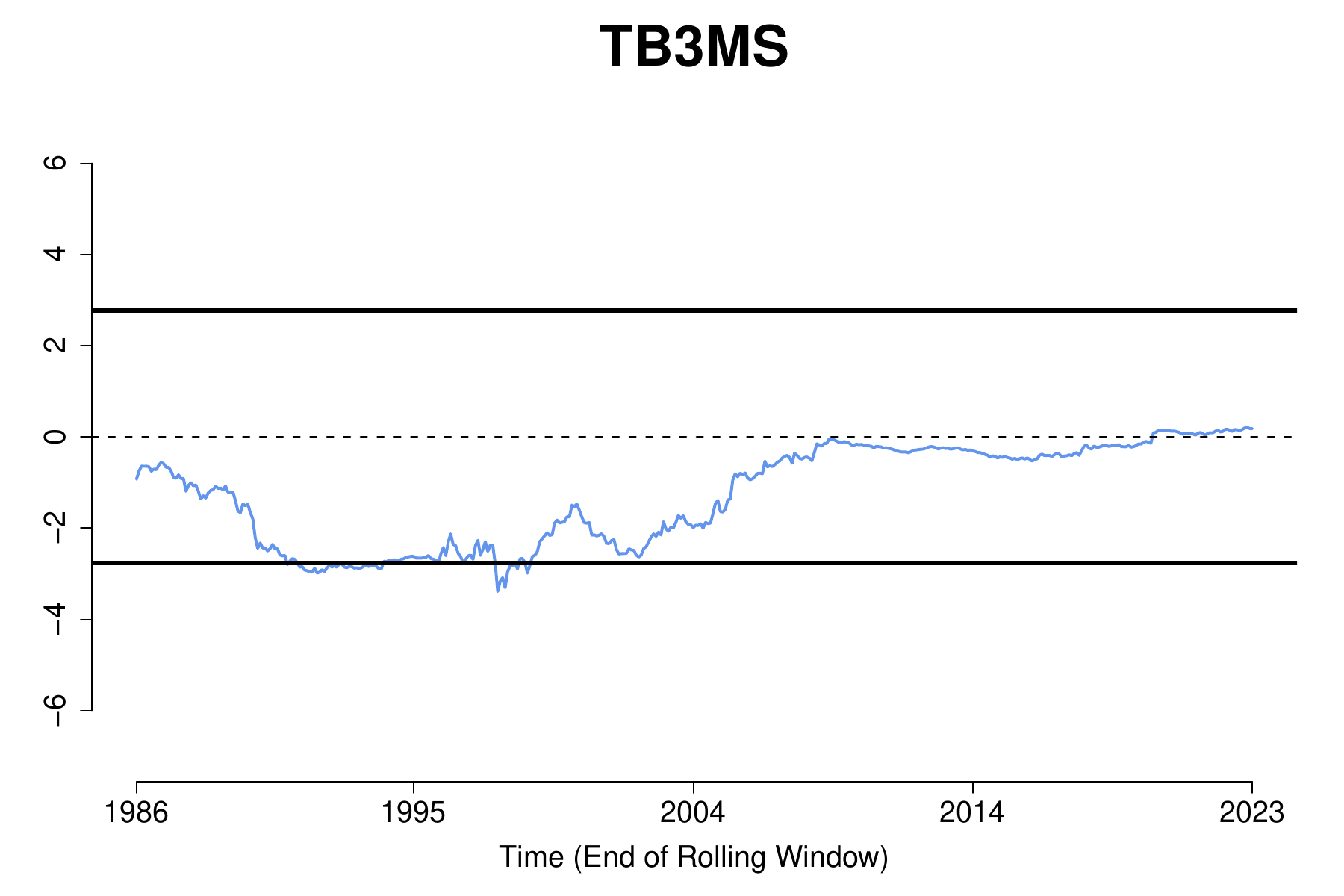}
    \end{minipage}

    \centering
    \begin{minipage}[b]{0.49\textwidth}
    \includegraphics[width=\textwidth]{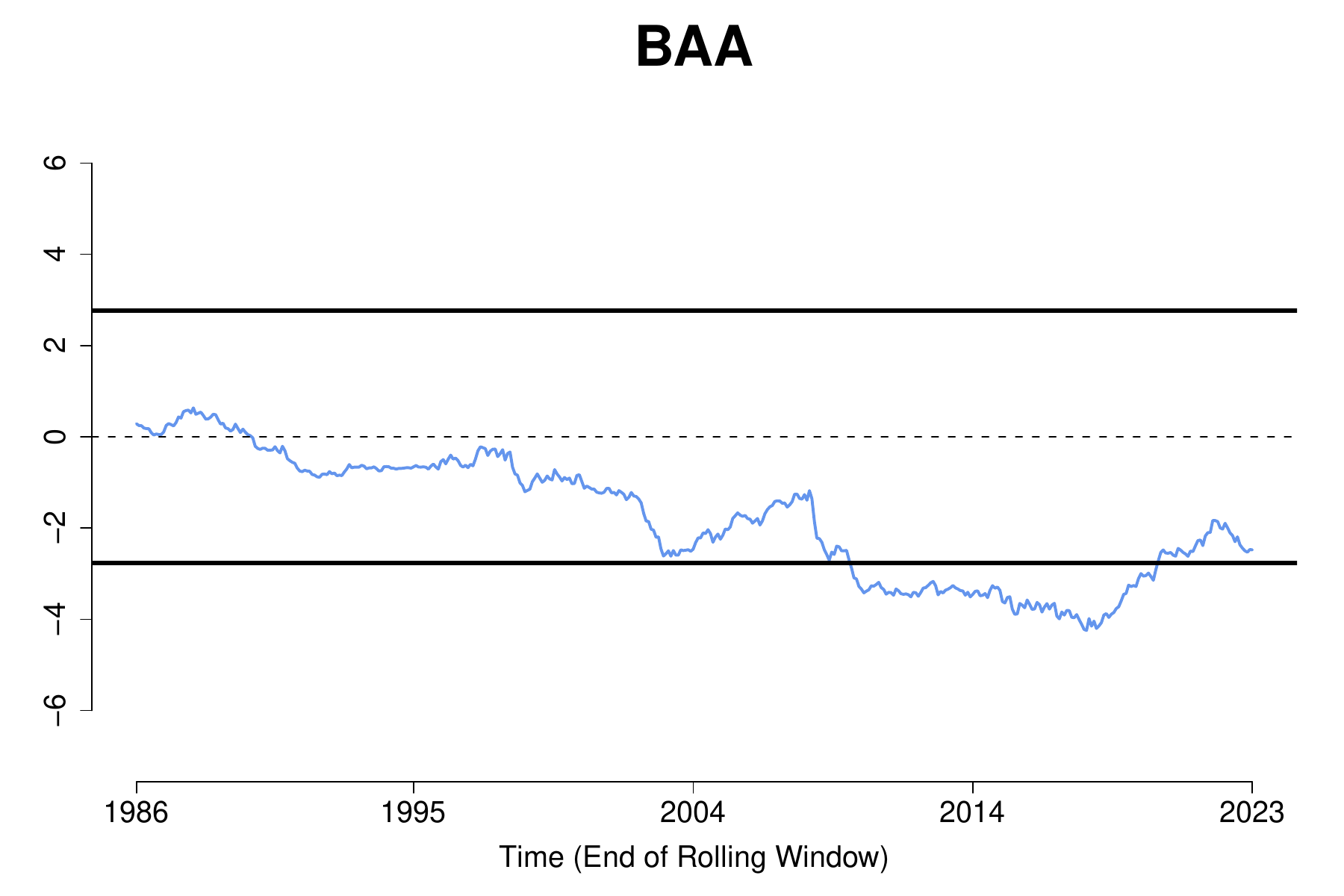}
    \end{minipage}
  \begin{minipage}[b]{0.49\textwidth}
  
    \includegraphics[width=\textwidth]{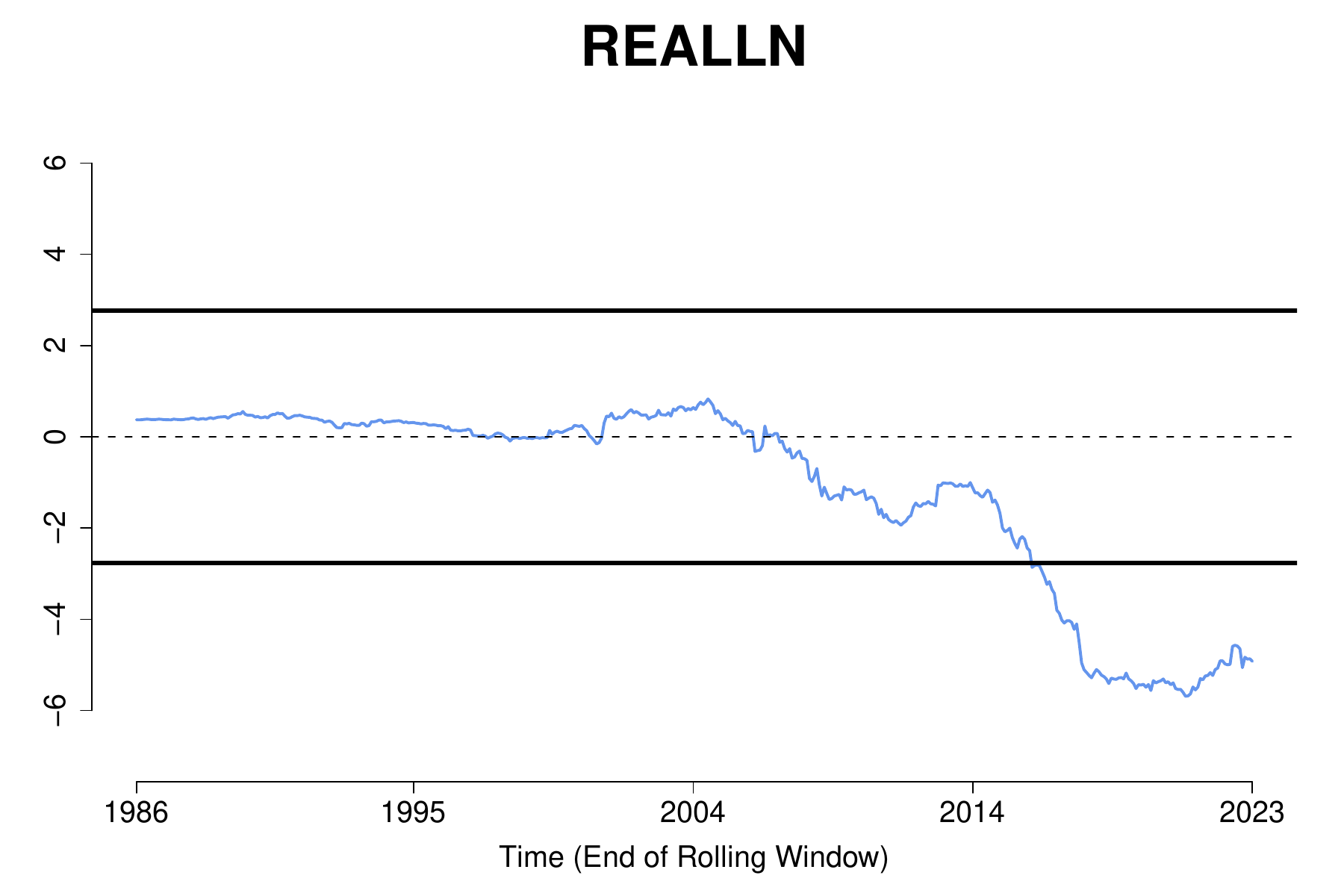}
    \end{minipage}

    \centering
    \begin{minipage}[b]{0.49\textwidth}
    
    \includegraphics[width=\textwidth]{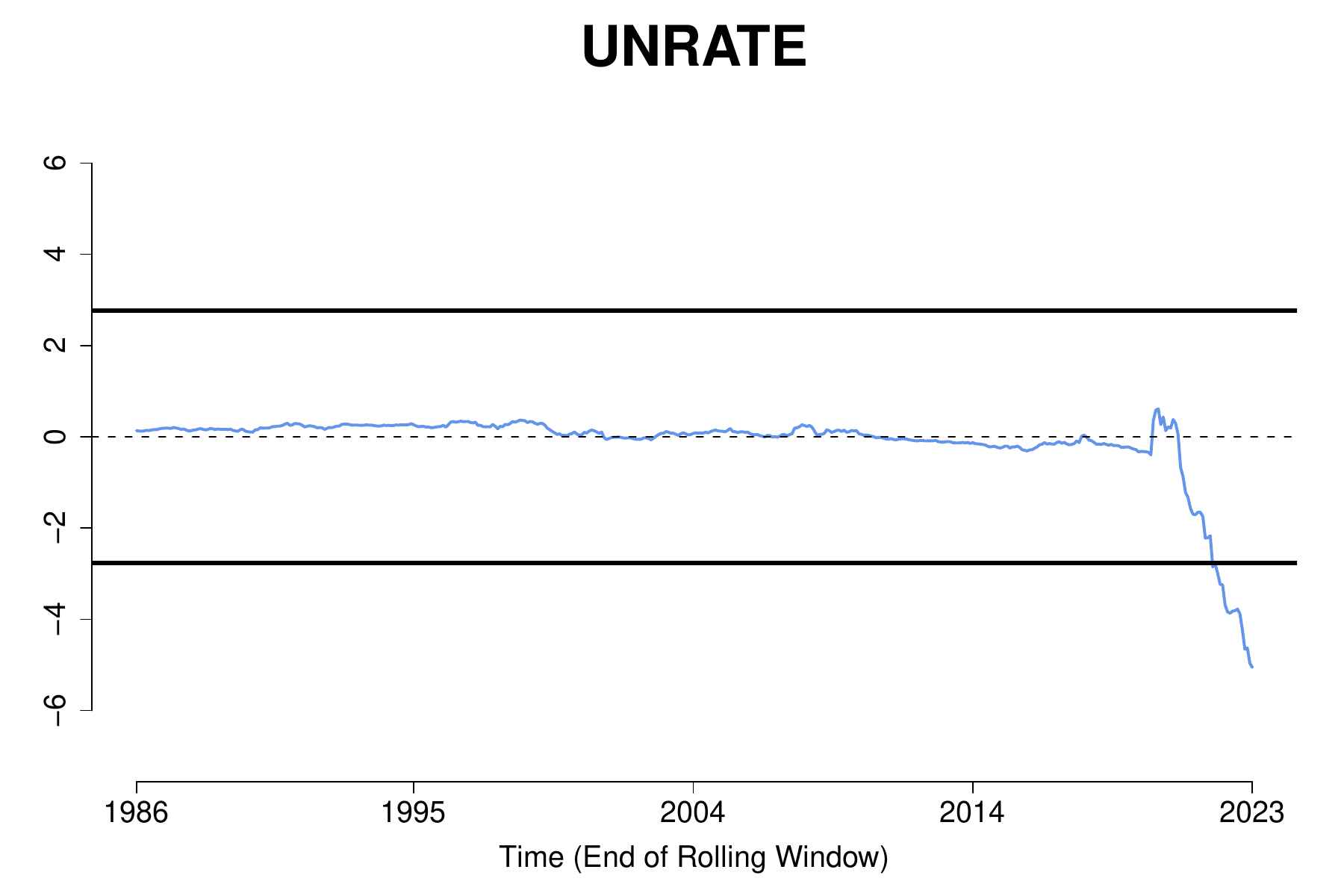}
    \end{minipage}
  \begin{minipage}[b]{0.49\textwidth}
  
    \includegraphics[width=\textwidth]{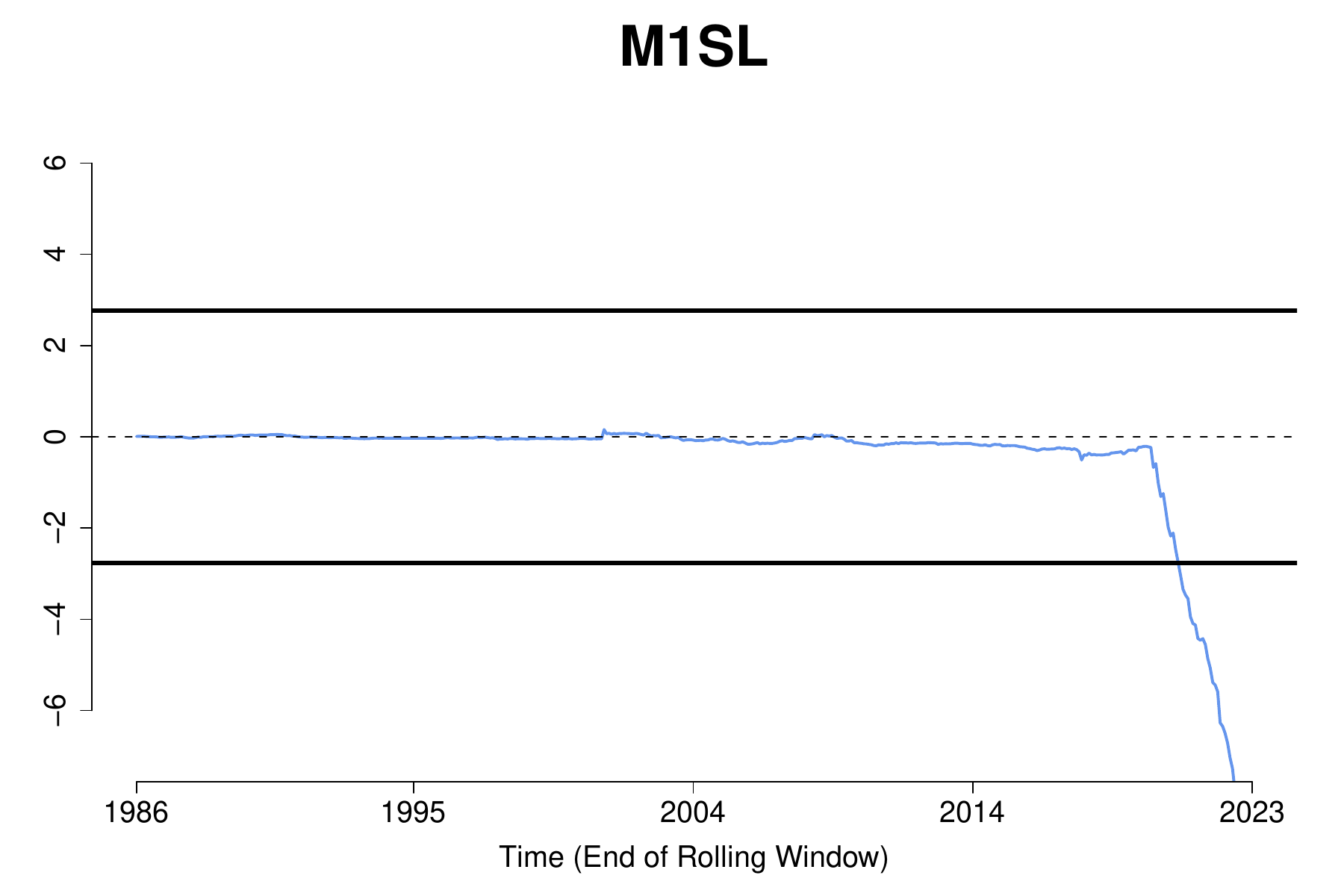}
    \end{minipage}

        \caption{Time path of the fluctuation test statistics for six variables, which are scaled rolling averages of $\text{FE}_{it}(\tau) - \text{FE}_{it}(\infty)$; we set the parameter $\mu = 0.3$ which determines the window. The large negative values correspond to when the proposed truncation-based method outperforms the standard method without any truncation. The horizontal lines are the critical values. } 
    \label{fig:FRED_fluc_test}
\end{figure}

In \Cref{fig:FRED_fluc_test}, we visualise the outputs from the fluctuation test for a selection of six variables. These six variables are chosen as they display the differing relative forecasting performance between the two approaches over three distinct periods. 
The first row of \Cref{fig:FRED_fluc_test} shows the results from the federal funds rate (FEDFUNDS) and the yield for a 3-month treasury bill (TB3MS), both of which are the measures of short-term interest rates. The two plots show that the tail-robust estimator is to be favoured in the earlier period of the dataset, where the Fedfunds reached a historic high in the 1980s to combat inflation \citep{FRH2013}. The second row presents the results for Moody's Seasoned Baa Corporate Bond Yield (BAA), which is the average yield of lower investment-grade bonds, and the REALLN variable which measures the amount of loans that commercial banks are lending for real estate purposes. These variables changed drastically during the 2007--09 financial crisis, coinciding with the periods over which the truncation-based forecasts are favoured against the non-robust counterparts. The final row of plots considers the unemployment rate (UNRATE) and the M1 money supply (M1SL), both of which had sharp increases in response to the COVID-19 pandemic.
\cred{We notice that the gains from data truncation are most pronounced when the periods of economic turmoil are included in the estimation window, rather than at the exact time of the crisis. 
This is as expected since our method lessens the impact of extreme observations included in the training data for estimating the forecasting model, rather than aiming at forecasting the extreme behaviour of the data during the extreme events.}

\section{Conclusions}

We propose a tail-robust estimation method for factor-adjusted VAR modelling of high-dimensional time series, which simultaneously accounts for pervasive correlations in the data as well as idiosyncratic interconnectivity between the variables. 
We utilise data truncation for achieving tail-robustness, and show its efficacy theoretically with the rates of estimation explicitly dependent on the moment condition. 
\cred{
Under the factor-adjusted VAR model, the quantities of interest including the VAR parameters, are characterised by the second moments of the data, whose tail-robust estimation forms the backbone of our proposed method including forecasting.
It is envisioned that beyond this linear modelling framework, understanding the informative role played by extreme observations e.g.\ in predicting extreme events, would be highly relevant and interesting, which we leave for future research.}

\printbibliography

\newpage

\appendix

\numberwithin{equation}{section}
\numberwithin{figure}{section} 
\numberwithin{table}{section} 
\numberwithin{theorem}{section}
\numberwithin{proposition}{section}
\numberwithin{lemma}{section}

\section{Additional simulation results} \label{sec:additional_simul}

\subsection{\cred{Robustness to VAR order mis-specification}} \label{sec:order_missspec}

In this subsection, we investigate the robustness of the proposed estimator to mis-specification of the VAR order. Estimating the VAR order in high-dimensional settings is a challenging problem, and even in light-tailed scenarios there is limited theoretical guidance \citep{Owens2023}. Recall that we denote the true VAR order by $d$, and in this subsection we now additionally denote the order selected for model fitting by $d_f$.

We envision that under-specification of the true VAR order $d$ at $d_f < d$, where $d_f$ is the order selected for model fitting, will encourage the estimation method to find the best order-$d_f$ approximation of the VAR process, while its over-specification at $d_f > d$, will result in the elements of $\widehat{\mathbf{A}}_\ell, \, d + 1 \le \ell \le d_f$, being close to zero. In what follows, we focus on the over-specification case, which is more amenable to systematic evaluation in high dimensions.

To assess this, we generate data from the VAR process of order $d$:
\begin{gather*}
    \bX_t = \mathbf{A}_1\bX_{t-1} + \ldots + \mathbf{A}_d\bX_{t-d} + \boldsymbol{\varepsilon}_t \, .
\end{gather*}

with $\mathbf{A}_\ell, \, 1 \le \ell \le d$, generated as follows.

\begin{enumerate}[label = (V\arabic*)]
\item[(V3)] \label{eq:lagged_banded_A}
For $\ell = 1,\ldots, d$, we generate $\widetilde{\mathbf{A}}_\ell  = [\widetilde{A}_{\ell, ij}]_{i, j = 1}^p$, as
\[
\widetilde{A}_{\ell,ij} =
0.9^{\ell-1}\left(
0.5 \cdot \mathbb{I}_{\{i = j\}}
+ \text{\upshape sign}(i-j) \cdot 0.4 \cdot \mathbb{I}_{\{|i-j|=1\}}
\right),
\]
and set $\mathbf{A}_\ell = 0.99\, c^\star  \widetilde{\mathbf{A}}_\ell$,
where $c^\star > 0$ is the unique constant satisfying
\[
\Lambda_{\max}\!\left(
\widetilde{\mathcal A}(c^\star \widetilde{\mathbf A}_1,\ldots,c^\star \widetilde{\mathbf A}_d)
\right) = 1 \, ,
\]
with
\[
\widetilde{\mathcal A}(\widetilde{\mathbf A}_1,\ldots,\widetilde{\mathbf A}_d)
=
\begin{pmatrix}
\widetilde{\mathbf A}_1 & \widetilde{\mathbf A}_2 & \cdots & \widetilde{\mathbf A}_{d-1} & \widetilde{\mathbf A}_d \\
\mathbf I_p & \mathbf 0 & \cdots & \mathbf 0 & \mathbf 0 \\
\mathbf 0 & \mathbf I_p &  & \mathbf 0 & \mathbf 0 \\
\vdots & & \ddots & \vdots & \vdots\\
\mathbf 0 & \mathbf 0 & \cdots & \mathbf I_p & \mathbf 0
\end{pmatrix}
\in \mathbb{R}^{pd \times pd}.
\]
The scaling used in $\mathbf{A}_\ell = 0.99\, c^\star  \widetilde{\mathbf{A}}_\ell$  ensures that the resultant VAR process is stable. 
\end{enumerate}

We evaluate the relative performance of our estimator (with truncation) against that of the `standard' approach without any data truncation (i.e.\ $\tau = \infty$), using the relative mean errors (RME), consistent with the definition in \eqref{eq:RME}.

\begin{equation} \label{eq:App_RME}
\mathrm{RME}_{\mathbb{A}}=\frac{\sum_{i=1}^{200}\left|\widehat{\mathbb{A}}^{(i)}(\tau) - \mathbb{A}\right|_{m}}{\sum_{i=1}^{200}\left|\widehat{\mathbb{A}}^{(i)} -\mathbb{A}\right|_{m}} \, ,
\end{equation}

where $m$ denotes the matrix norm used to compute $\mathrm{RME}_{\mathbb{A}}$; we employ the element-wise maximum norm (`max') and the maximum column Euclidean norm ($\ell_{2, \infty}$).
We denote by $\widehat{\mathbb{A}}(\tau) = [\widehat{\mathbf{A}}_1(\tau), \ldots, \widehat{\mathbf{A}}_{d_f}(\tau)] \in \mathbb{R}^{p \times pd_f}$ the matrix of estimated VAR coefficients, and accounting for that $d_f > d$, we set $\mathbb{A} = [\mathbf{A}_1, \ldots, \mathbf{A}_d, \textbf{O}_{p \times p(d_f - d)}] \in \mathbb{R}^{p \times pd_f}$. 
The results obtained with $(d, d_f) = (1, 2)$ and $(d, d_f) = (2, 3)$ are reported in \Cref{tab:auto_var_banded_lasso_diag_d0_1_d_2} and \Cref{tab:auto_var_banded_lasso_diag_d0_2_d_3}, respectively.

They show that when the VAR order is over-specified, our method still outperforms the standard method, providing evidence that the method maintains tail-robustness even when $d$ is over-specified. In addition, \Cref{fig:heatmap_estimate} visualises the average estimated coefficient matrix for our method when $(d, d_f) = (2, 3)$, 
which shows that the coefficients corresponding to the third lag are estimated close to zero on average, with only minor deviations along the diagonal and off-diagonal entries corresponding to where $\mathbf{A}_1$ and $\mathbf{A}_2$ are not zero.

\begin{table}[H] 
    \centering
    \caption{VAR coefficient matrix estimation errors measured by $\mathrm{RME}_{\mathbb{A}}$ in~\eqref{eq:App_RME} with different matrix norms from $200$ realisations, when the data are generated as in \ref{item:fac_0} and \hyperref[eq:lagged_banded_A]{(V3)}, as $n$, $p$ and the innovation distribution vary (see \ref{innov_dist_gauss}--\ref{innov_dist_log_norm}). The VAR coefficient matrices are estimated with VAR order set to $d_f = 2$ while the true VAR order is $d = 1$.}
    \begin{tabular}{c cc cc cc cc cc}
        \toprule & 
                 \multicolumn{2}{c}{log-normal} &
                 \multicolumn{2}{c}{$t_{2.1}$} & \multicolumn{2}{c}{$t_{3}$} & \multicolumn{2}{c}{$t_{4}$} & 
                 \multicolumn{2}{c}{Normal}\\
        \cmidrule(l{2pt}r{2pt}){2-3}
        \cmidrule(l{2pt}r{2pt}){4-5}
        \cmidrule(l{2pt}r{2pt}){6-7}
        \cmidrule(l{2pt}r{2pt}){8-9}
        \cmidrule(l{2pt}r{2pt}){10-11} $(n,p)$
        &  max   & $\ell_{2,\infty}$ & max & $\ell_{2,\infty}$ & max & $\ell_{2,\infty}$ & max & $\ell_{2,\infty}$ & max & $\ell_{2,\infty}$\\
        \midrule 
(100,50) & 0.771 & 0.675 & 0.416 & 0.404 & 0.842 & 0.740 & 1.002 & 0.940 & 1.064 & 1.081 \\ 
  (100,100) & 0.662 & 0.578 & 0.344 & 0.341 & 0.764 & 0.610 & 0.959 & 0.842 & 1.016 & 1.057 \\ 
  (200,50) & 0.782 & 0.742 & 0.408 & 0.400 & 0.825 & 0.779 & 0.987 & 0.962 & 1.127 & 1.143 \\ 
  (200,100) & 0.704 & 0.623 & 0.332 & 0.333 & 0.716 & 0.595 & 0.971 & 0.909 & 1.104 & 1.115 \\ 
  (500,100) & 0.680 & 0.644 & 0.315 & 0.307 & 0.685 & 0.593 & 0.971 & 0.939 & 1.119 & 1.123 \\ 
  (500,200) & 0.619 & 0.556 & 0.225 & 0.213 & 0.580 & 0.435 & 0.965 & 0.901 & 1.123 & 1.137 \\ 
        \bottomrule
    \end{tabular}
    \label{tab:auto_var_banded_lasso_diag_d0_1_d_2}
\end{table}

\begin{table}[H] 
    \centering
    \caption{VAR coefficient matrix estimation errors measured by $\mathrm{RME}_{\mathbb{A}}$ in~\eqref{eq:App_RME} with different matrix norms from $200$ realisations, when the data are generated as in \ref{item:fac_0} and \hyperref[eq:lagged_banded_A]{(V3)}, as $n$, $p$ and the innovation distribution vary (see \ref{innov_dist_gauss}--\ref{innov_dist_log_norm}). The VAR coefficient matrices are estimated with VAR order set to $d_f = 3$ while the true VAR order is $d = 2$.
    }
    \begin{tabular}{c cc cc cc cc cc}
        \toprule & 
                 \multicolumn{2}{c}{log-normal} &
                 \multicolumn{2}{c}{$t_{2.1}$} & \multicolumn{2}{c}{$t_{3}$} & \multicolumn{2}{c}{$t_{4}$} & 
                 \multicolumn{2}{c}{Normal}\\
        \cmidrule(l{2pt}r{2pt}){2-3}
        \cmidrule(l{2pt}r{2pt}){4-5}
        \cmidrule(l{2pt}r{2pt}){6-7}
        \cmidrule(l{2pt}r{2pt}){8-9}
        \cmidrule(l{2pt}r{2pt}){10-11} $(n,p)$
        &  max   & $\ell_{2,\infty}$ & max & $\ell_{2,\infty}$ & max & $\ell_{2,\infty}$ & max & $\ell_{2,\infty}$ & max & $\ell_{2,\infty}$\\
        \midrule 
(100,50) & 0.591 & 0.668 & 0.332 & 0.407 & 0.699 & 0.720 & 0.939 & 0.898 & 1.052 & 1.022 \\ 
  (100,100) & 0.467 & 0.582 & 0.279 & 0.351 & 0.534 & 0.577 & 0.822 & 0.785 & 1.026 & 0.995 \\ 
  (200,50) & 0.733 & 0.789 & 0.343 & 0.430 & 0.823 & 0.833 & 0.990 & 1.032 & 1.030 & 1.155 \\ 
  (200,100) & 0.586 & 0.657 & 0.277 & 0.340 & 0.653 & 0.631 & 0.970 & 0.935 & 1.030 & 1.095 \\ 
  (500,100) & 0.819 & 0.854 & 0.280 & 0.317 & 0.779 & 0.739 & 1.077 & 1.218 & 1.114 & 1.386 \\ 
  (500,200) & 0.707 & 0.693 & 0.252 & 0.310 & 0.668 & 0.612 & 1.009 & 1.120 & 1.037 & 1.302 \\ 
        \bottomrule
    \end{tabular}
    \label{tab:auto_var_banded_lasso_diag_d0_2_d_3}
\end{table}

\begin{figure}[H] 
    \centering
    \begin{tabular}{c}
    \includegraphics[width=.9\textwidth]{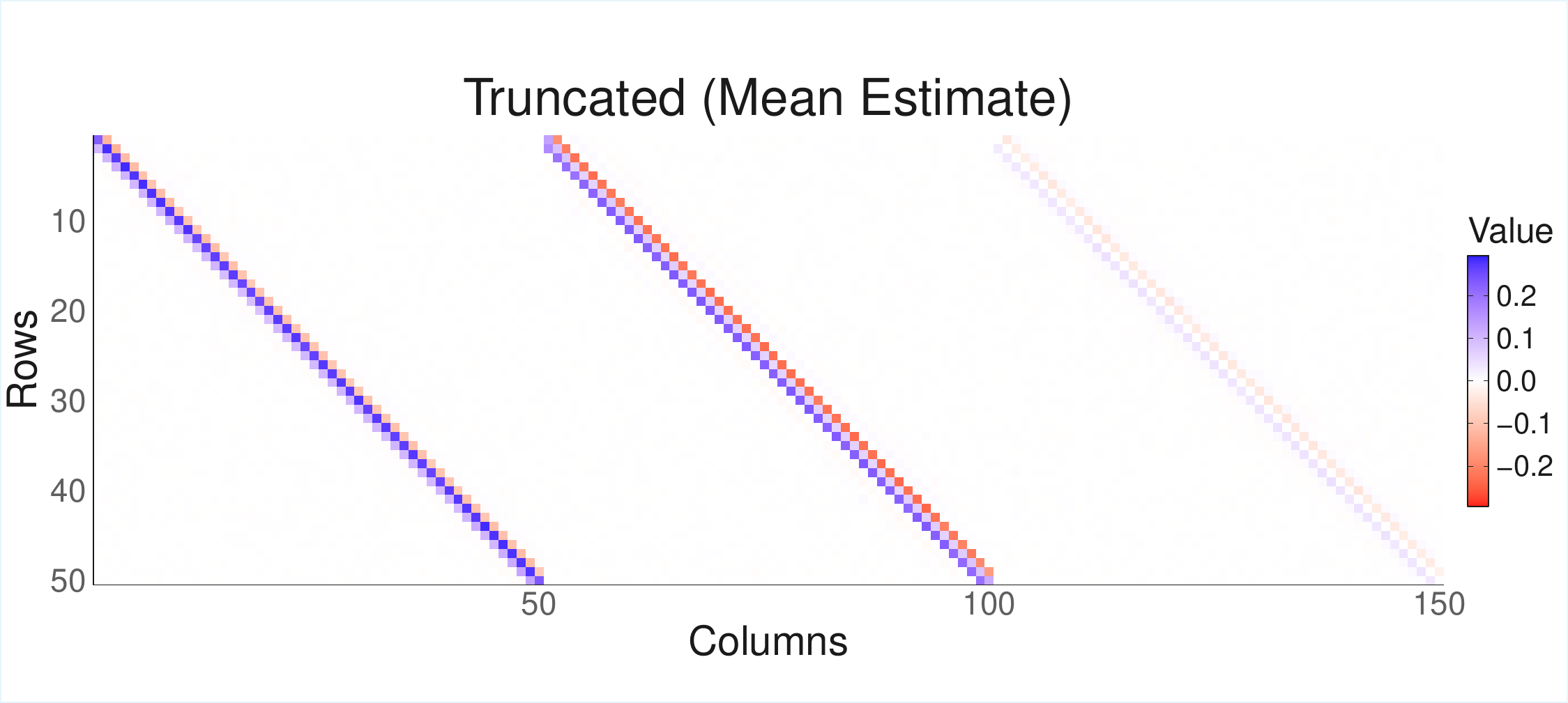} \\
    \includegraphics[width=.9\textwidth]{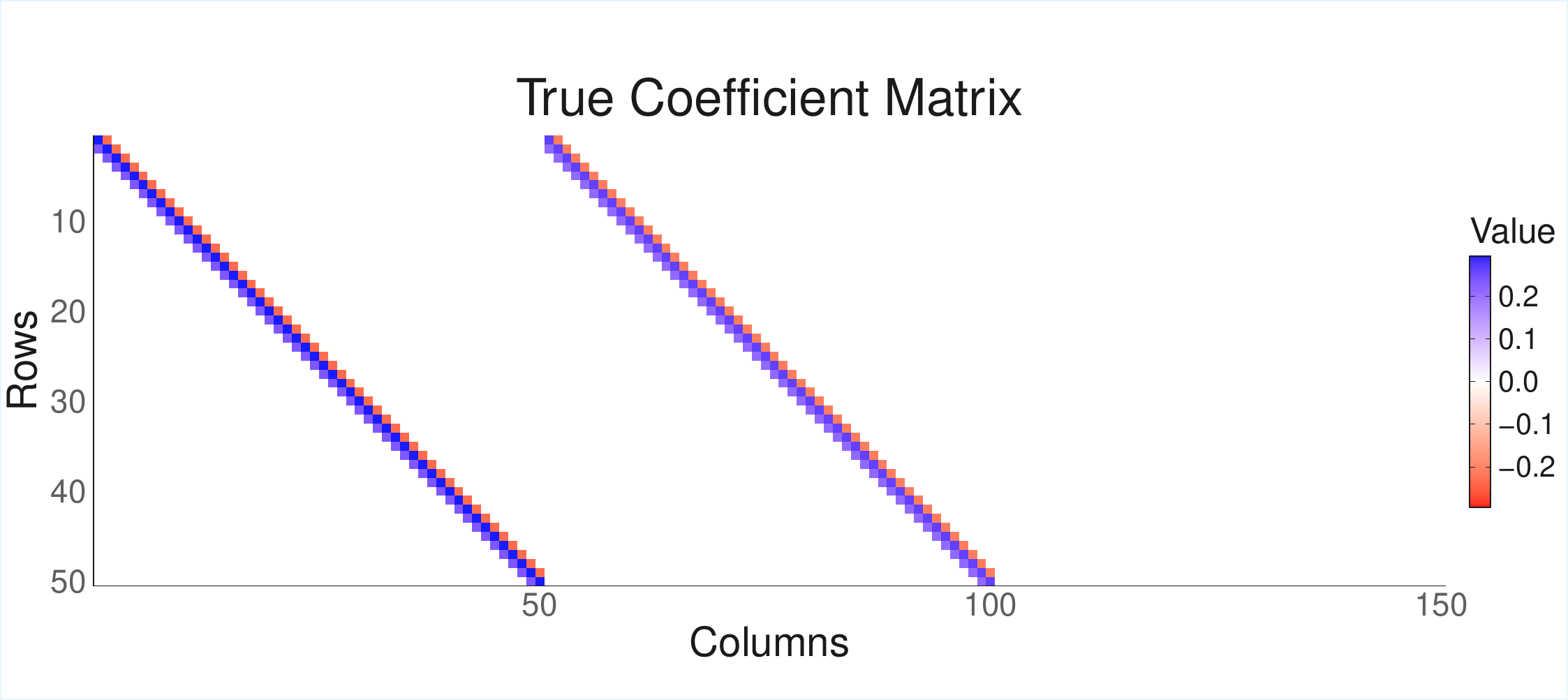}
    \end{tabular}
    \caption{Heatmap of the coefficient matrix estimates $\widehat{\mathbb{A}}(\tau)$ returned by the proposed method over $200$ repetitions (top), and the true coefficient matrix $\mathbb{A}$ (bottom). The data are generated according to \ref{item:fac_0} and \hyperref[eq:lagged_banded_A]{(V3)}, with $(n = 200, p = 50)$ and Student-$t$ innovations with $2.1$ degrees of freedom. The VAR coefficients are estimated with order $d_f = 3$, while the true VAR order is $d = 2$.}
    \label{fig:heatmap_estimate}
\end{figure}

\clearpage

\subsection{Covariance estimation for independent data} \label{sec:cov_sim}

In this subsection, we present covariance estimation errors of the sample covariance with truncation. The motivation for providing such results, is to illustrate the performance of our tuning procedure for the goal of covariance estimation, which may be of independent interest. We investigate the performance of the truncation estimator of the covariance for i.i.d data, as is done in other papers studying tail-robust covariance estimation such as \citet{Ke2019}, under the following scenario:
\begin{enumerate}[label = (V\arabic*), wide, labelindent=0pt, start = 0]
\item \label{item:var_0} {\bf Independent data.} $\mathbf{A} = \textbf{O}$.
\end{enumerate}

For this subsection, we only include the sample covariance in the CV measure, that is, $d$ is set to zero in \eqref{eq:CV_measure}. \Cref{tab:iid_cov_diag} and \Cref{tab:iid_cov_pwrdec} present
a summary of the relative performance of the sample covariance estimator with truncated data against that of the sample covariance with no truncation (with $\tau = \infty$), either when (S1) $\bSigma_{\bxii} = \mathbf{I}$ and (S2) $\bSigma_{\bxii} = [0.9^{\vert i - j \vert}]_{i, j = 1}^p$, respectively. The errors are measured by
\begin{equation} \label{eq:RME_cov}
\mathrm{RME}_{\bGamma}=\frac{\sum_{i=1}^{200}\left|\widehat{\bGamma}_\bx^{(i)}(\tau) - \bGamma_\bx\right|_{m}}{\sum_{i=1}^{200}\left|\widehat{\bGamma}_\bx^{(i)}(\infty) -\bGamma_\bx\right|_{m}} \, ,
\end{equation}
here, $m$ denotes the matrix
norm used to compute the relative mean error (`RME'), which we will vary in the tables below, max denotes the max
element-wise norm, and F the Frobenius norm.

\begin{table}[H] 
    \centering
    \caption{
    Covariance estimation errors measured by $\mathrm{RME}_{\bGamma}$ \eqref{eq:RME_cov} with different matrix norms from 200 realisations. With i.i.d data (\ref{item:fac_0}, \ref{item:var_0}) generated with a diagonal covariance structure (S1), with varying $(n, p)$ and innovation distributions \ref{innov_dist_gauss}--\ref{innov_dist_log_norm}.
    }
    \begin{tabular}{c cc cc cc cc cc}
        \toprule & 
                 \multicolumn{2}{c}{log-normal} &
                 \multicolumn{2}{c}{$t_{2.1}$} & \multicolumn{2}{c}{$t_{3}$} & \multicolumn{2}{c}{$t_{4}$} & 
                 \multicolumn{2}{c}{Normal}\\
        \cmidrule(l{2pt}r{2pt}){2-3}
        \cmidrule(l{2pt}r{2pt}){4-5}
        \cmidrule(l{2pt}r{2pt}){6-7}
        \cmidrule(l{2pt}r{2pt}){8-9}
        \cmidrule(l{2pt}r{2pt}){10-11}
        $(n,p)$ &  max   & F & max & F & max & F & max & F & max & F\\
        \midrule 
(50,100) & 0.152 & 0.543 & 0.052 & 0.450 & 0.089 & 0.657 & 0.283 & 0.796 & 1.007 & 0.842 \\ 
  (50,200) & 0.133 & 0.579 & 0.028 & 0.078 & 0.083 & 0.631 & 0.193 & 0.808 & 0.954 & 0.824 \\ 
  (100,200) & 0.151 & 0.626 & 0.027 & 0.197 & 0.075 & 0.534 & 0.215 & 0.871 & 1.022 & 0.898 \\ 
  (200,50) & 0.270 & 0.662 & 0.120 & 0.546 & 0.158 & 0.637 & 0.360 & 0.816 & 1.094 & 0.952 \\ 
  (200,100) & 0.191 & 0.699 & 0.072 & 0.437 & 0.116 & 0.665 & 0.307 & 0.884 & 1.060 & 0.940 \\ 
        \bottomrule
    \end{tabular}
    \label{tab:iid_cov_diag}
\end{table}

\begin{table}[H] 
    \centering
    \caption{Covariance estimation errors measured by $\mathrm{RME}_{\bGamma}$ \eqref{eq:RME_cov} with different matrix norms from 200 realisations. With i.i.d data (\ref{item:fac_0}, \ref{item:var_0}) generated with a power-decay covariance structure (S2), with varying $(n, p)$ and innovation distributions \ref{innov_dist_gauss}--\ref{innov_dist_log_norm}.}
    \begin{tabular}{c cc cc cc cc cc}
        \toprule & 
                 \multicolumn{2}{c}{log-normal} &
                 \multicolumn{2}{c}{$t_{2.1}$} & \multicolumn{2}{c}{$t_{3}$} & \multicolumn{2}{c}{$t_{4}$} & 
                 \multicolumn{2}{c}{Normal}\\
        \cmidrule(l{2pt}r{2pt}){2-3}
        \cmidrule(l{2pt}r{2pt}){4-5}
        \cmidrule(l{2pt}r{2pt}){6-7}
        \cmidrule(l{2pt}r{2pt}){8-9}
        \cmidrule(l{2pt}r{2pt}){10-11}
       $(n,p)$ &  max   & F & max & F & max & F & max & F & max & F\\
        \midrule 
(50,100) & 0.241 & 0.731 & 0.152 & 0.189 & 0.252 & 0.764 & 0.478 & 0.897 & 1.026 & 0.953 \\ 
  (50,200) & 0.195 & 0.731 & 0.068 & 0.589 & 0.127 & 0.746 & 0.331 & 0.894 & 1.001 & 0.899 \\ 
  (100,200) & 0.203 & 0.793 & 0.048 & 0.575 & 0.156 & 0.682 & 0.361 & 0.909 & 1.072 & 0.948 \\ 
  (200,50) & 0.324 & 0.770 & 0.305 & 0.848 & 0.295 & 0.807 & 0.581 & 0.945 & 1.092 & 1.063 \\ 
  (200,100) & 0.272 & 0.806 & 0.026 & 0.727 & 0.197 & 0.741 & 0.435 & 0.926 & 1.090 & 1.005 \\ 
        \bottomrule
    \end{tabular}
    \label{tab:iid_cov_pwrdec}
\end{table}

\clearpage

\subsection{(Auto)covariance estimation for dependent data}
\label{sec:acv_sim}

In this subsection we present (auto)covariance estimation errors for dependent data,  we consider the DGPs defined in \Cref{sec:simulations}. In addition, we provide figures presenting how the choice of $\tau$ from our CV tuning procedure, described in \Cref{sec:CV_tune_tau}, varies with $n$ and $p$.

\subsubsection{Covariance estimation}

\Cref{tab:var_banded_cov_diag} and \Cref{tab:var_renyi_cov_diag}, present a summary of the relative performance of the sample covariance estimator with truncated data against that of the sample covariance with no truncation (with $\tau = \infty$), measured by \ref{eq:RME_cov}, for data generated from a VAR$(1)$ model under \ref{eq:tri_diag_A_sparse_var} and \ref{eq:renyi_A}, respectively. From the results we observe that the largest ratio of improvement is in the heaviest tailed case, as expected, which lessens until we approach efficiency in the Gaussian case. 

\Cref{fig:tau_vs_p_and_n} plots the choice of $\tau$ against increasing sample size and dimension for various innovation distribution settings when the data are generated from a VAR$(1)$ model \ref{eq:tri_diag_A_sparse_var}. \Cref{fig:tau_vs_p_and_n} shows clearly that for heavier-tailed distributions a larger $\tau$ is chosen, increasing with sample size agreeing with the theoretical choice of $\tau$ in \eqref{eq:tau_prop}.

\begin{table}[h!t!] 
    \centering
    \caption{Covariance estimation errors measured by $\mathrm{RME}_{\bGamma}$ \eqref{eq:RME_cov} with different matrix norms from 200 realisations. With data generated from the process \ref{item:fac_0}, \ref{eq:tri_diag_A_sparse_var}, diagonal structure (S1), with varying $(n, p)$ and innovation distributions \ref{innov_dist_gauss}--\ref{innov_dist_log_norm}.}
    \begin{tabular}{c cc cc cc cc cc}
        \toprule & 
                 \multicolumn{2}{c}{log-normal} &
                 \multicolumn{2}{c}{$t_{2.1}$} & \multicolumn{2}{c}{$t_{3}$} & \multicolumn{2}{c}{$t_{4}$} & 
                 \multicolumn{2}{c}{Normal}\\
        \cmidrule(l{2pt}r{2pt}){2-3}
        \cmidrule(l{2pt}r{2pt}){4-5}
        \cmidrule(l{2pt}r{2pt}){6-7}
        \cmidrule(l{2pt}r{2pt}){8-9}
        \cmidrule(l{2pt}r{2pt}){10-11}
        $(n,p)$ &  max   & F & max & F & max & F & max & F & max & F\\
        \midrule 
(50,100) & 0.353 & 0.785 & 0.222 & 0.746 & 0.312 & 0.766 & 0.557 & 0.761 & 0.803 & 0.698 \\ 
  (50,200) & 0.331 & 0.803 & 0.061 & 0.326 & 0.241 & 0.790 & 0.501 & 0.780 & 0.733 & 0.621 \\ 
  (100,200) & 0.316 & 0.862 & 0.105 & 0.636 & 0.213 & 0.834 & 0.549 & 0.864 & 0.855 & 0.720 \\ 
  (200,50) & 0.468 & 0.871 & 0.133 & 0.433 & 0.285 & 0.779 & 0.792 & 0.907 & 1.045 & 0.871 \\ 
  (200,100) & 0.412 & 0.899 & 0.192 & 0.709 & 0.351 & 0.883 & 0.644 & 0.898 & 0.974 & 0.844 \\ 
        \bottomrule
    \end{tabular}
    \label{tab:var_banded_cov_diag}
\end{table}

\begin{figure}[H]
    \centering
    \begin{minipage}[b]{0.48\textwidth}
    \includegraphics[width=\textwidth]{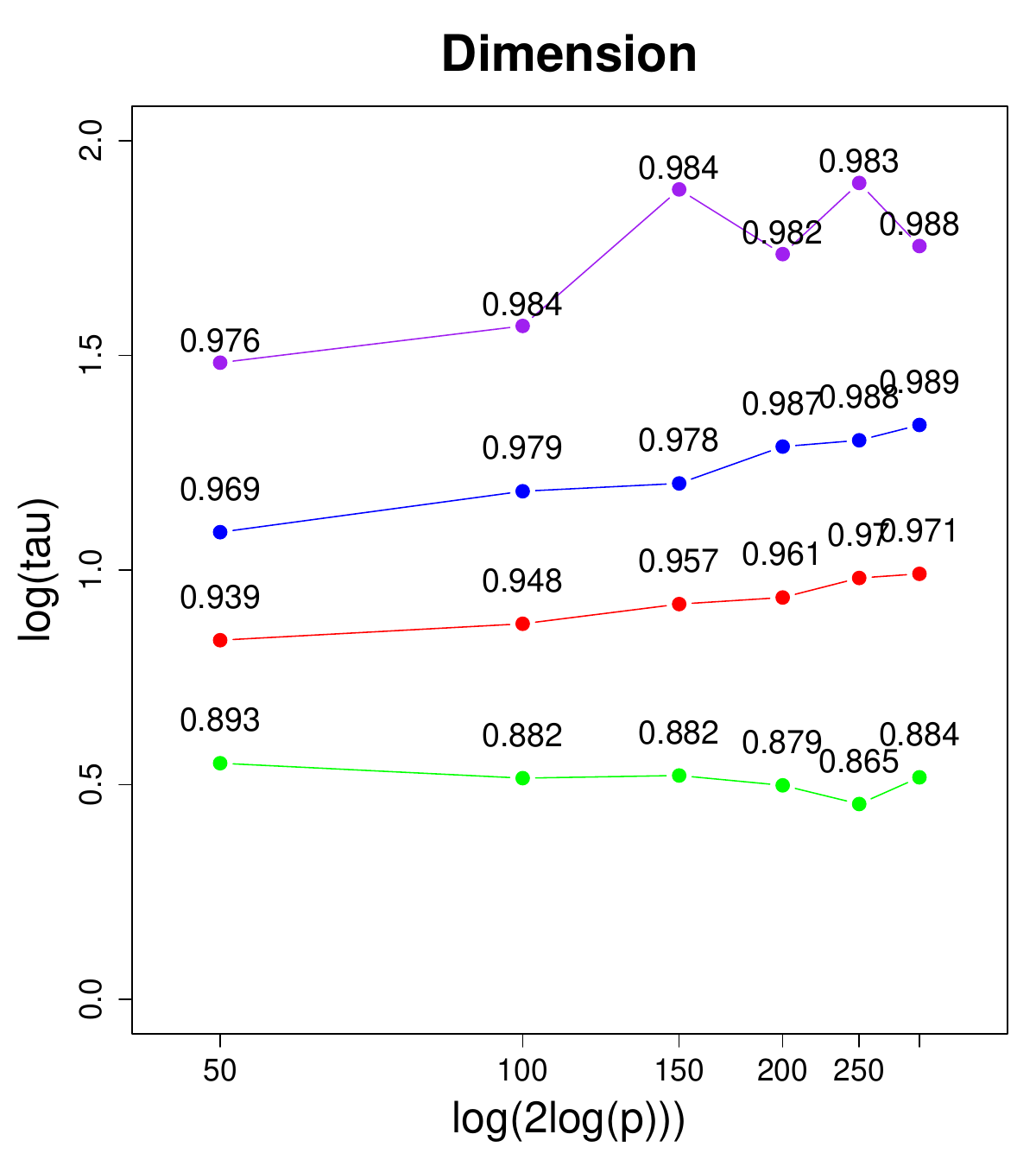}
    \end{minipage}
  \begin{minipage}[b]{0.48\textwidth}
  
    \includegraphics[width=\textwidth]{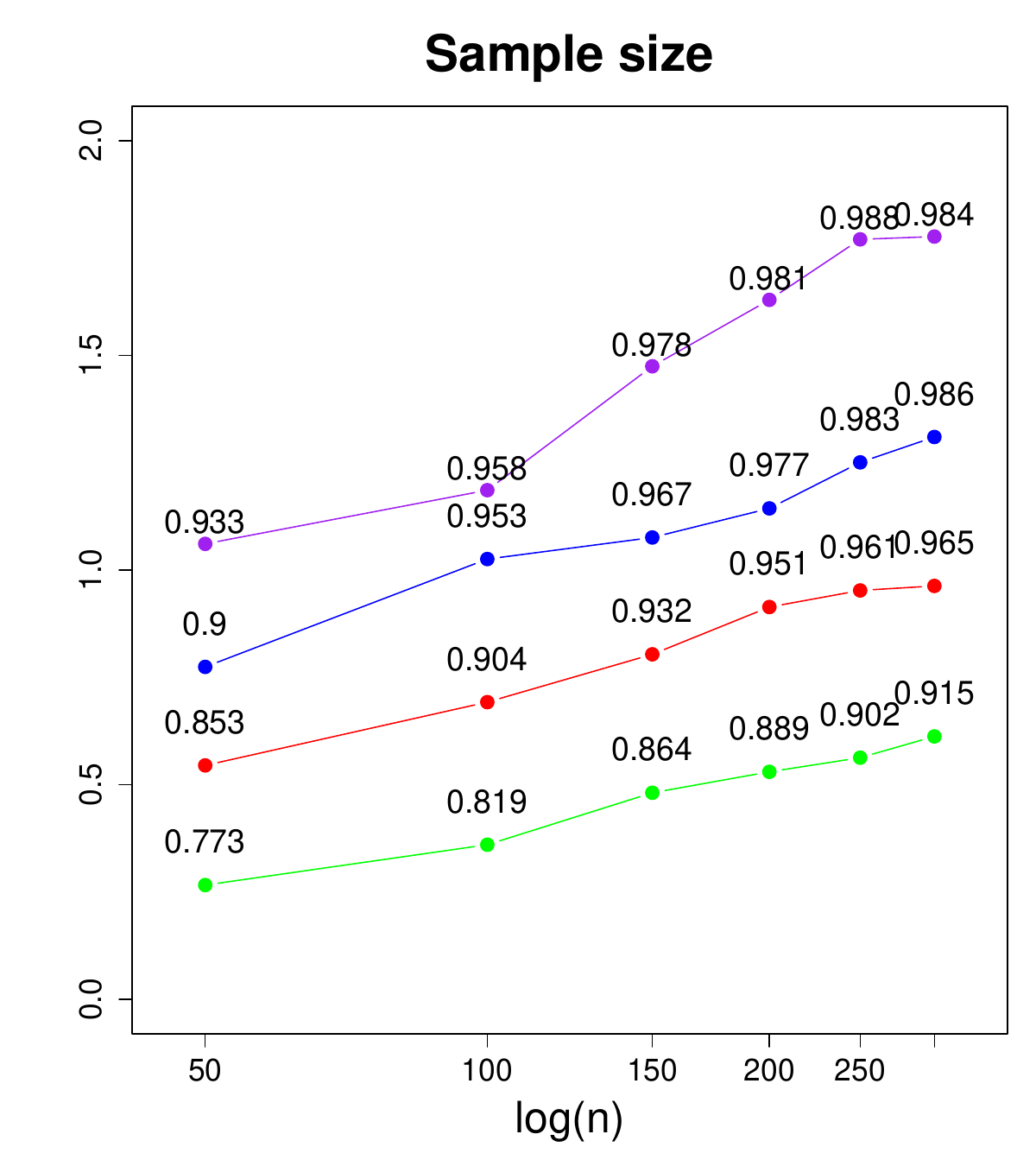}
    \end{minipage}
        \begin{minipage}[t]{\textwidth}
        \centering
        \includegraphics[width=8cm]{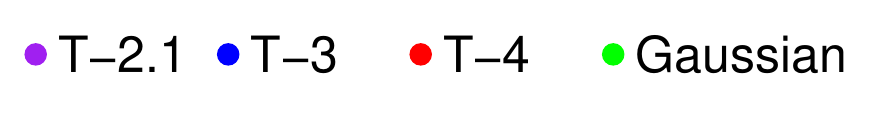} 
    \end{minipage}
    \centering
    \begin{minipage}[b]{0.48\textwidth}

    \includegraphics[width=\textwidth]{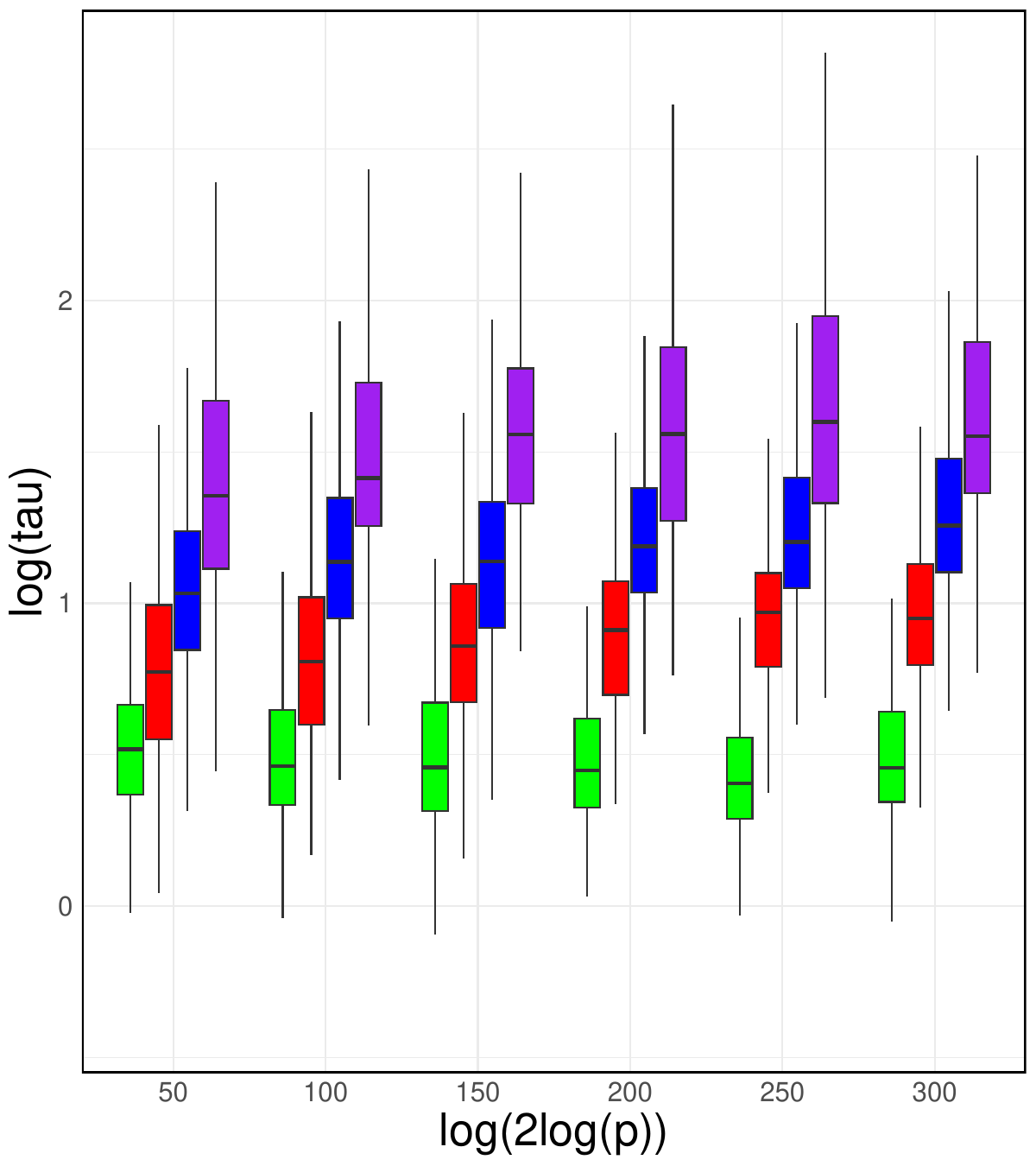}
    \end{minipage}
  \begin{minipage}[b]{0.48\textwidth}
  
    \includegraphics[width=\textwidth]{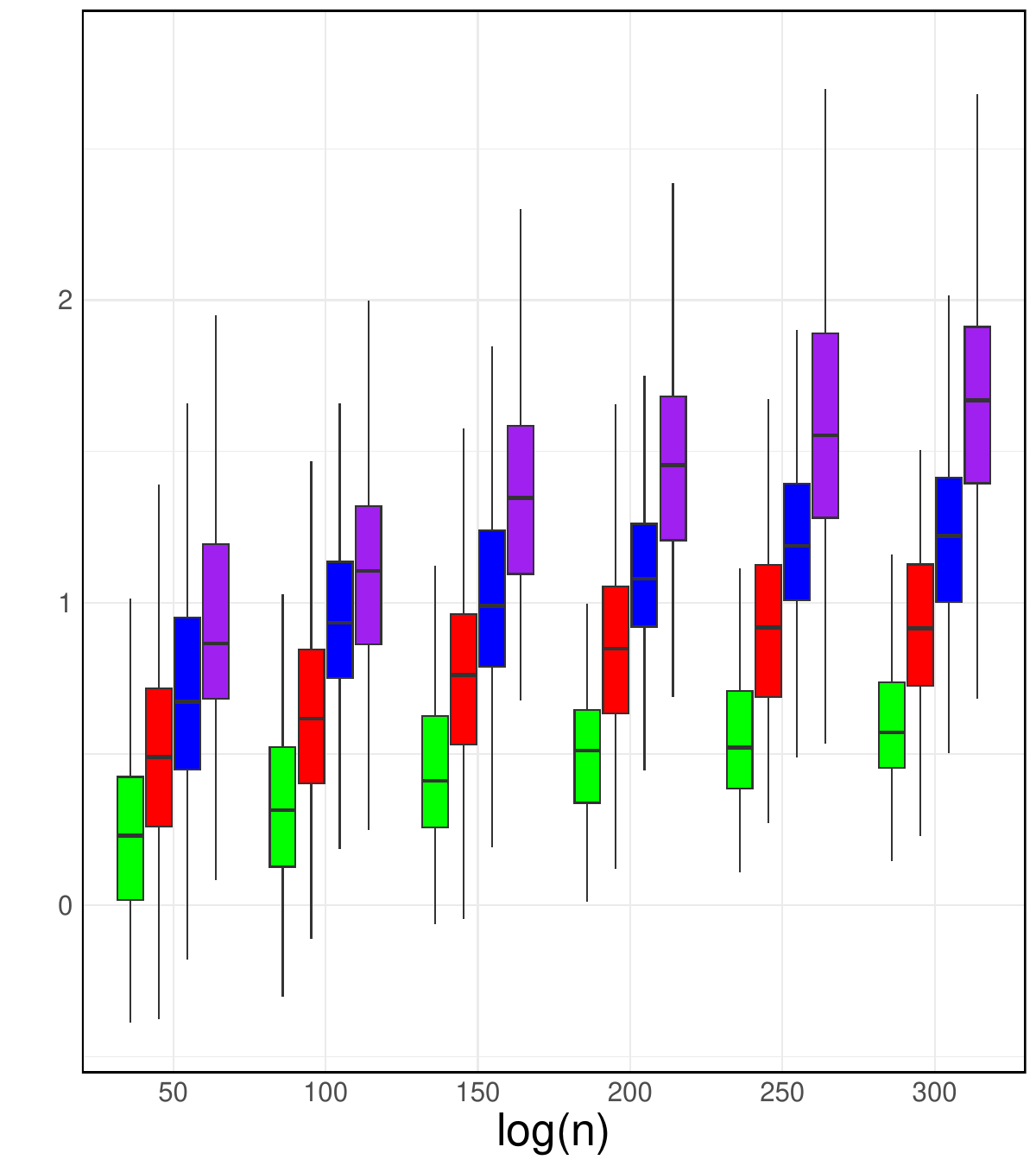}
    \end{minipage}

        \caption{Plots of $\log(\tau)$ chosen from the CV procedure described in \Cref{sec:CV_tune_tau}, with $d = 0$ in the CV measure \eqref{eq:CV_measure}, over 200 simulations against dimension $p$ ($n$ fixed at 200), and sample size $n$ ($p$ fixed at 100), and as the innovation distribution varies (see~\ref{innov_dist_gauss}--\ref{innov_dist_log_norm}). The data are generated as in \ref{item:fac_0}, \ref{eq:tri_diag_A_sparse_var}, (S1). The first row shows the average $\tau$ chosen across simulations, the numbers on each point is the average percentage of data points in absolute values that are less than $\tau$. The second row gives the box plots of the chosen $\tau$ across the simulations.}
    \label{fig:tau_vs_p_and_n}
\end{figure}

\begin{table}[h!t!] 
    \centering
    \caption{Covariance estimation errors measured by $\mathrm{RME}_{\bGamma}$ \eqref{eq:RME_cov} with different matrix norms from 200 realisations. With data generated from the process \ref{item:fac_0}, \ref{eq:renyi_A}, diagonal structure (S1), with varying $(n, p)$ and innovation distributions \ref{innov_dist_gauss}--\ref{innov_dist_log_norm}.}
    \begin{tabular}{c cc cc cc cc cc}
        \toprule & 
                 \multicolumn{2}{c}{log-normal} &
                 \multicolumn{2}{c}{$t_{2.1}$} & \multicolumn{2}{c}{$t_{3}$} & \multicolumn{2}{c}{$t_{4}$} & 
                 \multicolumn{2}{c}{Normal}\\
        \cmidrule(l{2pt}r{2pt}){2-3}
        \cmidrule(l{2pt}r{2pt}){4-5}
        \cmidrule(l{2pt}r{2pt}){6-7}
        \cmidrule(l{2pt}r{2pt}){8-9}
        \cmidrule(l{2pt}r{2pt}){10-11}
        $(n,p)$ &  max   & F & max & F & max & F & max & F & max & F\\
        \midrule 
(50,100) & 0.162 & 0.590 & 0.073 & 0.372 & 0.126 & 0.638 & 0.274 & 0.820 & 1.061 & 0.851 \\ 
  (50,200) & 0.154 & 0.629 & 0.065 & 0.434 & 0.064 & 0.596 & 0.240 & 0.842 & 1.039 & 0.825 \\ 
  (100,200) & 0.154 & 0.677 & 0.066 & 0.426 & 0.108 & 0.722 & 0.246 & 0.883 & 1.091 & 0.893 \\ 
  (200,50) & 0.282 & 0.689 & 0.101 & 0.374 & 0.229 & 0.694 & 0.363 & 0.847 & 1.132 & 0.954 \\ 
  (200,100) & 0.249 & 0.733 & 0.036 & 0.206 & 0.148 & 0.711 & 0.326 & 0.895 & 1.113 & 0.938 \\ 
        \bottomrule
    \end{tabular}
    \label{tab:var_renyi_cov_diag}
\end{table}

\subsubsection{Autocovariance estimation}

In this subsection, motivated by the fact that \Cref{prop:fac_adj_idio_rate} requires the rate on the term
\begin{gather*}
    \max \left\{\left| \widehat{\bGamma}(\tau) - \bGamma\right|_\infty ,  \left|\widehat{\bgamma}(\tau) - \bgamma\right|_\infty \right\} \, ,
\end{gather*}
which includes lag-$h$ autocovariance estimates for $0\leq h \leq d$, we now include the autocovariance matrix of lag one when computing the errors of our estimates. In addition, we include the lag one autocovariance in our CV measure, that is $d$ is now set to one in \eqref{eq:CV_measure}.
Therefore, for the tables in this subsection we compute the following modified RME:
\begin{equation} \label{eq:RME_lag}
\mathrm{RME}_{\ell}=\frac{\sum_{i=1}^{200}\underset{0 \leq h \leq d}{\max}\left|\widehat{\boldsymbol{\Gamma}}_{\bx}(\tau,h)-\boldsymbol{\Gamma}_{\bx}(h)\right|_{m}}{\sum_{i=1}^{200}\underset{0 \leq h \leq d}{\max}\left|\widehat{\boldsymbol{\Gamma}}_{\bx}(\infty, h)-\boldsymbol{\Gamma}_{\bx}(h)\right|_{m}} \, .
\end{equation} 
As was written in Section~\ref{sec:cov_sim}, $m$ denotes the matrix norm used to compute the $\mathrm{RME}_{\ell}$, which we will vary in the tables below.
\Cref{tab:auto_var_banded_cov_diag,tab:auto_var_renyi_cov_diag,tab:auto_fac_var_banded_diag,tab:auto_fac_var_renyi_diag}, display these $\mathrm{RME}_{\ell}$ values for the varying DGPs, and the results are very similar to the results for covariance estimation in Section~\ref{sec:cov_sim}.
In addition, in \Cref{fig:tau_vs_p_and_n_auto} we plot the $\tau$ values chosen by CV, analogous to \Cref{fig:tau_vs_p_and_n}. Comparing to the plots of \Cref{fig:tau_vs_p_and_n}, it shows that there is a minor difference in the $\tau$ that is chosen when $d = 1$ is used compared to $d = 0$ in the CV measure \eqref{eq:CV_measure}.

\begin{table}[h!t!] 
    \centering
    \caption{Maximum of covariance and autocovariance estimation errors, as measured by $\mathrm{RME}_\ell$~\eqref{eq:RME_lag} with $d = 1$, with different matrix norms from 200 realisations. With data generated from the process \ref{item:fac_0} and \ref{eq:tri_diag_A_sparse_var}, with varying $(n, p)$ and innovation distributions \ref{innov_dist_gauss}--\ref{innov_dist_log_norm}.}
    \begin{tabular}{c cc cc cc cc cc}
        \toprule & 
                 \multicolumn{2}{c}{log-normal} &
                 \multicolumn{2}{c}{$t_{2.1}$} & \multicolumn{2}{c}{$t_{3}$} & \multicolumn{2}{c}{$t_{4}$} & 
                 \multicolumn{2}{c}{Normal}\\
        \cmidrule(l{2pt}r{2pt}){2-3}
        \cmidrule(l{2pt}r{2pt}){4-5}
        \cmidrule(l{2pt}r{2pt}){6-7}
        \cmidrule(l{2pt}r{2pt}){8-9}
        \cmidrule(l{2pt}r{2pt}){10-11}
        $(n,p)$ &  max   & F & max & F & max & F & max & F & max & F\\
        \midrule 
(50,100) & 0.357 & 0.793 & 0.222 & 0.752 & 0.319 & 0.771 & 0.566 & 0.750 & 0.778 & 0.650 \\ 
  (50,200) & 0.337 & 0.810 & 0.061 & 0.340 & 0.245 & 0.795 & 0.512 & 0.771 & 0.708 & 0.565 \\ 
  (100,200) & 0.319 & 0.868 & 0.105 & 0.645 & 0.216 & 0.840 & 0.557 & 0.855 & 0.846 & 0.678 \\ 
  (200,50) & 0.469 & 0.875 & 0.133 & 0.445 & 0.287 & 0.788 & 0.805 & 0.898 & 1.054 & 0.855 \\ 
  (200,100) & 0.412 & 0.904 & 0.192 & 0.716 & 0.355 & 0.888 & 0.650 & 0.891 & 0.977 & 0.815 \\ 
        \bottomrule
    \end{tabular}
    \label{tab:auto_var_banded_cov_diag}
\end{table}

\begin{table}[h!t!] 
    \centering
    \caption{Maximum of covariance and autocovariance estimation errors, as measured by $\mathrm{RME}_\ell$~\eqref{eq:RME_lag} with $d = 1$, with different matrix norms from 200 realisations. With data generated from the process \ref{item:fac_0} and \ref{eq:renyi_A}, with varying $(n, p)$ and innovation distributions \ref{innov_dist_gauss}--\ref{innov_dist_log_norm}.}
    \begin{tabular}{c cc cc cc cc cc}
        \toprule & 
                 \multicolumn{2}{c}{log-normal} &
                 \multicolumn{2}{c}{$t_{2.1}$} & \multicolumn{2}{c}{$t_{3}$} & \multicolumn{2}{c}{$t_{4}$} & 
                 \multicolumn{2}{c}{Normal}\\
        \cmidrule(l{2pt}r{2pt}){2-3}
        \cmidrule(l{2pt}r{2pt}){4-5}
        \cmidrule(l{2pt}r{2pt}){6-7}
        \cmidrule(l{2pt}r{2pt}){8-9}
        \cmidrule(l{2pt}r{2pt}){10-11}
        $(n,p)$ &  max   & F & max & F & max & F & max & F & max & F\\
        \midrule 
(50,100) & 0.164 & 0.614 & 0.072 & 0.371 & 0.126 & 0.671 & 0.279 & 0.827 & 1.087 & 0.782 \\ 
  (50,200) & 0.155 & 0.644 & 0.065 & 0.438 & 0.065 & 0.641 & 0.244 & 0.843 & 1.056 & 0.752 \\ 
  (100,200) & 0.154 & 0.701 & 0.066 & 0.434 & 0.108 & 0.753 & 0.247 & 0.888 & 1.121 & 0.845 \\ 
  (200,50) & 0.283 & 0.718 & 0.101 & 0.376 & 0.229 & 0.732 & 0.364 & 0.869 & 1.196 & 0.926 \\ 
  (200,100) & 0.249 & 0.759 & 0.036 & 0.215 & 0.148 & 0.754 & 0.328 & 0.906 & 1.168 & 0.906 \\ 
        \bottomrule
    \end{tabular}
    \label{tab:auto_var_renyi_cov_diag}
\end{table}

\begin{figure}[h!t!]
    \centering
    \begin{minipage}[b]{0.48\textwidth}
    \includegraphics[width=\textwidth]{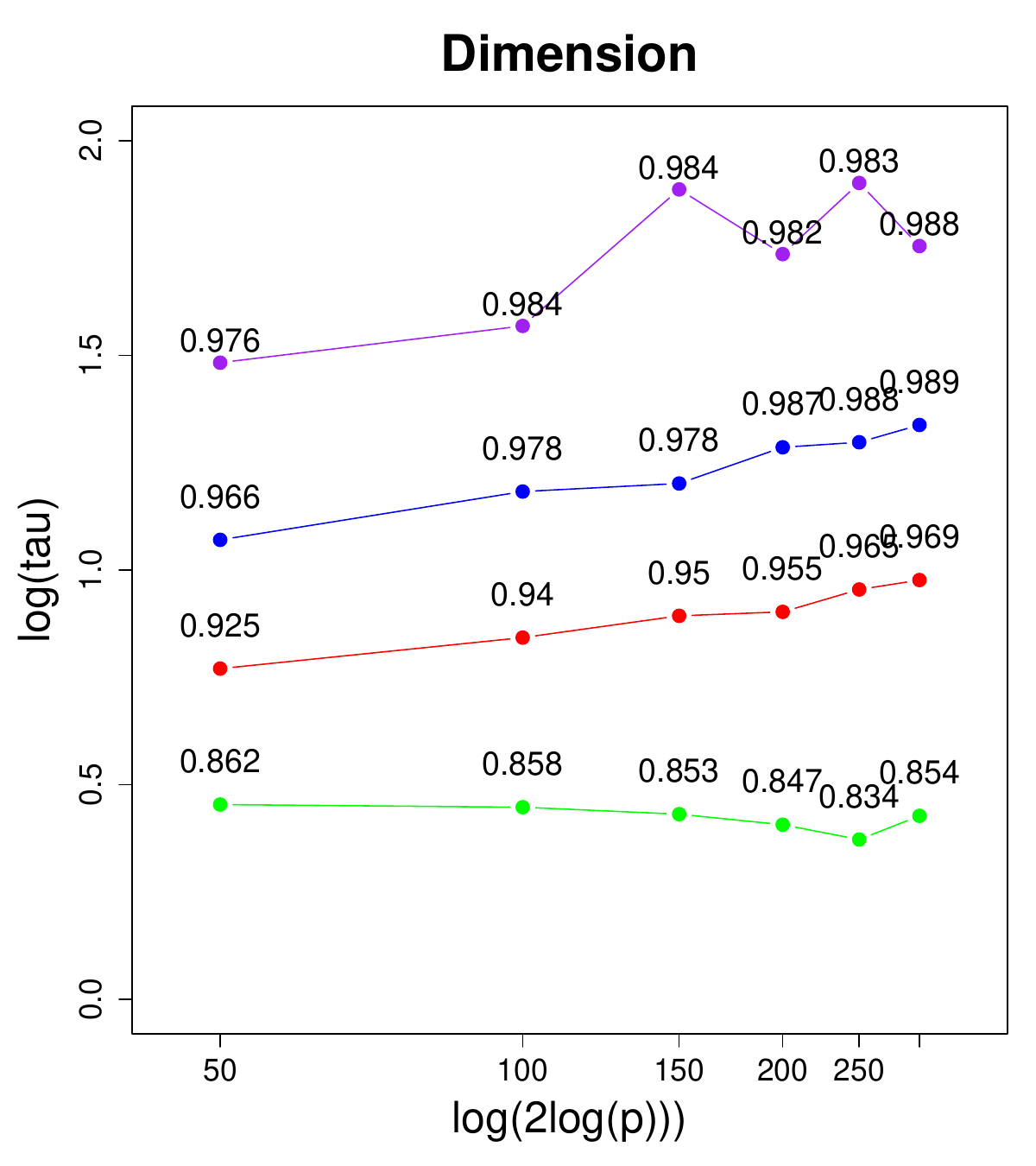}
    \end{minipage}
  \begin{minipage}[b]{0.48\textwidth}
  
    \includegraphics[width=\textwidth]{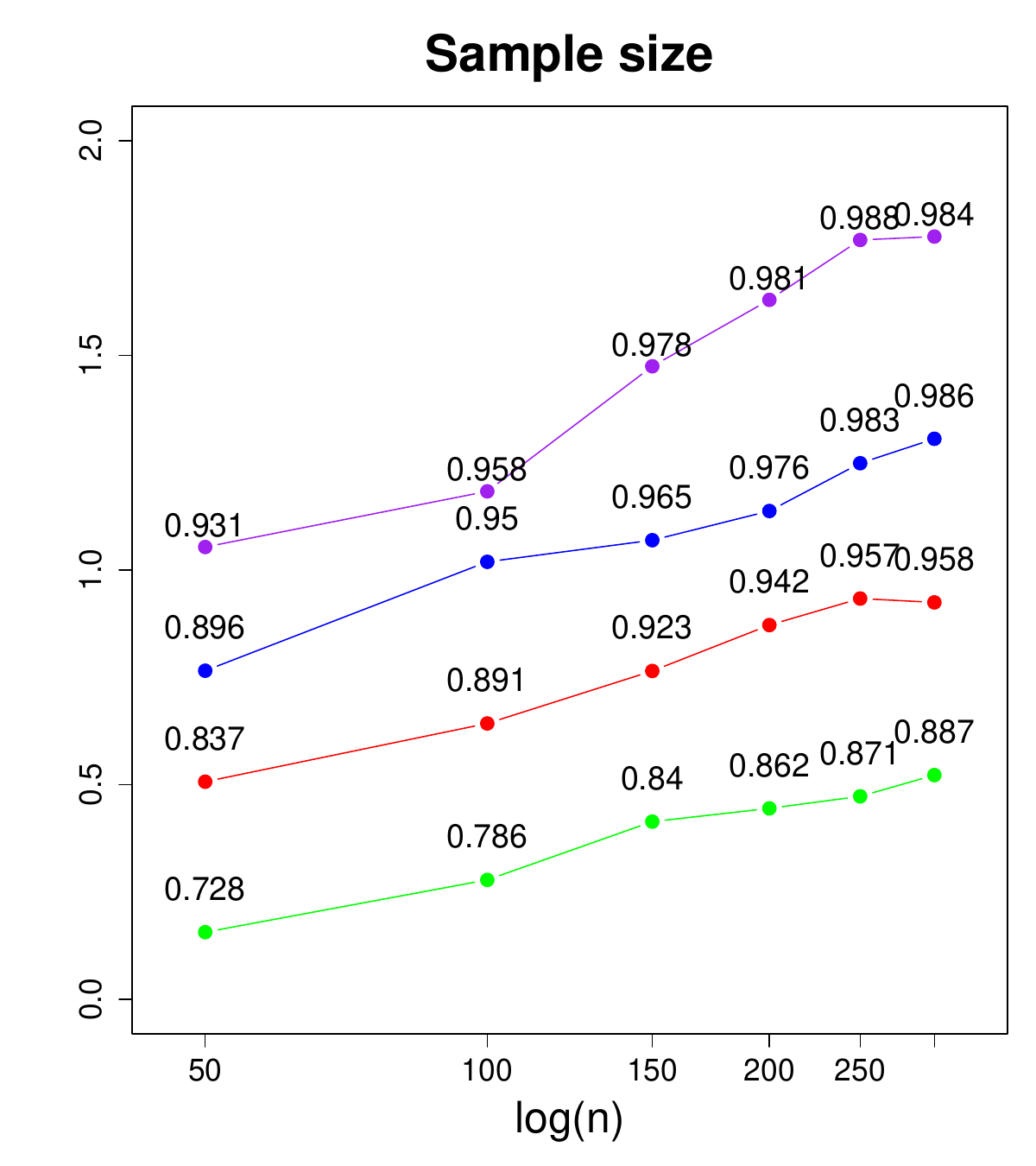}
    \end{minipage}
        \begin{minipage}[t]{\textwidth}
        \centering
        \includegraphics[width=8cm]{legend-distr_purp.pdf} 
    \end{minipage}
    \centering
    \begin{minipage}[b]{0.48\textwidth}
    
    \includegraphics[width=\textwidth]{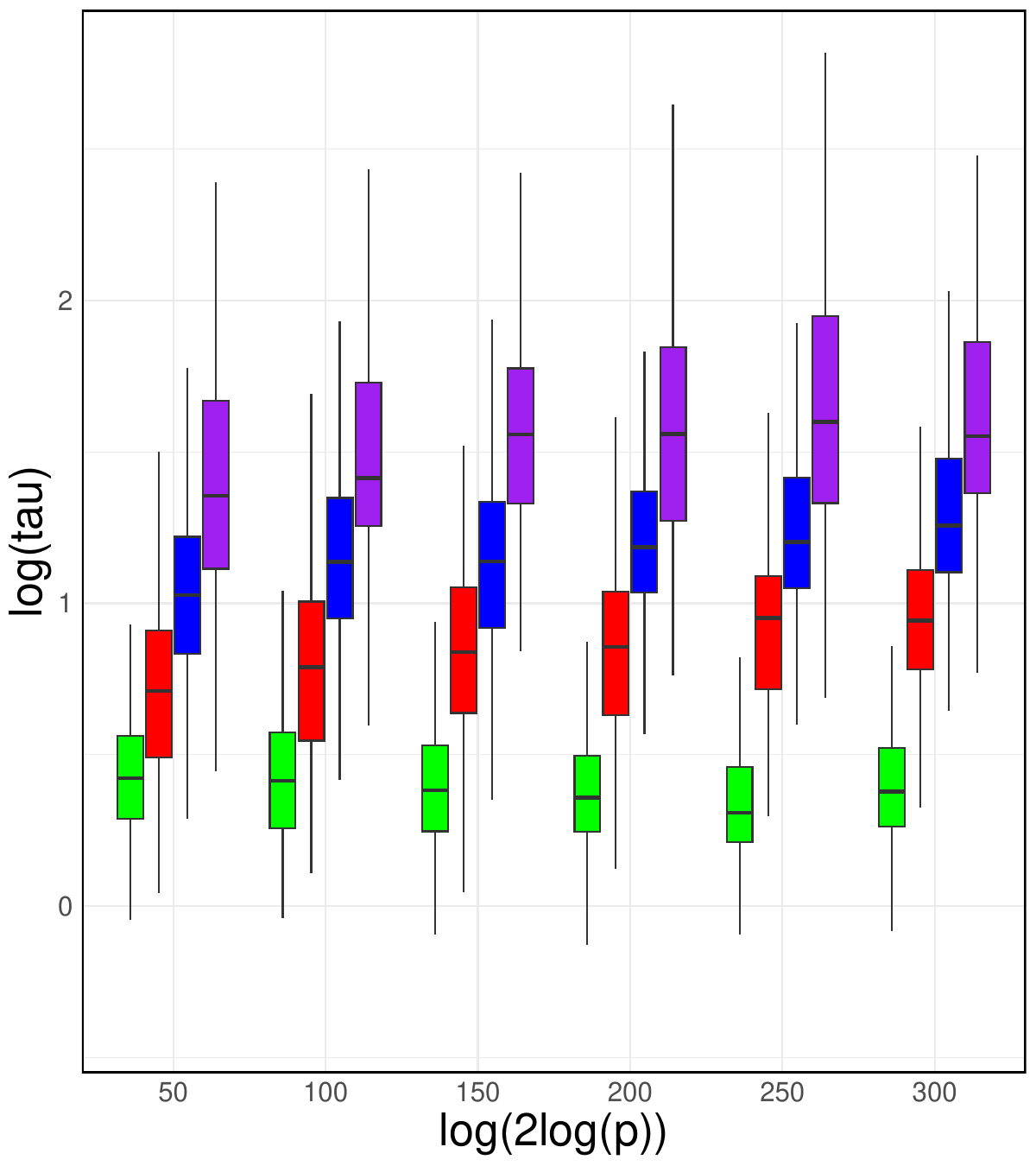}
    \end{minipage}
  \begin{minipage}[b]{0.48\textwidth}
  
    \includegraphics[width=\textwidth]{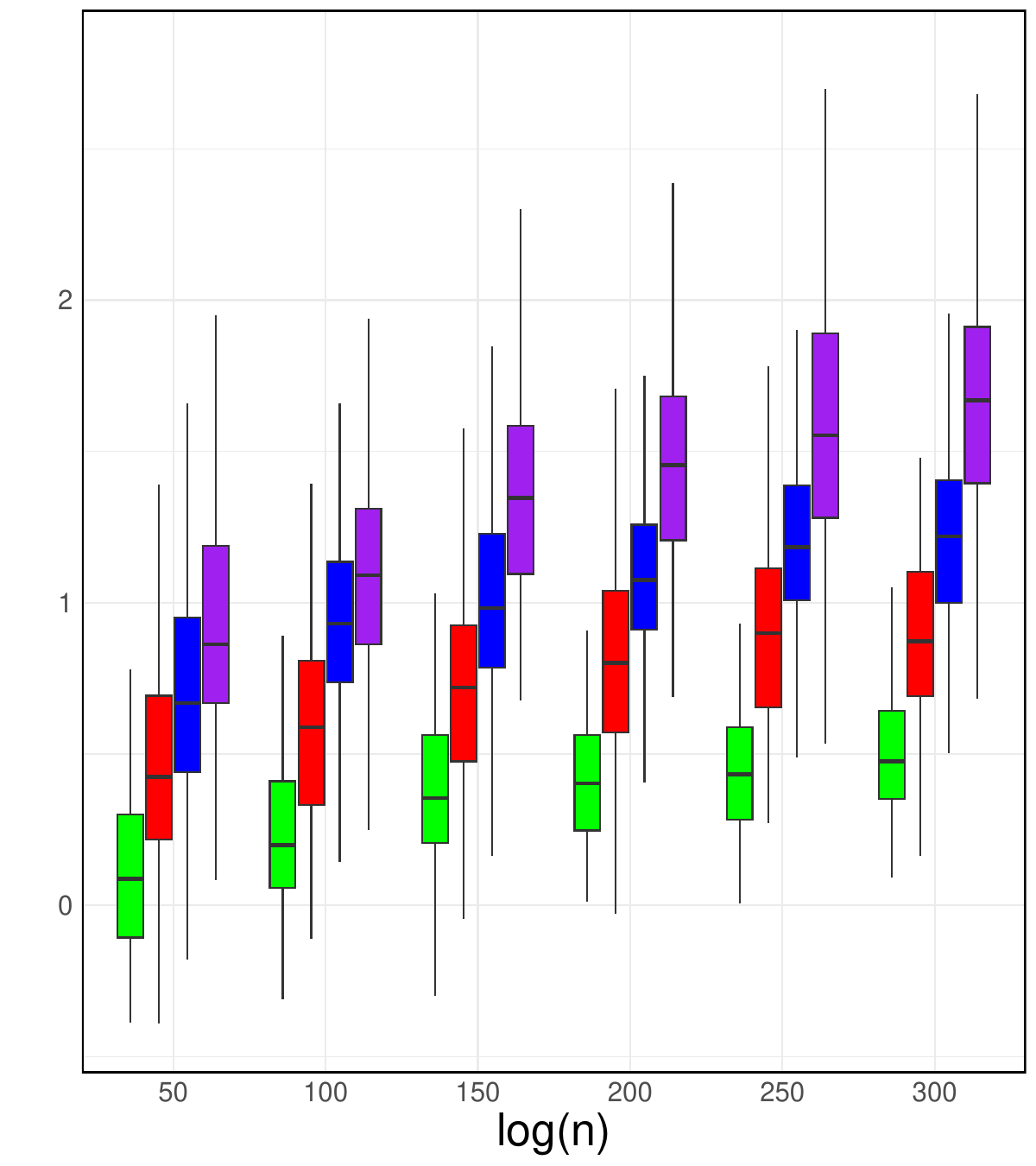}
    \end{minipage}

        \caption{Plots of $\log(\tau)$ chosen from the CV procedure described in \Cref{sec:CV_tune_tau}, with $d = 1$ in the CV measure \eqref{eq:CV_measure}, over 200 simulations against dimension $p$ ($n$ fixed at 200) and sample size $n$ ($p$ fixed at 100), and as the innovation distribution varies (see~\ref{innov_dist_gauss}--\ref{innov_dist_log_norm}). The data are generated as in \ref{item:fac_0}, \ref{eq:tri_diag_A_sparse_var}, (S1). The first row shows the average $\tau$ chosen across simulations, the numbers of each point is the average percentage of data points in absolute values that is less than $\tau$. The second row gives the box plots of the chosen $\tau$ across the simulations.}
    \label{fig:tau_vs_p_and_n_auto}
\end{figure}

\begin{table}[h!t!] 
    \centering
    \caption{Maximum of covariance and autocovariance estimation errors, as measured by $\mathrm{RME}_\ell$~\eqref{eq:RME_lag} with $d = 1$, with different matrix norms from 200 realisations. With data generated from the process \ref{item:fac_1} and \ref{eq:tri_diag_A_sparse_var}, with varying $(n, p)$ and innovation distributions \ref{innov_dist_gauss}--\ref{innov_dist_log_norm}.}
    \begin{tabular}{c cc cc cc cc cc}
        \toprule & 
                 \multicolumn{2}{c}{log-normal} &
                 \multicolumn{2}{c}{$t_{2.1}$} & \multicolumn{2}{c}{$t_{3}$} & \multicolumn{2}{c}{$t_{4}$} & 
                 \multicolumn{2}{c}{Normal}\\
        \cmidrule(l{2pt}r{2pt}){2-3}
        \cmidrule(l{2pt}r{2pt}){4-5}
        \cmidrule(l{2pt}r{2pt}){6-7}
        \cmidrule(l{2pt}r{2pt}){8-9}
        \cmidrule(l{2pt}r{2pt}){10-11}
        $(n,p)$ &  max   & F & max & F & max & F & max & F & max & F\\
        \midrule 
(100,50) & 0.550 & 0.757 & 0.194 & 0.668 & 0.397 & 0.633 & 0.700 & 0.830 & 0.889 & 0.915 \\ 
  (100,100) & 0.486 & 0.807 & 0.105 & 0.594 & 0.337 & 0.713 & 0.723 & 0.892 & 0.853 & 0.905 \\ 
  (200,50) & 0.587 & 0.829 & 0.466 & 0.869 & 0.403 & 0.686 & 0.750 & 0.855 & 0.963 & 0.961 \\ 
  (200,100) & 0.488 & 0.757 & 0.109 & 0.692 & 0.369 & 0.820 & 0.763 & 0.893 & 0.915 & 0.956 \\ 
  (500,100) & 0.570 & 0.818 & 0.210 & 0.824 & 0.404 & 0.796 & 0.751 & 0.954 & 0.983 & 1.000 \\ 
  (500,200) & 0.513 & 0.851 & 0.157 & 0.853 & 0.197 & 0.851 & 0.684 & 0.940 & 0.978 & 0.999 \\ 
        \bottomrule
    \end{tabular}
    \label{tab:auto_fac_var_banded_diag}
\end{table}

\begin{table}[t] 
    \centering
    \caption{Maximum of covariance and autocovariance estimation errors, as measured by $\mathrm{RME}_\ell$~\eqref{eq:RME_lag} with $d = 1$, with different matrix norms from 200 realisations. With data generated from the process \ref{item:fac_1} and \ref{eq:renyi_A}, with varying $(n, p)$ and innovation distributions \ref{innov_dist_gauss}--\ref{innov_dist_log_norm}.}
    \begin{tabular}{c cc cc cc cc cc}
        \toprule & 
                 \multicolumn{2}{c}{log-normal} &
                 \multicolumn{2}{c}{$t_{2.1}$} & \multicolumn{2}{c}{$t_{3}$} & \multicolumn{2}{c}{$t_{4}$} & 
                 \multicolumn{2}{c}{Normal}\\
        \cmidrule(l{2pt}r{2pt}){2-3}
        \cmidrule(l{2pt}r{2pt}){4-5}
        \cmidrule(l{2pt}r{2pt}){6-7}
        \cmidrule(l{2pt}r{2pt}){8-9}
        \cmidrule(l{2pt}r{2pt}){10-11}
        $(n,p)$ &  max   & F & max & F & max & F & max & F & max & F\\
        \midrule 
(100,50) & 0.356 & 0.744 & 0.218 & 0.715 & 0.272 & 0.615 & 0.603 & 0.820 & 1.051 & 0.945 \\ 
  (100,100) & 0.302 & 0.818 & 0.076 & 0.573 & 0.177 & 0.678 & 0.501 & 0.898 & 0.998 & 0.931 \\ 
  (200,50) & 0.468 & 0.817 & 0.265 & 0.796 & 0.233 & 0.629 & 0.606 & 0.842 & 1.095 & 0.968 \\ 
  (200,100) & 0.312 & 0.733 & 0.082 & 0.632 & 0.230 & 0.814 & 0.540 & 0.893 & 1.083 & 0.967 \\ 
  (500,100) & 0.363 & 0.819 & 0.122 & 0.737 & 0.256 & 0.759 & 0.524 & 0.936 & 1.112 & 0.993 \\ 
  (500,200) & 0.333 & 0.863 & 0.122 & 0.855 & 0.103 & 0.831 & 0.429 & 0.949 & 1.100 & 0.990 \\ 
        \bottomrule
    \end{tabular}
    \label{tab:auto_fac_var_renyi_diag}
\end{table}

\clearpage

\subsection{\cred{Computational cost}} \label{sec:comp_cost}

In this subsection, we analyse the computational complexity of the proposed method and compare its empirical running time with that of the competing estimator considered in \Cref{sec:results}.

 The cross-validation procedure we implement, described in Section~\ref{sec:CV_tune_tau}, requires the computation of $2(d + 1)$ sample (auto)covariance matrices, where $d$ denotes the order of the VAR process, over a grid of candidate truncation parameter values of length $J$. Hence, the computational cost scales, with regards to the grid length $J$, dimension $p$ and sample size $n$, as $O(dJnp^2)$. For comparison, the $\ell_1$-regularised Huber regression estimator of \citet{Wang2021}, implemented in the R package \texttt{adaHuber} \citep{adaHuber} and adopted for comparative simulations in \Cref{sec:results}, has per-iteration tuning cost of order $O(np)$.
 
We numerically investigate the run time of the two methods in \Cref{fig:comp_time} , which demonstrates that our method (with the CV-based truncation parameter selection) is faster. 
In the top left panel of \Cref{fig:comp_time}, reporting the results from applying the two methods to the data from the banded VAR setting (\ref{item:fac_0} and \ref{eq:tri_diag_A_sparse_var} in Section~\ref{sec:simulations}), we observe a drop in the run time as the dimension increases from $p = 150$ to $p = 200$. Upon close investigation, this turns out to be an artefact due to the implementation of the function \texttt{cv.glmnet} in the R package \texttt{glmnet} \citep{Friedman2010} which we use to perform the VAR estimation step in Equation~(8). In this function, different approaches are employed to generate the penalty parameter path according to the dimension of the data. The plot in the bottom row of \Cref{fig:comp_time} illustrates that when the penalty parameter path is fixed for all $p$ with that chosen by \texttt{cv.glmnet} for $p =200$, this drop in time for increasing dimension of dataset is no longer visible.

\begin{figure}[h!t!]
    \centering
    \begin{minipage}[b]{0.49\textwidth}
        \includegraphics[width=\textwidth]{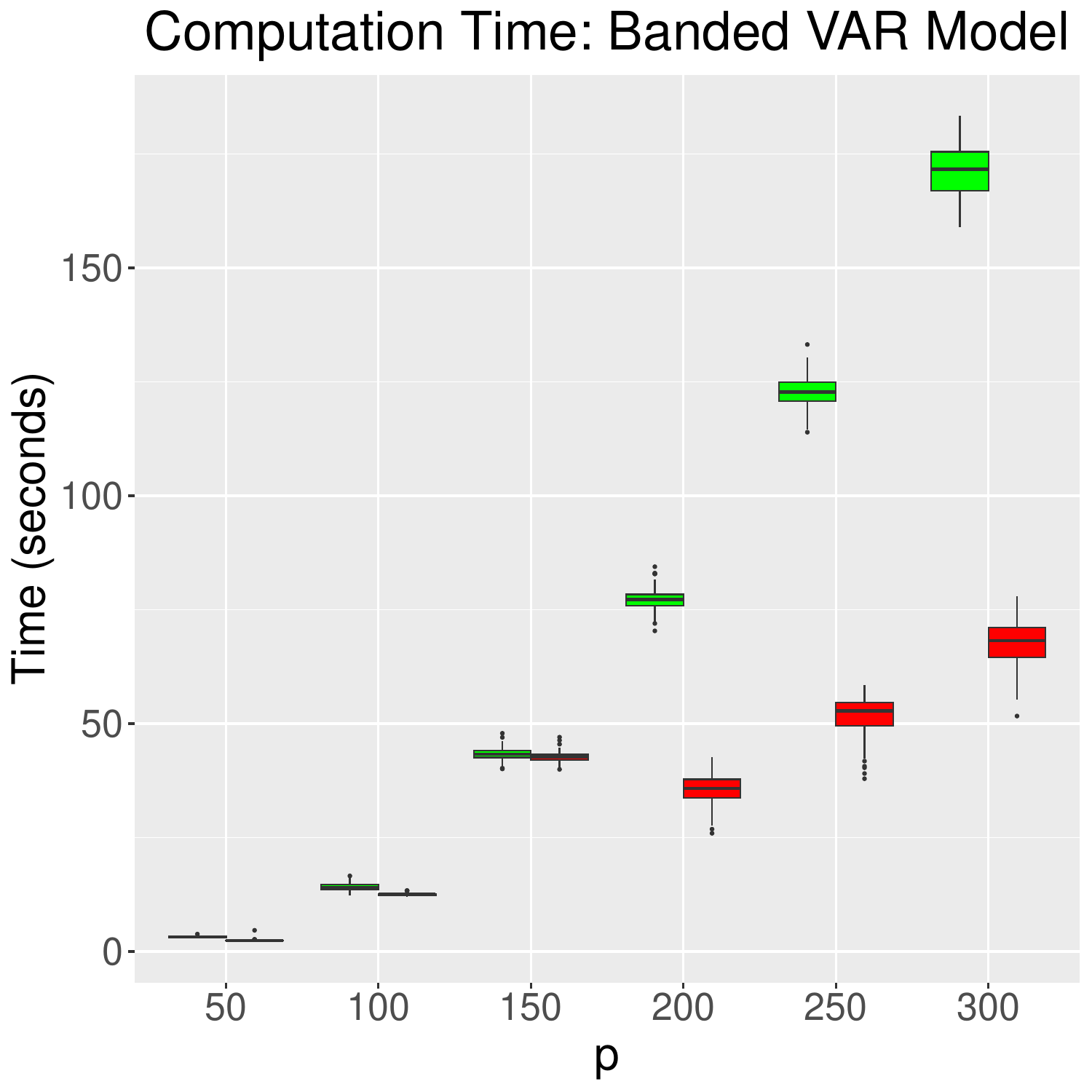}
    \end{minipage}
    \begin{minipage}[b]{0.49\textwidth}
        \includegraphics[width=\textwidth]{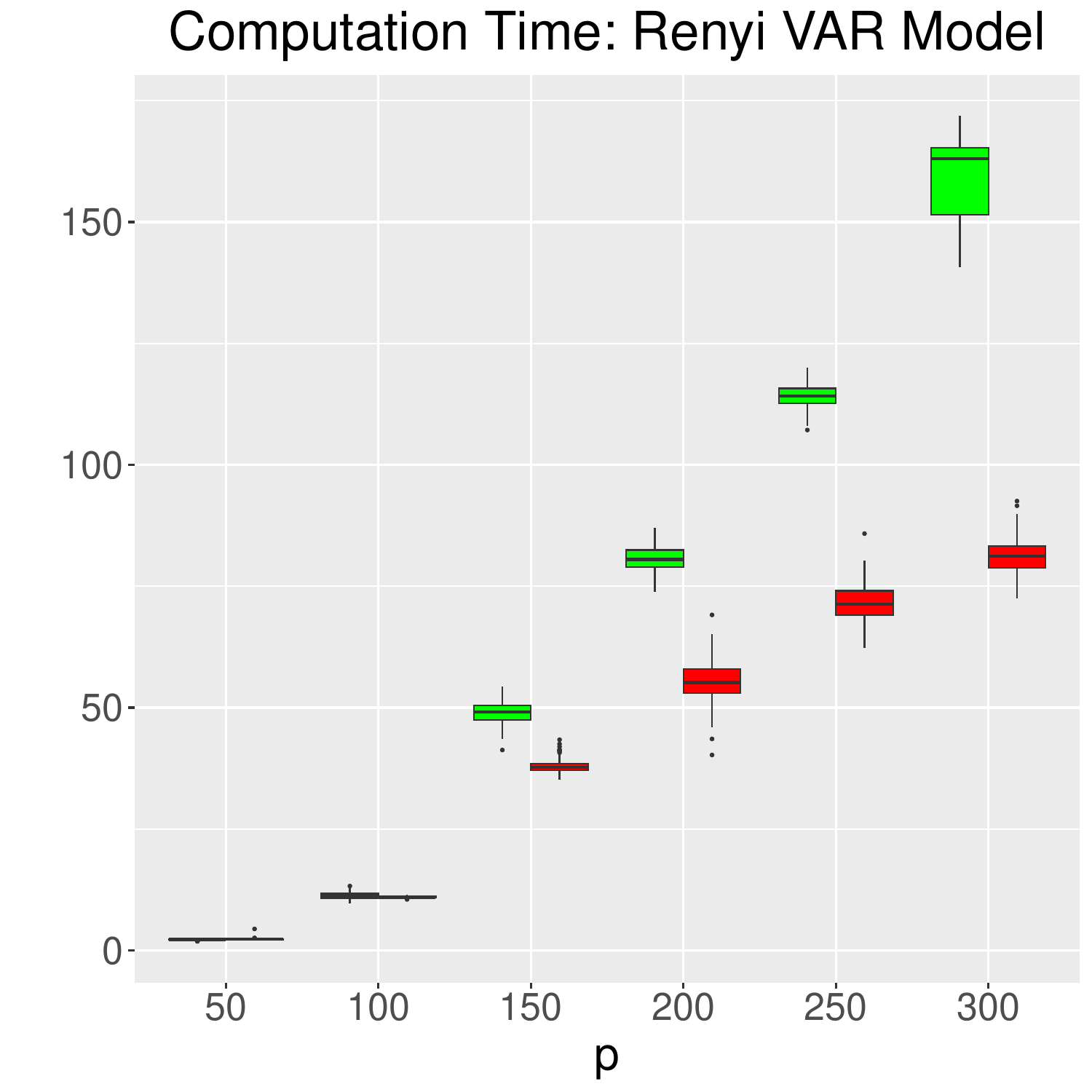}
    \end{minipage}
    
    \begin{minipage}[t]{\textwidth}
        \centering
        \includegraphics[width=8cm]{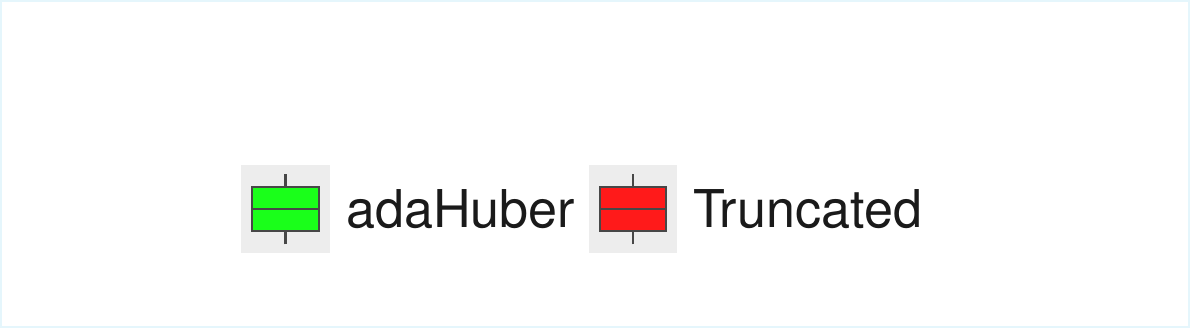} 
    \end{minipage}
    \begin{minipage}[b]{0.5\textwidth}
    \centering
        \includegraphics[width=\textwidth]{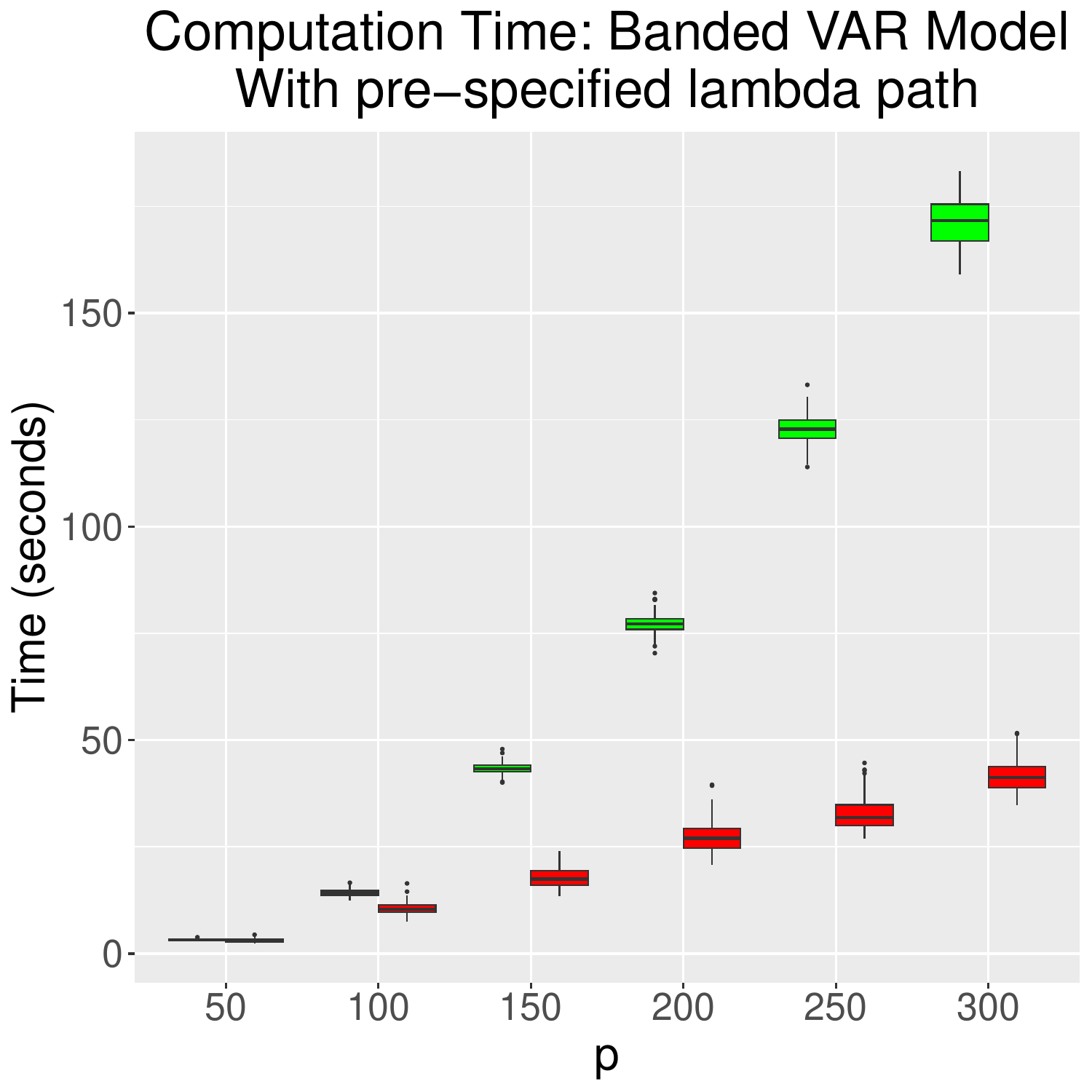}
    \end{minipage}
    \caption{Box plots of run time in seconds from $200$ realisations generated with $n = 200$ and varying $p \in \{50, 100, 150, 200, 250, 300\}$, returned by our proposed method (truncated) and the Huber loss-based method (adaHuber). In the top left and bottom panels, the data are generated as in \ref{item:fac_0} and \ref{eq:tri_diag_A_sparse_var} and in the top right panel, as in \ref{item:fac_0} and \ref{eq:renyi_A}.
    In the bottom panel, for both methods, we supply the pre-specified penalty parameter paths to remove any artefact from the implementation of the function \texttt{cv.glmnet}.}
    \label{fig:comp_time}
\end{figure}

\clearpage

\subsection{\cred{Forecasting performance on simulated data}} \label{sec:forc_simul}

In this subsection, we investigate the forecasting performance of our method with simulated data. Previously, in \Cref{sec:real-data} , we had analysed the forecasting performance with a real dataset, we now generate data with the same factor-adjusted VAR process of order $1$, as considered in \Cref{sec:simulations}, namely:
\begin{gather*}
    \bX_t = \bLambda \bF_t + \bxii_t
    \text{ \ with \ }  
    \bxii_t = \mathbf{A}\bxii_{t-1} + \boldsymbol{\varepsilon}_t \, .
\end{gather*}
The idiosyncratic process is generated according to the banded VAR model, described in \ref{eq:tri_diag_A_sparse_var}, and the loading matrix $\bLambda$ and factor process $\bF_t$ are generated as described in \ref{item:fac_1}. We simulate data with varying $(n, p)$ and innovation distributions \ref{innov_dist_gauss}--\ref{innov_dist_log_norm}, and for each scenario we generate 200 realisations.

As described in \Cref{sec:real-data}, we generate forecasts by estimating the best linear predictor of the process. Specifically, the best linear predictor of the common (factor-driven) and idiosyncratic (VAR) components at time stamp $t + h$ given the preceding $T$ observations, 
are given by
\begin{gather}
    \bchi_{t + h \vert T} = \boldsymbol{\Gamma}_{\bchi}(h) \mathbf{E}_{\bchi} \mathbf{M}_{\bchi}^{-1} \mathbf{E}_{\bchi}^{\top} \bchi_t \label{eq:com_forc} \\
    \text{and} \quad \boldsymbol{\xi}_{t + h \vert T} = \sum_{\ell=1}^{\min(d,h)} \mathbf{A}_{\ell} \boldsymbol{\xi}_{t+h-\ell} \, . \label{eq:idio_forc}
\end{gather}

For this simulation exercise, we focus on the one-step ahead forecast error of $\bX_{n + 1}$,  given data generated from the preceding $n$ observations. We compute the estimates of the best linear predictors \eqref{eq:com_forc} and \eqref{eq:idio_forc}, in the same manner as in \Cref{sec:real-data}, with the estimate of Equation \eqref{eq:com_forc}, given by 
\begin{align} \label{eq:com_pred_no_trunc}
    \widehat{\bchi}_{n+1 \mid n}(\tau) = \widehat{\bGamma}_{\bchi}(\tau,1) \widehat{\mathbf{E}}_{\bx}(\tau) (\widehat{\mathbf{M}}_{\bx}(\tau))^{-1} (\widehat{\mathbf{E}}_{\bx}(\tau))^{\top} \bX_n \, ,
\end{align}
where $\widehat{\bGamma}_{\bchi}(\tau,h) = n^{-1} \sum_{t = h+1}^n \widehat{\bchi}_{t}(\tau)   \widehat{\bchi}_{t-h}(\tau)^\top$ denotes an estimator of the autocovariance of $\bchi_t$ at lag $h$, the matrix $\widehat{\mathbf{M}}_{\bx}(\tau)$ contains the leading $r$ eigenvalues of $\widehat{\bGamma}_{\bx}(\tau, 0)$ on its diagonal, $\widehat{\mathbf{E}}_{\bx}(\tau)$ contains the corresponding eigenvectors, and
$\widehat{\bchi}_{t}(\tau) = \widehat{\mathbf{E}}_{\bx}(\tau) ( \widehat{\mathbf{E}}_{\bx}(\tau) )^\top \bX_t(\tau)$. For the idiosyncratic component, we estimate $\bxii_{n + 1 \vert n} = \mathbf{A}\bxii_{n}$ by
\begin{align}
\label{eq:var_forecast}
\widehat{\boldsymbol{\xi}}_{n+1 \mid n}(\tau)= \widehat{\mathbb{A}}(\tau) \widehat{\boldsymbol{\xi}}_{n}(\tau), \; \text{ \ where \ } \widehat{\boldsymbol{\xi}}_{n}(\tau) = \bX_{n} - \widehat{\bchi}_{n}(\tau) \, .
\end{align}
Then, the combined forecast is given by
\begin{gather*}
 \widehat{\bX}_{n+1 \mid n}(\tau) = \widehat{\bchi}_{n+1 \mid n}(\tau) + \widehat{\boldsymbol{\xi}}_{n+1 \mid n}(\tau) \, .
\end{gather*}
To assess the forecasting performance in the simulated setting, we use the following metric:
\begin{gather}
\frac{\left| \widehat{\bX}_{n+1 \mid n}(\tau) - \bX_{n+1 \mid n} \right|_{\infty}}
{\left| \bX_{n+1 \mid n} \right|_{\infty}} \, , \label{eq:forcast_error}
\end{gather}
where $\bX_{n+1 \mid n} = \bchi_{n + 1 \vert n} + \bxii_{n + 1 \vert n}$.
The results are reported in \Cref{tab:forecast_errors_trunc}, where we present the average and the standard deviations of the errors from the $200$ simulations. The findings indicate that for heavy-tailed innovation distributions, the truncation-based method outperforms the standard approach without any truncation (i.e.\ when $\tau = \infty$) in the majority of cases. For the Gaussian case, where the standard approach is expected to perform well, there is only a small loss in efficiency introduced by data truncation.

\begin{table}[H]
\centering
\setlength{\tabcolsep}{3pt}
\caption{Forecasting errors measured as in Equation \eqref{eq:forcast_error} over 200 realisations. We report the average and the standard deviations. The data are generated as in \ref{item:fac_1} and \ref{eq:tri_diag_A_sparse_var}, with varying $n$, $p$ and the innovation distributions (see \ref{innov_dist_gauss}--\ref{innov_dist_log_norm}).}
{
\begin{tabular}{cccccccccc}
\toprule
 &  &
\multicolumn{2}{c}{log-normal} &
\multicolumn{2}{c}{$t_{2.1}$} &
\multicolumn{2}{c}{$t_{3}$} &
\multicolumn{2}{c}{Normal} \\
\cmidrule(lr){3-4}
\cmidrule(lr){5-6}
\cmidrule(lr){7-8}
\cmidrule(lr){9-10}
$(n,p)$ & Method
& Mean & SD
& Mean & SD
& Mean & SD
& Mean & SD \\
\midrule

\multirow{2}{*}{(100,50)}
& Truncated   & \textbf{0.51}  & (0.13) & \textbf{0.54} & (0.19) & 0.49 & (0.13) & 0.50 & (0.12) \\  
& NotTruncated & 0.52 & (0.15) & 0.56 & (0.20) & \textbf{0.49} & (0.18) & \textbf{0.47}  & (0.12) \\[4pt]

\multirow{2}{*}{(100,100)}
& Truncated   & \textbf{0.48} & (0.14) & \textbf{0.50} & (0.16) & \textbf{0.49} & (0.12) & 0.49 & (0.11) \\
& NotTruncated & 0.50 & (0.14) & 0.57 & (0.21) & 0.51 & (0.17) & \textbf{0.43} & (0.09) \\[4pt]

\multirow{2}{*}{(200,50)}
& Truncated   & \textbf{0.35} & (0.11) & \textbf{0.38}  & (0.14) & \textbf{0.35} & (0.10) & 0.34 & (0.08) \\
& NotTruncated & 0.36 & (0.11) & 0.43 & (0.19) & 0.36 & (0.13) & \textbf{0.33} & (0.08)\\[4pt]

\multirow{2}{*}{(200,100)}
& Truncated   & \textbf{0.31} & (0.09) & \textbf{0.35} & (0.16) & \textbf{0.31}  & (0.09) & 0.33 & (0.08) \\
& NotTruncated & 0.33 & (0.11) & 0.45  & (0.21) & 0.33 & (0.11) & \textbf{0.30} & (0.07)  \\[4pt]

\multirow{2}{*}{(500,100)}
& Truncated   & \textbf{0.20} & (0.05) & \textbf{0.23} & (0.10) & \textbf{0.22}  & (0.06) & 0.22 & (0.04) \\
& NotTruncated & 0.20 & (0.06) & 0.31  & (0.19) & 0.22 & (0.08) & \textbf{0.21}  & (0.04) \\[4pt]

\multirow{2}{*}{(500,200)}
& Truncated   & \textbf{0.18} & (0.04) & \textbf{0.20}  & (0.10) & \textbf{0.18}  & (0.05) & 0.20  & (0.04) \\
& NotTruncated & 0.20  & (0.06) & 0.37  & (0.23) & 0.22 & (0.10) & \textbf{0.18} & (0.04) \\[4pt]

\bottomrule
\end{tabular}}
\label{tab:forecast_errors_trunc}
\end{table}

\clearpage

\section{Proofs}

We define here notation that is used within the proofs.
Throughout, we denote the autocovariance matrix of lag~$h$, $h \ge 0$, of the common component $\bchi_t$, and truncated data $\bX_t(\tau)$ as
\begin{align*}
\bGamma_{\bchi}(h) = \frac{1}{n-h} \sum_{t = h + 1}^n \mathbb{E}\left( \bchi_t \bchi_{t - h}^\top \right) \quad \text{and} \quad
\bGamma_{\bx}(\tau,h) = \frac{1}{n-h} \sum_{t = h + 1}^n \mathbb{E}\left( \bX_t(\tau) \bX_{t - h}^\top(\tau) \right) \, .
\end{align*}
Furthermore, for the estimated common component $\widehat{\bchi}_t$ and estimated idiosyncratic component $\widehat{\bxii}_t$ we denote their sample autocovariance matrices by
\begin{align*}
\widehat{\bGamma}_{\bchi}(h) = \frac{1}{n-h} \sum_{t = h + 1}^n  \widehat{\bchi}_t \widehat{\bchi}_{t - h}^\top \quad \text{and} \quad
\widehat{\bGamma}_{\bxii}(h) = \frac{1}{n-h} \sum_{t = h + 1}^n  \widehat{\bxii}_t \widehat{\bxii}_{t - h}^\top \, .
\end{align*}

We also make use of big $O$ and little $o$ notation, which we define here. Given two sequences $\{a_n\}$ and $\{b_n\}$, we write $a_n = O(b_n)$ if there exists a finite constant $C > 0$ and $N \in \mathbb{N}_0 = \mathbb{N} \cup \{0\}$ such that $|a_n|\,|b_n|^{-1} \leq C$ for all $n \geq N$.  
We write \(a_n = o(b_n)\) if \(|a_n|\,|b_n|^{-1} \to 0\) as \(n \to \infty\). The $i$-th row of the loading matrix $\bLambda$ is defined by $\blambda_i$, that is we can write: $\bLambda = \left[\blambda_1, \dots, \blambda_p\right]^\top$, and the $(i,j)$-th entry is given by $\lambda_{i,j}$.

\subsection{Proof of \Cref{prop:peligrad_inequality_trunc_autocov_error}}
\label{A:proof_trunc_cov}

\label{A:peligrad_rate}

In this proof, we use the concentration inequality given by \citet[][Theorem 2]{Peligrad2009} to compute the estimation convergence rate of the truncated autocovariance estimator. We select our choice of truncation parameter $\tau$, by matching the rate of the variance component and the upper bound of the bias component.

\begin{theorem}[Theorem 2 from \citet{Peligrad2009}] \label{thm:peligrad_bernstein}
Let $(X_j)_{j\geq1}$ be a sequence of centred real-valued random variables.
Suppose that the sequence is $\alpha$-mixing with the coefficients decaying exponentially, $\alpha(\ell) \lesssim \exp(-2c\ell)$, for some $c \in \left(0, \infty\right)$ and that there exists a positive constant $K$ such
that $\sup_{i\geq1} \operatorname{esssup} |X_i| \leq K$. For the sum $S_n = \sum_{k=1}^{n}X_k$, there is a constant $C$ depending only on $c$ such that for all $n \geq 2$,
$$
\mathbb{P}\left(\left|S_n\right| \geq x\right) \leq \exp \left(-\frac{C x^2}{v^2 n+K^2+x K\log^2(n)}\right) \, ,
$$
where $v^2$ is defined as
\begin{align} \label{eq:peligrad_v}
    v^2=\sup _{i>0}\left(\operatorname{Var}\left(X_i\right)+2 \sum_{j>i}\left|\operatorname{Cov}\left(X_i, X_j\right)\right|\right) \, .
\end{align}
\end{theorem}
\vspace{7pt}
To first compute the error rate of the variance component, we apply \Cref{thm:peligrad_bernstein} with the sequence of centred random variables $(Z_t)_{t\geq1}$, defined as $Z_t = X_{it}(\tau) X_{jt-h}(\tau) - \mathbb{E}[X_{it}(\tau) X_{jt-h}(\tau)]$, which are upper bounded by
\begin{align}
    |Z_t| \leq |X_{it}(\tau) X_{jt-h}(\tau)| + |\mathbb{E}\left[X_{it}(\tau) X_{jt-h}(\tau)\right]| \leq 2\tau^2 \, . \nonumber
\end{align}
We now look to obtain an upper bound on the covariances in the definition of $v^2$ in \eqref{eq:peligrad_v}. To do so, we use the following proposition, which provides a bound on the covariance of two random variables in terms of their $\alpha$-mixing dependency and moments.

\begin{proposition}[Proposition 2.5 (i) from \citet{Yao_fan_book}]\label{prop:alpha_mixing_cov_bound}

If $\mathbb{E}\left\{|X|^p+|Y|^q\right\}<\infty$ for some $p, q \geq 1$ and $1 / p+ 1 / q<1$, it holds that
\begin{align} \label{eq:yan_fao_cov_bound}
|\operatorname{Cov}(X, Y)| \leq 8 \alpha^{1 / r}\left\{\mathbb{E}|X|^p\right\}^{1 / p}\left\{\mathbb{E}|Y|^q\right\}^{1 / q},
\end{align}
where $\alpha$ denotes the $\alpha$-measure of dependence between random variables $X$ and $Y$ and $r=(1-1 / p-1 / q)^{-1}$.
\end{proposition}

Firstly, note that
\begin{align}
    \mathbb{E}[|Z_t|^2] &= \mathbb{E}[(X_{it}(\tau) X_{jt-h}(\tau))^2] - \mathbb{E}\left[X_{it}(\tau) X_{jt-h}(\tau)\right]^2 \nonumber\\
    &\leq \mathbb{E}[(X_{it}(\tau) X_{jt-h}(\tau))^2] = \mathbb{E}[\left|X_{it}(\tau) X_{jt-h}(\tau)\right|^{1-\epsilon} \left|X_{it}(\tau) X_{jt-h}(\tau)\right|^{1+\epsilon}] \nonumber\\
    &\leq \tau^{2-2\epsilon}\cred{\widetilde{M}_{\epsilon}} \, , \nonumber
\end{align}
which follows from that $\left|X_{it}(\tau) X_{jt-h}(\tau)\right| \leq \tau^2$, and the Cauchy-Schwarz inequality with \Cref{lem:peligrad_holds}. 
Thus, we obtain
\begin{align}
    \mathbb{E}[|Z_t|^2]^{\frac{1}{2}} &\leq \tau^{1-\epsilon}\sqrt{\cred{\widetilde{M}_{\epsilon}}}.  \label{eq:2bound}
\end{align}
Now, for some $\delta > 0$, we have
\begin{align*}
     &\mathbb{E}[|Z_t|^{2+2\delta}] = \mathbb{E}[|Z_t|^{2\delta}|Z_t|^{2}] \leq 2^{2\delta}\tau^{4\delta} \tau^{2-2\epsilon} \cred{\widetilde{M}_{\epsilon}}.
\end{align*}
Setting $\delta = 1/\log(n)$, and using the choice of $\tau$ in \eqref{eq:tau_prop} for which we have $\log(\tau) \lesssim \log(n)$, we obtain
\begin{align}
     \mathbb{E}[|Z_t|^{2+2\delta}]^{\frac{1}{2+2\delta}} &\lesssim \tau^{\frac{4\delta}{2 + 2\delta}} \tau^{\frac{2-2\epsilon}{2+2\delta}} \sqrt{\cred{\widetilde{M}_{\epsilon}}} \nonumber \\
     &\lesssim \tau^{1-\epsilon}\sqrt{\cred{\widetilde{M}_{\epsilon}}} \, . \label{eq:2plus2deltabound}
\end{align}
Then, combining the two upper bounds \eqref{eq:2plus2deltabound} and \eqref{eq:2bound}, we have
\begin{align}
    \mathbb{E}[|Z_t|^2]^{\frac{1}{2}}\mathbb{E}[|Z_t|^{2+2\delta}]^{\frac{1}{2+2\delta}} \lesssim \tau^{2-2\epsilon} \cred{\widetilde{M}_{\epsilon}} \, . \label{eq:yan_fao_cov_bound_moments}
\end{align}
Now, plugging \eqref{eq:yan_fao_cov_bound_moments} into \eqref{eq:yan_fao_cov_bound} in \Cref{prop:alpha_mixing_cov_bound} with
$p=2$ and $q=2+2\delta$, and by  \Cref{lem:peligrad_holds}~(b), we have the following bound on $v^2$ that appears in \eqref{eq:peligrad_v}:
\begin{align}
    \left|\operatorname{Cov}\left(Z_t, Z_{t+n}\right)\right| &\lesssim {\cred{\widetilde{M}_{\epsilon}}}\tau^{2-2\epsilon} \exp{\left( -\frac{2cn}{r}\right)}, \nonumber\\
    \operatorname{Var}(Z_t) + 2\sum_{k>t}\left|\operatorname{Cov}\left(Z_t, Z_k\right)\right| &\lesssim {\cred{\widetilde{M}_{\epsilon}}} \tau^{2-2\epsilon}
 \sum_{l = 1}^{n} \exp{\left( -\frac{2cl}{r}\right)} \lesssim { {\cred{\widetilde{M}_{\epsilon}}} \tau^{2-2\epsilon} \log(n)} \, , 
 \nonumber
\end{align}
where $r = 2\log(n) + 2$. Hence we obtain the bound $v^2 \lesssim {\cred{\widetilde{M}_{\epsilon}}} \tau^{2-2\epsilon} \log(n)$.

Returning to the Bernstein type inequality of \Cref{thm:peligrad_bernstein}, by the union bound we have
\begin{align}
\mathbb{P}\left(|\widehat{\bGamma}_\bx(\tau,h) - \bGamma_\bx(\tau,h)|_\infty \geq x \right) &\leq 
p^2 \max_{1 \le i, j \le p} \mathbb{P}\left(\frac{1}{n}\left|\sum_{t=1}^{n}X_{it}(\tau) X_{jt-h}(\tau) - \mathbb{E}[X_{it}(\tau) X_{jt-h}(\tau)]\right| \geq x \right) \nonumber \\
&\leq \exp \left(- \frac{Cn^2x^2}{v^2n + K^2 + xnK\log^2(n)} \, + \, 2\log(p)\right) \, , \nonumber
\end{align}
with $K = 2\tau^2$. 
In identifying $x$ to fulfil
\begin{gather}
    n^2x^2 \gtrsim \max\{n \log(n) \tau^{2-2\epsilon}, \, \tau^4, \, xn\tau^2\log^2(n)\}\log(p) \nonumber\\ 
    \iff  x \gtrsim \max\left\{\left( \sqrt{\frac{\cred{\widetilde{M}_{\epsilon}}\log(n)}{n}} \tau^{1-\epsilon}, \, \frac{\tau^2}{n}\right)\sqrt{\log(p)}, \, \frac{\tau^2\log^2(n)}{n}\log(p)\right\} \, , \nonumber
\end{gather}
we note that the choice of 
\begin{gather}
    \tau \asymp \left(\frac{n \cred{\widetilde{M}_{\epsilon}}}{\log^2(n)\log(p)}\right)^{\frac{1}{2+2\epsilon}} \, , 
    \nonumber
\end{gather}
leads to
\begin{align} 
    |\widehat{\bGamma}_\bx(\tau,h) - \bGamma_\bx(\tau,h)|_\infty = O_P\left(\frac{\log(p) \log^2(n) }{n} \tau^{2}\right) = O_P\left(\left(\frac{\log(p) \log^2(n) \cred{\widetilde{M}_{\epsilon}}^\frac{1}{\epsilon} }{n}\right)^{\frac{\epsilon}{1+\epsilon}}\right). \nonumber
\end{align}
Also, with such a choice of $\tau$, \Cref{lem:bias_bound} gives us 
\begin{align}
|\bGamma_\bx(\tau,h) - \bGamma_\bx(h)|_\infty \lesssim \frac{\cred{\widetilde{M}_{\epsilon}}}{\tau^{2\epsilon}} = O\left(\left(\frac{\log(p) \log^2(n) \cred{\widetilde{M}_{\epsilon}}^\frac{1}{\epsilon} }{n}\right)^{\frac{\epsilon}{1+\epsilon}}\right). \nonumber
\end{align}
Altogether,  by the triangle inequality, we derive that
\begin{gather} 
\label{eq:trunc_rate_denote}
|\widehat{\bGamma}_\bx(\tau,h) - \bGamma_\bx(h)|_\infty = O_P\left(\left(\frac{\log(p) \log^2(n) \cred{\widetilde{M}_{\epsilon}}^\frac{1}{\epsilon} }{n}\right)^{\frac{\epsilon}{1+\epsilon}}\right) \, . 
\end{gather}

\subsubsection{Preliminary result for \Cref{prop:peligrad_inequality_trunc_autocov_error}}

The following results hold under the conditions made in \Cref{prop:peligrad_inequality_trunc_autocov_error}.

\begin{lemma} \label{lem:peligrad_holds}
Under the model \eqref{eq:fac_adj_model_def} and \Cref{assump:fac_adj_fac_loading,assum:tail,assump:fac_adj_fac_idio,assump_fac_adj:8}, the process $\bX_t$ has the following properties. 

\begin{enumerate}[topsep=0pt, label = (\alph*)]
    \item  There exists some constant $\cred{\widetilde{M}_{\epsilon}} > 0$, such that $\max_{1 \le i \le p} \mathbb{E}[|X_{it}|^{2 + 2 \epsilon}] \leq \cred{\widetilde{M}_{\epsilon}}$ for the $\epsilon \in (0,1]$ chosen in \Cref{assum:tail}.
    \item  The bivariate process $\{(X_{it}, X_{jt})\}$ is $\alpha$-mixing, with the mixing coefficients decaying exponentially such that $\alpha^{\mspace{-1mu} \mathbf{\bX}}\mspace{-2mu}(m) \leq \exp(-2cm)$, $c > 0 $, for all $1 \leq i,j \leq p$. 
\end{enumerate}
\end{lemma}

\begin{proof}
By recursively using Minkowski's inequality we have that for all $1 \leq i \leq p$: 
\begin{gather*}
\|X_{it}\|_{2+2\epsilon} \leq \lambda_{i1}\|F_{1t}\|_{2+2\epsilon} + \dots + \lambda_{ir}\|F_{rt}\|_{2+2\epsilon} + \|\xi_{it}\|_{2+2\epsilon} \, . 
\end{gather*} 
Then, by \Cref{assump:fac_adj_fac_loading}.\ref{finite_loadings}, \Cref{assum:tail}, and that $r$ is fixed, we have
\begin{gather} \label{eq:X_bounded}
 \max_{1 \le i \le p}\mathbb{E}\left[|X_{it}|^{2+2\epsilon}\right] \leq \cred{\widetilde{M}_{\epsilon}} \; .
\end{gather}
Now, by \Cref{assump:fac_adj_fac_idio}, \Cref{assump_fac_adj:8}, and \citet[][Theorem 6.2]{bradley2007introduction}, and that the process generated by a measurable function applied to an $\alpha$-mixing process is still $\alpha$-mixing in the same sense, we have that the bivariate process $\{(X_{it}, X_{jt})\}$ is $\alpha$-mixing with the same rate $\alpha^{\mspace{-1mu} \mathbf{\bX}}\mspace{-2mu}(m) \leq \exp(-2cn)$, $c > 0 $.
\end{proof}

\begin{lemma} \label{lem:bias_bound}
Under \Cref{assump:fac_adj_fac_loading}.\ref{finite_loadings}, and \Cref{assum:tail}, we have 
\begin{gather*} 
|\bGamma_\bx(\tau,h) - \bGamma_\bx(h)|_\infty \leq \frac{2 \cred{\widetilde{M}_{\epsilon}}}{\tau^{2\epsilon}}  \, ,
\label{eq:trunc_bias}
\end{gather*}
where $\cred{\widetilde{M}_{\epsilon}}$ is given in \Cref{lem:peligrad_holds}~(a).
\end{lemma}
\begin{proof} 
For any $1 \leq i,j \leq p$, the $(i,j)$-th element of $|\bGamma_\bx(\tau,h) - \bGamma_\bx(h)|$ satisfies
    \begin{align*}
        &\phantom{\leq} \left|\mathbb{E}\left[X_{it} X_{jt-h}\right] - \mathbb{E}\left[X_{it}(\tau) X_{jt-h}(\tau)\right]\right| \\
        &\leq \mathbb{E}\left[|X_{it} X_{jt-h}| \left(\mathbb{I}_{\{|X_{it}| > \tau\}} + \mathbb{I}_{\{|X_{jt-h}| > \tau\}}\right)\right] \\
        &= \mathbb{E}\left[|X_{it} X_{jt-h}| \mathbb{I}_{\{|X_{it}| > \tau\}} \right] + \mathbb{E}\left[|X_{it} X_{jt-h}| \mathbb{I}_{\{|X_{jt-h}| > \tau\}} \right] \\
        &\overset{(1)}{\leq} \mathbb{E}\left[|X_{it} X_{jt-h}|^{1+\epsilon}\right]^{\frac{1}{1+\epsilon}}\mathbb{E}\left[{\mathbb{I}_{\{|X_{it}| > \tau\}}}^{\frac{1+\epsilon}{\epsilon}}\right]^{\frac{\epsilon}{1+\epsilon}} + \mathbb{E}\left[|X_{it} X_{jt-h}|^{1+\epsilon}\right]^{\frac{1}{1+\epsilon}}\mathbb{E}\left[{\mathbb{I}_{\{|X_{jt-h}| > \tau\}}}^{\frac{1+\epsilon}{\epsilon}}\right]^{\frac{\epsilon}{1+\epsilon}} \\
        &\leq \mathbb{E}\left[|X_{it} X_{jt-h}|^{1+\epsilon}\right]^{\frac{1}{1+\epsilon}} \left[\mathbb{P}\left(|X_{it}| \geq \tau \right)^{\frac{\epsilon}{1+\epsilon}} + \mathbb{P}\left(|X_{jt-h}| \geq \tau \right)^{\frac{\epsilon}{1+\epsilon}}\right] \\
        &\leq  \left( \mathbb{E}\left[|X_{it}|^{2+2\epsilon}\right] \mathbb{E}\left[|X_{jt-h}|^{2+2\epsilon}\right]\right)^{\frac{1}{2+2\epsilon}} \left[\mathbb{P}\left(|X_{it}| \geq \tau \right)^{\frac{\epsilon}{1+\epsilon}} + \mathbb{P}\left(|X_{jt-h}| \geq \tau \right)^{\frac{\epsilon}{1+\epsilon}}\right] \\
        &\overset{(2)}{\leq} \cred{\widetilde{M}_{\epsilon}}^{\frac{1}{1+\epsilon}} \left[\left(\frac{\mathbb{E}\left[|X_{it}|^{2+2\epsilon}\right]}{\tau^{2+2\epsilon}}\right)^{\frac{\epsilon}{1+\epsilon}} + \left(\frac{\mathbb{E}\left[|X_{jt-h}|^{2+2\epsilon}\right]}{\tau^{2+2\epsilon}}\right)^{\frac{\epsilon}{1+\epsilon}} \right]\\
        &\leq \frac{2 \cred{\widetilde{M}_{\epsilon}}}{\tau^{2\epsilon}} \, ,
    \end{align*}
    where step $(1)$ is by Holder's inequality, and $(2)$ is by \Cref{lem:peligrad_holds} and Markov's inequality.
\end{proof}

\subsection{Proof of \Cref{prop:fac_adj_idio_rate}} \label{sec:implementation_theory}

The proof of \Cref{prop:fac_adj_idio_rate}  follow the proofs of Remark 3.2 from \citet[][]{Barigozzi2024}. We assume the model definition \eqref{eq:fac_adj_model_def}, and \Cref{assump:VAR,assump:fac_adj_fac_idio,assump_fac_adj:8,assump:fac_adj_fac_loading}. For the rest of the proofs, for ease of notation, we define the rate with
\begin{align*}
    \psi_{n,p} = \left(\frac{\log(p) \log^2(n) \cred{\widetilde{M}_{\epsilon}}^\frac{1}{\epsilon} }{n}\right)^{\frac{\epsilon}{1+\epsilon}} \, .
\end{align*}
We denote the eigenvalues and eigenvectors of $\widehat{\bGamma}_\bx(\tau)$ by
$(\hat{\mu}_{\bx,j}, \widehat{\mathbf{e}}_{\bx ,j}), \, j \ge 1$, and $\widehat{e}_{\bx ,ij}$ the $i$-th element of $\widehat{\mathbf{e}}_{\bx ,j}$. 
In addition, for the following proofs we drop the $(\tau)$ notation from the matrices $\widehat{\mathbf{E}}_{\bx}(\tau)$, and  $\widehat{\mathbf{M}}_{\bx}(\tau)$ to simplify notation, recall these matrices are defined as $\widehat{\mathbf{E}}_{\bx} = [\widehat{\mathbf{e}}_{\bx ,j}, \, 1 \le j \le r]$
and $\widehat{\mathbf{M}}_{\bx} = \text{diag}(\hat{\mu}_{\bx,j}, \, 1 \le j \le r)$. We also define $\mathbf{E}_{\bchi}$, $\mathbf{M}_{\bchi}$,
with $(\mu_{{\bchi},j}, \mathbf{e}_{\bchi ,j})$, the eigenvalues and eigenvectors of $\bGamma_{\bchi}(0)$.
Throughout, we let $\bphi_i$ denote the vector of length $p$ containing zeroes in all entries except a one in its $i$-th entry.

Let us write
\begin{gather*}
    \widehat{\bGamma}_{{\bxii}}(h) = \widehat{\bGamma}_{\bx}(\tau,h)  - \widehat{\mathbf{E}}_{\bx}  \widehat{\mathbf{E}}_{\bx}^\top \widehat{\bGamma}_{\bx}(\tau,h) - \widehat{\bGamma}_{\bx}(\tau,h)^\top \widehat{\mathbf{E}}_{\bx}  \widehat{\mathbf{E}}_{\bx}^\top + \widehat{\bGamma}_{\bchi}(\tau,h).
\end{gather*}
Denoting by $\widetilde{\bGamma}_{\bxii}(h) = \widehat{\bGamma}_\bx(\tau,h) - \widehat{\bGamma}_{\bchi}(h)$, it follows that
\begin{align*}
    \left|\widehat{\bGamma}_{{\bxii}}(h) - \bGamma_{\bxii}(h) \right|_\infty \leq  \left|\widehat{\bGamma}_{{\bxii}}(h) - \widetilde{\bGamma}_{\bxii}(h) \right|_\infty +  \left|\widetilde{\bGamma}_{\bxii}(h) - \bGamma_{\bxii}(h) \right|_\infty.
\end{align*}
By \Cref{lem:cov_xiidiff_rate}, we have the second term bounded as
\begin{align*}
\left|\widetilde{\bGamma}_{\bxii}(h) - \bGamma_{\bxii}(h) \right|_\infty = O_P\left( \psi_{n, p} \vee \frac{1}{\sqrt p} \right).  
\end{align*}
For the first term we have:
\begin{align*}
    \left|\widehat{\bGamma}_{{\bxii}}(h) - \widetilde{\bGamma}_{\bxii}(h) \right|_\infty &=  \left| 2\widehat{\bGamma}_{\bchi}(\tau,h) - \widehat{\mathbf{E}}_{\bx}  \widehat{\mathbf{E}}_{\bx}^\top \widehat{\bGamma}_\bx(\tau,h) - \widehat{\bGamma}_\bx(\tau,h)^\top \widehat{\mathbf{E}}_{\bx}  \widehat{\mathbf{E}}_{\bx}^\top \right|_\infty \\
    &\leq 2\left| \widehat{\mathbf{E}}_{\bx}  \widehat{\mathbf{E}}_{\bx}^\top \widehat{\bGamma}_\bx(\tau,h) (\textbf{I} - \widehat{\mathbf{E}}_{\bx}  \widehat{\mathbf{E}}_{\bx}^\top)  \right|_\infty \, .
\end{align*}
This final term is an estimate of the cross autocovariance between the estimated common and idiosyncratic component, $\widehat{\bchi}_t$, $\widehat{\bxii}_t$, hence the estimand is the $\textbf{O}$ matrix. We now bound this as
\begin{align*}
    \left| \widehat{\mathbf{E}}_{\bx}  \widehat{\mathbf{E}}_{\bx}^\top \widehat{\bGamma}_\bx(\tau,h) (\textbf{I} - \widehat{\mathbf{E}}_{\bx}  \widehat{\mathbf{E}}_{\bx}^\top)  \right|_\infty \leq &
    \left| \widehat{\mathbf{E}}_{\bx}  \widehat{\mathbf{E}}_{\bx}^\top \left( \widehat{\bGamma}_\bx(\tau,h) - \bGamma_\bx(h) \right)(\textbf{I} - \widehat{\mathbf{E}}_{\bx}  \widehat{\mathbf{E}}_{\bx}^\top)  \right|_\infty
    + \left| \widehat{\mathbf{E}}_{\bx}  \widehat{\mathbf{E}}_{\bx}^\top \bGamma_\bx(h) (\textbf{I} - \widehat{\mathbf{E}}_{\bx}  \widehat{\mathbf{E}}_{\bx}^\top)  \right|_\infty
\\
        \leq  &\left| \widehat{\mathbf{E}}_{\bx}  \widehat{\mathbf{E}}_{\bx}^\top \left( \widehat{\bGamma}_\bx(\tau,h) - \bGamma_\bx(h) \right) \right|_\infty \\
     + &\left| \widehat{\mathbf{E}}_{\bx}  \widehat{\mathbf{E}}_{\bx}^\top \left( \widehat{\bGamma}_\bx(\tau,h) - \bGamma_\bx(h) \right) \widehat{\mathbf{E}}_{\bx}  \widehat{\mathbf{E}}_{\bx}^\top  \right|_\infty \\
     + &\left| \left(\widehat{\mathbf{E}}_{\bx}  \widehat{\mathbf{E}}_{\bx}^\top - \mathbf{E}_{\bchi}\mathbf{E}_{\bchi}^\top \right) \bGamma_\bx(h) \left( \textbf{I}  -  \widehat{\mathbf{E}}_{\bx}  \widehat{\mathbf{E}}_{\bx}^\top
     \right)
     \right|_\infty \\
     + &\left| \left( \mathbf{E}_{\bchi}\mathbf{E}_{\bchi}^\top \right) \bGamma_\bx(h) \left(\widehat{\mathbf{E}}_{\bx}  \widehat{\mathbf{E}}_{\bx}^\top - \mathbf{E}_{\bchi}\mathbf{E}_{\bchi}^\top \right)
     \right|_\infty \\
     + &\left| \left(\mathbf{E}_{\bchi}  \mathbf{E}_{\bchi}^\top \right) \bGamma_{\bchi}(h) \left( \textbf{I}  -  \mathbf{E}_{\bchi}  \mathbf{E}_{\bchi}^\top
     \right)
     \right|_\infty \\
     + &\left| \left(\mathbf{E}_{\bchi}  \mathbf{E}_{\bchi}^\top \right) \bGamma_{\bxii}(h) \left( \textbf{I}  -  \mathbf{E}_{\bchi}  \mathbf{E}_{\bchi}^\top
     \right)
     \right|_\infty  = I + II + III + IV + V + VI.
\end{align*}

From \Cref{lem:fnets_E.20}.\ref{E.20.2}, and \Cref{prop:peligrad_inequality_trunc_autocov_error} we have that
\begin{align*}
    I &\leq \max_{1 \leq i \leq p} \left| \bphi_i^\top \widehat{\mathbf{E}}_{\bx}  \right|_2 \left\|\widehat{\mathbf{E}}_{\bx}\right\|_2  \max_{1 \leq j \leq p} \left| \left( \widehat{\bGamma}_\bx(\tau,h) - \bGamma_\bx(h) \right) \bphi_j  \right|_2 = O_P(\psi_{n,p}) \, , \, \\
     II &\leq \max_{1 \leq i \leq p} \left| \bphi_i^\top \widehat{\mathbf{E}}_{\bx}  \right|_2 \left\|\widehat{\mathbf{E}}_{\bx}\right\|_2 \left\| \left( \widehat{\bGamma}_\bx(\tau,h) - \bGamma_\bx(h) \right) \right\|_2 \left\|\widehat{\mathbf{E}}_{\bx}\right\|_2 \max_{1 \leq j \leq p} \left|  \widehat{\mathbf{E}}_{\bx} \bphi_j \right|_2 = O_P(\psi_{n,p}) \, .
\end{align*}
Next, from \Cref{lem:fnets_E.20}.\ref{E.20.1}-\ref{E.20.2}, and \eqref{eq:eigen_product_row_bound} we have that
\begin{align*}
     III &\leq  \max_{1 \leq i \leq p} \left| \bphi_i^\top \left(\widehat{\mathbf{E}}_{\bx}  \widehat{\mathbf{E}}_{\bx}^\top -\mathbf{E}_{\bchi} \mathbf{E}_{\bchi}^\top\right) \right|_2 \max_{1 \leq j \leq p} \left|\bGamma_\bx(h) \bphi_j\right|_2 +  \max_{1 \leq i \leq p} \left| \bphi_i^\top \left(\widehat{\mathbf{E}}_{\bx}  \widehat{\mathbf{E}}_{\bx}^\top -\mathbf{E}_{\bchi} \mathbf{E}_{\bchi}^\top\right) \right|_2 \left\|\bGamma_\bx(h)\right\|_2  \|  \widehat{\mathbf{E}}_{\bx}\|_2 \left|\widehat{\mathbf{E}}_{\bx}^\top\bphi_j\right|_2  \\
     &= O_P\left(\psi_{n,p} \vee \frac{1}{\sqrt{p}}\right) \\
     IV &\leq \max_{1 \leq i \leq p} \left| \bphi_i^\top \mathbf{E}_{\bchi}  \right|_2 \left\|\mathbf{E}_{\bchi}\right\|_2\left\|\bGamma_\bx(h)\right\|_2\max_{1 \leq j \leq p} \left| \left(\widehat{\mathbf{E}}_{\bx}  \widehat{\mathbf{E}}_{\bx}^\top -\mathbf{E}_{\bchi} \mathbf{E}_{\bchi}^\top\right) \bphi_j \right|_2 = O_P\left(\psi_{n,p} \vee \frac{1}{\sqrt{p}}\right)
\end{align*}
Next by using that the column span of the loadings $\bLambda$ is equal to the column span of $\mathbf{E}_{\bchi}$ we have
\begin{align*}
    V = &\left| \left(\mathbf{E}_{\bchi}  \mathbf{E}_{\bchi}^\top \right) \bGamma_{\bchi}(h) \left( \mathbf{E}_{\bchi}  \mathbf{E}_{\bchi}^\top  -  \mathbf{E}_{\bchi}  \mathbf{E}_{\bchi}^\top
     \right)
     \right|_\infty = 0 \, .
\end{align*}
Then, finally by \Cref{lem:fnets_E.20}.\ref{E.20.3} and  \eqref{eq:gamma:xi:bound} we have that
\begin{align*}
    VI \leq \max_{1 \leq i \leq p} &\left|\bphi_i^\top \mathbf{E}_{\bchi} \right|_2 \left\|\mathbf{E}_{\bchi}^\top \right\|_2 \left\|\bGamma_{\bxii}(h)\right\|_2 \left\|\left( \textbf{I}  -  \mathbf{E}_{\bchi}  \mathbf{E}_{\bchi}^\top
     \right)
     \right\|_2 = O_P\left(\frac{1}{\sqrt{p}}\right) \, .
\end{align*}
Therefore, we obtain that $ \left|\widehat{\bGamma}_{{\bxii}}(h) - \widetilde{\bGamma}_{\bxii}(h) \right|_\infty = O_P(\psi_{n,p})$, and hence
\begin{gather*}
\left|\widehat{\bGamma}_{{\bxii}}(h) - \bGamma_{\bxii}(h) \right|_\infty = O_P\left(\psi_{n,p} \vee \frac{1}{\sqrt{p}}\right) \, .
\end{gather*}

\subsubsection{Preliminary results for \Cref{prop:fac_adj_idio_rate}}

The following results hold under the conditions made in \Cref{prop:fac_adj_idio_rate}.

\begin{lemma} \label{lem:eigen_space_gap_bound} There exists a diagonal matrix $\bO$ with $\pm 1$ on its diagonal, such that
\begin{enumerate}[topsep=5pt,label=(\roman*)] 
\item $\| \widehat{\mathbf{E}}_{\bx} - 
\mathbf{E}_{\bchi} \bO\|_F = O_P\left(\psi_{n,p} \vee \frac{1}{p} \right)$ \label{eigen_space_gap_bound_2}.
   \item $ \frac{1}{p}\|\widehat{\bGamma}_{\bx}(\tau) - \bGamma_{\bchi}\|_2 = O_P\left(\psi_{n,p} \vee \frac{1}{p} \right)$.
\end{enumerate}
\end{lemma}
\begin{proof}
Applying Theorem 2 from \citet[][]{Yu2015} iteratively to the eigenvectors of $\widehat{\mathbf{E}}_{\bx}$ and 
$\mathbf{E}_{\bchi}$, there exists a diagonal orthogonal matrix $\bO$ such that
\begin{gather*}
    \| \widehat{\mathbf{E}}_{\bx} - 
\mathbf{E}_{\bchi} \bO\|_F \leq \frac{r \|\widehat{\bGamma}_{\bx}(\tau) - \bGamma_{\bchi}\|_2} {\min \left\{\underset{1 \leq i \leq r-1}{\min}(\mu_{{\bchi},i} - \mu_{{\bchi},i+1}), \mu_{{\bchi},r}\right\}} \, ,
\end{gather*}
the fraction is well defined since by \Cref{assump:fac_adj_fac_loading}.\ref{loading_assum} the eigenvalues are distinct. Now using that there exists a constant $C>0$ such that $\|\bGamma_{\bxii}(h)\|_2 \leq C$ for all $0 \leq h \leq d$ from \eqref{eq:gamma:xi:bound} and \eqref{eq:trunc_rate_denote}, we have
\begin{align}
    \frac{1}{p} \|\widehat{\bGamma}_{\bx}(\tau) - \bGamma_{\bchi}\|_2 \leq \frac{1}{p} \| \widehat{\bGamma}_{\bx}(\tau) - \bGamma_{\bx} \|_2 + \frac{1}{p} \| \bGamma_{\bx} - \bGamma_{\bchi} \|_2 = O_P\left(\psi_{n,p} \vee \frac{1}{p} \right),\nonumber 
    \\ 
    \text{hence} \quad \| \widehat{\mathbf{E}}_{\bx} - 
\mathbf{E}_{\bchi} \bO\|_F = r \cdot O_P\left(\psi_{n,p} \vee \frac{1}{p} \right) \, . \nonumber
\end{align}
\end{proof}

\begin{lemma} \label{lem:eig_hat}

\begin{enumerate}[topsep=0pt,label=(\roman*)] 
    \item $ \frac{1}{p}\hat{\mu}_{\bx,r} \geq  \alpha_r\left(1 + o_p\left( 1 \right)\right)$ \label{eig_x_hat_i} .
\item $ 1 / \hat{\mu}_{\bx,r} = O_P(1/p)$ .\label{eig_x_hat_ii}
\item $ 1 / \mu_{\bx,r} = O(1/p)$ \label{eig_x_hat_iii}
\end{enumerate}
\end{lemma}
\begin{proof}
By \Cref{lem:eigen_space_gap_bound}~(ii) and Weyl's inequality, for all $1 \leq j \leq r$, 
\begin{align}
    \frac{1}{p}|\hat{\mu}_{\bx,j} - \mu_{\bchi,j}| \leq
    \frac{1}{p}\|\widehat{\bGamma}_{\bx}(\tau) - \bGamma_{\bchi}\|_2 = O_P\left(\psi_{n,p} \vee \frac{1}{p} \right) \, .  \label{eq:weyl_common_cov}
\end{align}
Using \Cref{assump:fac_adj_fac_loading}.\ref{loading_assum},  we have that 
\begin{align}
     &\mu_{\bchi,r}  \geq p \alpha_r>0 \nonumber \\
     \implies &\frac{1}{p}\mu_{\bchi,r}  \geq  \alpha_r>0 \label{eq:eig_vaL_r_pos}
\end{align}
hence the matrix $\mathbf{M}_{\bchi}$ is invertible. Using \eqref{eq:weyl_common_cov}, \eqref{eq:eig_vaL_r_pos} and that $\psi_{n,p} \rightarrow 0$ as $n,p \rightarrow \infty$, gives 
\begin{align*}
    \frac{1}{p} \hat{\mu}_{\bx,r} &\geq \frac{1}{p} \mu_{\bchi,r} - \frac{1}{p} | \mu_{\bchi,r} - \hat{\mu}_{\bx,r}| \geq  \alpha_r + O_P\left(\psi_{n,p} \vee \frac{1}{p} \right) \\
       &\geq  \alpha_r\left(1 + o_p\left( 1 \right)\right) \, ,
\end{align*}
thus proving \ref{eig_x_hat_i}.
Rearranging \ref{eig_x_hat_i} gives us
$$
    \frac{1}{\hat{\mu}_{\bx,r}} \leq \frac{1}{p \alpha_r (1 + o_p(1))} \, ,
$$
and hence \ref{eig_x_hat_ii}.

Now by \eqref{eq:gamma:xi:bound}, and Weyl's inequality, for all $1 \leq j \leq r$, 
\begin{align}
    \frac{1}{p}|\mu_{\bx,j} - \mu_{\bchi,j}| \leq
    \frac{1}{p}\|\bGamma_{\bx} - \bGamma_{\bchi}\|_2 = O\left(\frac{1}{p} \right) \, .  \label{eq:weyl_common_cov_pop}
\end{align}
Using \eqref{eq:weyl_common_cov_pop}, \eqref{eq:eig_vaL_r_pos} and that $\psi_{n,p} \rightarrow 0$ as $n,p \rightarrow \infty$, gives 
\begin{align*}
    \frac{1}{p} \mu_{\bx,r} &\geq \frac{1}{p} \mu_{\bchi,r} - \frac{1}{p} | \mu_{\bchi,r} - \mu_{\bx,r}| \geq  \alpha_r + O\left(\frac{1}{p} \right) \\
       &\geq  \alpha_r\left(1 + o\left( 1 \right)\right) \, ,
\end{align*}
and then rearranging gives us
$$
    \frac{1}{\mu_{\bx,r}} \leq \frac{1}{p \alpha_r (1 + o(1))} \, ,
$$
and hence \ref{eig_x_hat_iii}.
\end{proof}

\begin{lemma} \label{lem:fnets_E.20}
$ \ $
\begin{enumerate}[topsep=0pt,label=(\roman*)]
    \item $ \max_{1 \leq j \leq r} \max_{1 \leq i \leq p} | e_{\bchi,ij}| = O(1/\sqrt{p})$. \label{E.20.1}
\item $ \max_{1 \leq j \leq r}  \max_{1 \leq i \leq p} | \widehat{e}_{\bx,ij}| = O_P(1/\sqrt{p})$.\label{E.20.2}
\item $ \max_{1 \leq j \leq r}  \max_{1 \leq i \leq p} | e_{\bx,ij}| = O(1/\sqrt{p})$.\label{E.20.3}
\end{enumerate}
\end{lemma}

\begin{proof}
By \Cref{assump:fac_adj_fac_loading}.\ref{finite_loadings} and that the $r$ is fixed, we have that
$ \, \gamma_{\bchi,ii}  = \blambda_i^\top \blambda_i \leq C \quad \forall 1 \leq i \leq p $.
Denote by $\bGamma_{\bchi}(h) = \left[\gamma_{\bchi,ii'}(h), 1 \leq i,i', \leq p \right]$.
Since $\, \gamma_{\bchi,ii} =  \sum_{j =1}^r \mu_{\bchi, j} |e_{\bchi,ij}|^2 = O(1)$, it follows that
\begin{gather*}     
\max_{1 \leq j \leq r} \mu_{\bchi, j}^{1/2} \max_{1 \leq i \leq p} |e_{\bchi,ij}| = O(1) \, .
\end{gather*}
Then under  \Cref{assump:fac_adj_fac_loading}.\ref{loading_assum}, we have \ref{E.20.1}.

As for \ref{E.20.2}, we have from \Cref{prop:peligrad_inequality_trunc_autocov_error} and \eqref{eq:X_bounded},
\begin{align*}
    \max_{1 \leq i \leq p} \widehat{\gamma}_{\bx,ii} \leq \max_{1 \leq i \leq p} \left| \widehat{\gamma}_{\bx,ii} - \gamma_{\bx,ii}\right|+ \max_{1 \leq i \leq p} \left|\gamma_{\bx,ii} \right| = O_P(1),
\end{align*}
which implies 
\begin{align*}
\max_{1 \leq i \leq p} \sum_{j=1}^p \widehat{\mu}_{\bx, j} |\widehat{e}_{\bx,ij}|^2 = O_P(1).
\end{align*}
Thus using \Cref{lem:eig_hat}.\ref{eig_x_hat_ii}, \ref{E.20.2} follows.

For \ref{E.20.3}, we have from \eqref{eq:X_bounded}, some positive constant $C>0$ such that
\begin{align*}
    \max_{1 \leq i \leq p} \left|\gamma_{\bx,ii} \right| \leq C \, ,
\end{align*}
which implies
\begin{align*}
\max_{1 \leq i \leq p} \sum_{j=1}^p \mu_{\bx, j} |e_{\bx,ij}|^2 = O(1).
\end{align*}
Thus using \Cref{lem:eig_hat}.\ref{eig_x_hat_iii}, \ref{E.20.3} follows.
\end{proof}

\begin{lemma} \label{lem:eigenspace_row}
\begin{gather*}
    \sqrt{p} \max_{1 \leq i \leq p} \left| \left(\widehat{\mathbf{E}}_{\bx} - 
\mathbf{E}_{\bchi} \bO \right)_i\right|_2 = O_P\left( \psi_{n,p} \vee \frac{1}{\sqrt{p}} \right) .
\end{gather*}
\end{lemma}
\begin{proof}
Recall that $\bphi_i$ denotes a vector of length $p$ containing zeroes in all entries except a one in its $i$-th entry.
Then,
\begin{align}
    \sqrt{p} \max_{1 \leq i \leq p} \left| \left(\widehat{\mathbf{E}}_{\bx} - 
\mathbf{E}_{\bchi} \bO \right)_i\right|_2  &= \frac{1}{\sqrt{p}}
\max_{1 \leq i \leq p} \left| \bphi_i^\top\left(\widehat{\bGamma}_{\bx}(\tau) \widehat{\mathbf{E}}_{\bx}\left( \frac{\widehat{\mathbf{M}}_{\bx}}{p}\right)^{-1} - 
\bGamma_{\bchi} \mathbf{E}_{\bchi} \left( \frac{\mathbf{M}_{\bchi}}{p}\right)^{-1} \bO \right)\right|_2  \nonumber \\ &\leq
 \frac{1}{\sqrt{p}} \max_{1 \leq i \leq p} \left|\bphi_i^\top\left(\widehat{\bGamma}_{\bx}(\tau)  - \bGamma_{\bchi} \right)\right|_2  \left\| \widehat{\mathbf{E}}_{\bx}\right\|_2  \left\| \left( \frac{\mathbf{M}_{\bchi}}{p}\right)^{-1} \right\|_2  \label{eq:eigen_row_bound}\\
&+ \frac{1}{\sqrt{p}} \max_{1 \leq i \leq p} \left| \bphi_i^\top \bGamma_{\bchi} \right|_2 \left\| \widehat{\mathbf{E}}_{\bx} \right\|_2 \left\| \left( \frac{\widehat{\mathbf{M}}_{\bx}}{p}\right)^{-1} - \left( \frac{\mathbf{M}_{\bchi}}{p}\right)^{-1} \right\|_2 \nonumber \\
&+ \frac{1}{\sqrt{p}} \max_{1 \leq i \leq p} \left|   \bphi_i^\top \bGamma_{\bchi}  \right|_2 \left\| \widehat{\mathbf{E}}_{\bx} - \mathbf{E}_{\bchi} \bO 
 \right\|_2 \left\| \left( \frac{\mathbf{M}_{\bchi}}{p}\right)^{-1} \right\|_2 \, . \nonumber
\end{align} 
Now we look at the individual parts of this decomposition. By \Cref{lem:eig_hat}.\ref{eig_x_hat_i}, for large enough $n$ and $p$, the matrix $\widehat{\mathbf{M}}_{\bx}$ is invertible with high probability
\begin{gather*}
    \left\| \left( \frac{\mathbf{M}_{\bchi}}{p}\right)^{-1}\right\|_2 = \frac{p}{\mu_{\bchi,r}} \quad \text{and} \quad \left\| \left( \frac{\widehat{\mathbf{M}}_{\bx}}{p}\right)^{-1}\right\|_2 = \frac{p}{\hat{\mu}_{\bx,r}} \leq \frac{1}{p^{\rho_r - 1}\alpha_r\left(1 + o_p\left( 1 \right)\right)} \leq \frac{p}{\alpha_r\left(1 + o_p\left( 1 \right)\right)} \; .
\end{gather*}
We now need a bound on $\left\| \widehat{\mathbf{M}}_{\bx}^{-1} -  \mathbf{M}_{\bchi}^{-1} \right\|_2 $, from \eqref{eq:weyl_common_cov}, we have
\begin{gather}
    \left\| \widehat{\mathbf{M}}_{\bx}^{-1} -  \mathbf{M}_{\bchi}^{-1}\right\|_2 =  \max_{1 \leq i \leq r} \left|\frac{1}{\hat{\mu}_{\bx,i}} - \frac{1}{\mu_{\bchi,i}} \right| \nonumber \nonumber \\
    = \frac{|\mu_{\bchi,i} - \hat{\mu}_{\bx,i}|}{\mu_{\bchi,i} \hat{\mu}_{\bx,i}} = \frac{ O_P\left(\psi_{n,p} \vee \frac{1}{p} \right)}{p\alpha_r\left(1 + o_p\left( 1 \right)\right)} = O_P\left( \frac{1}{p}\left(\psi_{n,p} \vee \frac{1}{p} \right) \right) \, . \label{eq:eigenvalue_diff_bound}
\end{gather}
In \eqref{eq:eigenvalue_diff_bound}, we have used that for stochastic sequences $X_n = O_P(a_n)$, $Y_n = o_p(1)$, then \newline $X_n / (1 \pm Y_n) = O_P(a_n)$. Then for a bound on the $\, \max_{1 \leq i \leq p} \left|\bphi_i^\top\left(\widehat{\bGamma}_{\bx}(\tau) - \bGamma_{\bchi} \right)\right|_2$ term we have
\begin{align}
    \frac{1}{\sqrt{p}} \max_{1 \leq i \leq p} \left|\bphi_i^\top\left(\widehat{\bGamma}_{\bx}(\tau)  - \bGamma_{\bchi} \right)\right|_2 &\leq \frac{1}{\sqrt{p}} \max_{1 \leq i \leq p}   \left| \left( \widehat{\bGamma}_{\bx}(\tau) - \bGamma_{\bx} \right)_i \right|_2 + \frac{1}{\sqrt{p}}\left\| \bGamma_{\bx} - \bGamma_{\bchi} \right\|_2  \nonumber \\
    &= O_P\left(\psi_{n,p} \vee \frac{1}{\sqrt{p}}\right) \label{eq:row_cov_common_bound} \, .
\end{align}

Returning to \eqref{eq:eigen_row_bound}, and using both \eqref{eq:eigenvalue_diff_bound} and \eqref{eq:row_cov_common_bound} we get
\begin{gather*}
    \sqrt{p} \max_{1 \leq i \leq p} \left| \left(\widehat{\mathbf{E}}_{\bx} - 
\mathbf{E}_{\bchi} \bO \right)_i\right|_2 = O_P\left( \psi_{n,p} \vee \frac{1}{\sqrt{p}} \right)
\end{gather*}
\end{proof}

\begin{lemma} \label{lem:cov_xiidiff_rate}
Let $\widetilde{\bGamma}_{\bxii}(h) = \widehat{\bGamma}_\bx(\tau,h) - \widehat{\bGamma}_{\bchi}(h)$, and assume that $\psi_{n,p} \rightarrow 0$ as $n,p \rightarrow \infty$, we then have the rate
\begin{gather*}
    \left|\widetilde{\bGamma}_{\bxii}(h) - \bGamma_{\bxii}(h)\right|_{\infty} = O_P\left( \psi_{n,p} \vee \frac{1}{\sqrt{p}} \right) \, .
\end{gather*}
\end{lemma}
\begin{proof}
By \Cref{assump_fac_adj:8}, we have the following covariance decomposition
\begin{align}
    \bGamma_{\bx}(h) = \bGamma_{\bchi}(h) + \bGamma_{\bxii}(h)
    \label{eq:factor_cov_decompostion}
\end{align} 
for all fixed $h \ge 0$.
Let us write
\begin{gather} 
    |\widetilde{\bGamma}_{\bxii}(h) - \bGamma_{\bxii}(h)|_{\infty} \leq  
    |\widehat{\bGamma}_{\bx}(\tau,h) - \bGamma_{\bx}(h)|_{\infty} +
    |\widehat{\bGamma}_{\bchi}(h) - \bGamma_{\bchi}(h)|_{\infty}.
    \nonumber
\end{gather}
From \Cref{prop:peligrad_inequality_trunc_autocov_error}, we already have a bound for the truncated covariance of the observed data.
Consequently, the proof for \Cref{lem:cov_xiidiff_rate} amounts to bounding the element-wise max norm estimation error for the covariance matrix of the common component. 

From \eqref{eq:factor_pca_solutions} we have that 
\begin{align}
    \widehat{\bchi}_t &= \widehat{\mathbf{E}}_{\bx}  \widehat{\mathbf{E}}_{\bx}^\top (\bX_t({\tau})) \, , \quad \text{and hence} \quad \widehat{\bGamma}_{\bchi}(h) = \widehat{\mathbf{E}}_{\bx}  \widehat{\mathbf{E}}_{\bx}^\top \widehat{\bGamma}_{\bx}(\tau,h) \widehat{\mathbf{E}}_{\bx}  \widehat{\mathbf{E}}_{\bx}^\top \, . \label{eq:samp_cov_chi}
\end{align}
Additionally, using that the column span of the loadings $\bLambda$ is equal to the column span of $\mathbf{E}_{\bchi}$, 
\begin{gather}
    \mathbf{E}_{\bchi} (\mathbf{E}_{\bchi}^\top\mathbf{E}_{\bchi})^{-1}\mathbf{E}_{\bchi}^\top \bLambda \bF_t = \mathbf{E}_{\bchi} \mathbf{E}_{\bchi}^\top \bLambda \bF_t = \bLambda \bF_t \, , \quad \text{and hence}  \nonumber \\
    \bGamma_{\bchi}(h) = \mathbf{E}_{\bchi} \mathbf{E}_{\bchi}^\top \bGamma_{\bchi}(h)  \mathbf{E}_{\bchi} \mathbf{E}_{\bchi}^\top \label{eq:true_cov_chi} \, .
\end{gather}
From \eqref{eq:samp_cov_chi}, \eqref{eq:true_cov_chi} and \eqref{eq:factor_cov_decompostion}, we can decompose the common component max norm error into two parts:
\begin{align} 
    \left|\widehat{\bGamma}_{\bchi}(h) - \bGamma_{\bchi}(h)\right|_{\infty} &=
    \left|\widehat{\mathbf{E}}_{\bx}  \widehat{\mathbf{E}}_{\bx}^\top \widehat{\bGamma}_\bx(\tau,h)\widehat{\mathbf{E}}_{\bx}  \widehat{\mathbf{E}}_{\bx}^\top  - \mathbf{E}_{\bchi} \mathbf{E}_{\bchi}^\top \bGamma_{\bchi}(h)  \mathbf{E}_{\bchi} \mathbf{E}_{\bchi}^\top \right|_{\infty} \nonumber \\
    &\leq  \left| \widehat{\mathbf{E}}_{\bx}  \widehat{\mathbf{E}}_{\bx}^\top \widehat{\bGamma}_\bx(\tau,h)\widehat{\mathbf{E}}_{\bx}  \widehat{\mathbf{E}}_{\bx}^\top  - \mathbf{E}_{\bchi} \mathbf{E}_{\bchi}^\top \bGamma_{\bx}(h)  \mathbf{E}_{\bchi} \mathbf{E}_{\bchi}^\top \right|_{\infty} + \left| \mathbf{E}_{\bchi} \mathbf{E}_{\bchi}^\top  \bGamma_{\bxii}(h) \mathbf{E}_{\bchi} \mathbf{E}_{\bchi}^\top\right|_{\infty} \label{eq:common_var_error_decomp}\\
    &= I + II \nonumber
\end{align}
First looking at the term $II = \max_{1 \leq i,j \leq p} \left|\bphi_i^\top \mathbf{E}_{\bchi} \mathbf{E}_{\bchi}^\top  \bGamma_{\bxii}(h) \mathbf{E}_{\bchi} \mathbf{E}_{\bchi}^\top  \bphi_j\right|$, we have
\begin{align}
    II &\leq
    \max_{1 \leq i,j \leq p} \left| \mathbf{E}_{\bchi}^\top \bphi_i\right|_{2} \left| \mathbf{E}_{\bchi}^\top  \bGamma_{\bxii}(h) \mathbf{E}_{\bchi}  \mathbf{E}_{\bchi}^\top \bphi_j\right|_2 \nonumber \\
    &\leq \max_{1 \leq i \leq p} \left|\mathbf{E}_{\bchi}^\top \bphi_i \right|_{2} \left\| \mathbf{E}_{\bchi}^\top  \bGamma_{\bxii}(h) \mathbf{E}_{\bchi} \right\|_{2} \max_{1 \leq j \leq p} \left| \mathbf{E}_{\bchi}^\top \bphi_j\right|_2 \label{eq:fac_adj_part_2} \\ 
    &\leq \max_{1 \leq i \leq p} \left|\mathbf{E}_{\bchi}^\top \bphi_i \right|_{2} \left\| \mathbf{E}_{\bchi}^\top  \right\|_2 \left\| \bGamma_{\bxii}(h) \right\|_2 \left\|\mathbf{E}_{\bchi} \right\|_{2} \max_{1 \leq j \leq p} \left| \mathbf{E}_{\bchi}^\top \bphi_j\right|_2 \nonumber \, .
\end{align}
From \Cref{lem:fnets_E.20}.\ref{E.20.1} we have that $\max_{1 \leq i \leq p} |\mathbf{E}_{\bchi}^\top \bphi_i |_{2} = O(\sqrt{r}/ \sqrt{p}) \, $. And combining with $\|\mathbf{E}_{\bchi} \|_{2}$ = 1,  and that there exists a constant $C>0$ such that $\|\bGamma_{\bxii}(h)\|_2 \leq C$ for all $0 \leq h \leq d$ from \eqref{eq:gamma:xi:bound}, we get the rate $II = O(r/p)$.

Now looking at the term $I = | \widehat{\mathbf{E}}_{\bx}  \widehat{\mathbf{E}}_{\bx}^\top \widehat{\bGamma}_\bx(\tau,h)\widehat{\mathbf{E}}_{\bx}  \widehat{\mathbf{E}}_{\bx}^\top  - \mathbf{E}_{\bchi} \mathbf{E}_{\bchi}^\top \bGamma_{\bx}(h)  \mathbf{E}_{\bchi} \mathbf{E}_{\bchi}^\top |_{\infty}$,
\begin{align*}
    I \leq &\left|\left( \widehat{\mathbf{E}}_{\bx}  \widehat{\mathbf{E}}_{\bx}^\top -\mathbf{E}_{\bchi} \mathbf{E}_{\bchi}^\top\right)\widehat{\bGamma}_\bx(\tau,h)\widehat{\mathbf{E}}_{\bx}  \widehat{\mathbf{E}}_{\bx}^\top\right|_{\infty}  + \nonumber\\
    &\left|\mathbf{E}_{\bchi} \mathbf{E}_{\bchi}^\top \left(\widehat{\bGamma}_\bx(\tau,h) - \bGamma_{\bx}(h) \right) \widehat{\mathbf{E}}_{\bx}  \widehat{\mathbf{E}}_{\bx}^\top\right|_{\infty}  + \\
    &\left|\mathbf{E}_{\bchi} \mathbf{E}_{\bchi}^\top \widehat{\bGamma}_\bx(\tau,h)\left(\mathbf{E}_{\bchi} \mathbf{E}_{\bchi}^\top 
 - \widehat{\mathbf{E}}_{\bx}  \widehat{\mathbf{E}}_{\bx}^\top\right)\right|_{\infty} = III + IV + V \, .\nonumber
\end{align*}
Using the same arguments as in \eqref{eq:fac_adj_part_2}, and that from \Cref{lem:fnets_E.20}.\ref{E.20.2} we have
$\max_{i} \|\widehat{\mathbf{E}}_{\bx}^\top \bphi_i \|_{2} = O_P(\sqrt{r}/ \sqrt{p}) \, $, we get that
\begin{gather*}
    IV = O_P\left(\frac{r}{p} \cdot \psi_{n,p}\right) \, .
\end{gather*}
Then for $III$ we have
\begin{align}
III &= \max_{1 \leq i,j \leq p} \left|\bphi_i^\top \left( \widehat{\mathbf{E}}_{\bx}  \widehat{\mathbf{E}}_{\bx}^\top -\mathbf{E}_{\bchi} \mathbf{E}_{\bchi}^\top\right)\widehat{\bGamma}_\bx(\tau,h)\widehat{\mathbf{E}}_{\bx}  \widehat{\mathbf{E}}_{\bx}^\top \bphi_j \right| \nonumber \\
&\leq \max_{1 \leq i \leq p} \left| \bphi_i^\top \left(\widehat{\mathbf{E}}_{\bx}  \widehat{\mathbf{E}}_{\bx}^\top -\mathbf{E}_{\bchi} \mathbf{E}_{\bchi}^\top\right) \right|_2 \max_{1 \leq j \leq p} \left|\widehat{\bGamma}_\bx(\tau,h)\widehat{\mathbf{E}}_{\bx}  \widehat{\mathbf{E}}_{\bx}^\top \bphi_j \right|_2 \nonumber \\
&\leq \max_{1 \leq i \leq p} \left| \bphi_i^\top \left(\widehat{\mathbf{E}}_{\bx}  \widehat{\mathbf{E}}_{\bx}^\top -\mathbf{E}_{\bchi} \mathbf{E}_{\bchi}^\top\right) \right|_2 \left\|\widehat{\bGamma}_\bx(\tau,h)\right\|_2 \left\|\widehat{\mathbf{E}}_{\bx}\right\|_2  \max_{1 \leq j \leq p} \left|\widehat{\mathbf{E}}_{\bx}^\top \bphi_j \right|_2 \label{eq:IV_bound} \, .
\end{align}
Similarly, for $V$ we have
\begin{align}
    V \leq  \max_{1 \leq i \leq p} \left| \bphi_i^\top \widehat{\mathbf{E}}_{\bx}  \right|_2   \left\|\widehat{\mathbf{E}}_{\bx}\right\|_2  \left\|\widehat{\bGamma}_\bx(\tau,h)\right\|_2  \max_{1 \leq j \leq p}  \left| \bphi_j^\top  \left(\widehat{\mathbf{E}}_{\bx}  \widehat{\mathbf{E}}_{\bx}^\top -\mathbf{E}_{\bchi} \mathbf{E}_{\bchi}^\top\right) \right|_2\label{eq:VI_bound} \, .
\end{align}
The term $ \max_{1 \leq i \leq p}  \left| \bphi_i^\top  \left(\widehat{\mathbf{E}}_{\bx}  \widehat{\mathbf{E}}_{\bx}^\top -\mathbf{E}_{\bchi} \mathbf{E}_{\bchi}^\top\right) \right|_2 $ appearing in \eqref{eq:IV_bound} and \eqref{eq:VI_bound} can be decomposed in the following way
\begin{align*}
\max_{1 \leq i \leq p} \left| \bphi_i^\top \left(\widehat{\mathbf{E}}_{\bx}  \widehat{\mathbf{E}}_{\bx}^\top -\mathbf{E}_{\bchi} \mathbf{E}_{\bchi}^\top\right) \right|_2 &\leq \max_{1 \leq i \leq p} \left| \bphi_i^\top \left(\widehat{\mathbf{E}}_{\bx}   -\mathbf{E}_{\bchi} \bO\right)\widehat{\mathbf{E}}_{\bx}^\top \right|_2 + \max_{1 \leq i \leq p} \left| \bphi_i^\top \mathbf{E}_{\bchi} \bO\left(\widehat{\mathbf{E}}_{\bx}   -\mathbf{E}_{\bchi} \bO\right)^\top \right|_2 \\
&= VII + VIII \, .
\end{align*}
Then by \Cref{lem:eigenspace_row} and \Cref{lem:eigen_space_gap_bound}.\ref{eigen_space_gap_bound_2} we have
\begin{align*}
VII  &\leq \max_{1 \leq i \leq p} \left|  \left(\widehat{\mathbf{E}}_{\bx}   -\mathbf{E}_{\bchi} \bO\right)_i \right|_2  =  O_P\left( \frac{1}{\sqrt{p}}\left(\psi_{n,p} \vee \frac{1}{\sqrt{p}} \right)\right)\\
VIII &\leq \max_{1 \leq i \leq p} \left|(\mathbf{E}_{\bchi} \bO)_i\right|_2 \left\| \left(\widehat{\mathbf{E}}_{\bx}   -\mathbf{E}_{\bchi} \bO\right) \right\|_2 =  O_P\left( \frac{r^{3/2}}{\sqrt{p}}\left(\psi_{n,p} \vee \frac{1}{p} \right)\right)
\end{align*}
and therefore
\begin{gather}
    \max_{1 \leq i \leq p} \left| \bphi_i^\top \left(\widehat{\mathbf{E}}_{\bx}  \widehat{\mathbf{E}}_{\bx}^\top -\mathbf{E}_{\bchi} \mathbf{E}_{\bchi}^\top\right) \right|_2 = O_P\left( \frac{1}{\sqrt{p}}\left(\psi_{n,p} \vee \frac{1}{\sqrt{p}} \right)\right) \label{eq:eigen_product_row_bound} \, .
\end{gather}
Returning to the bounds \eqref{eq:IV_bound} and \eqref{eq:VI_bound}, using \eqref{eq:eigen_product_row_bound} we have
\begin{align*}
    III &\leq \max_{1 \leq i,j \leq p}\left| \bphi_i^\top \left(\widehat{\mathbf{E}}_{\bx}  \widehat{\mathbf{E}}_{\bx}^\top -\mathbf{E}_{\bchi} \mathbf{E}_{\bchi}^\top\right) \right|_2 \left\|\widehat{\bGamma}_\bx(\tau,h)\right\|_2 \left\|\widehat{\mathbf{E}}_{\bx}\right\|_2  \left|\widehat{\mathbf{E}}_{\bx}^\top \bphi_j \right|_2 \\
    &\overset{(1)}{=} O_P \left( \left( \frac{1}{\sqrt{p}} \left(\psi_{n,p} \vee \frac{1}{\sqrt{p}} \right)\right)p\left(1 + \psi_{n,p} \right) \left(\frac{\sqrt{r}}{\sqrt{p}}\right)\right) \\
    &\overset{(2)}{=} O_P \left( \sqrt{r} \left(\psi_{n,p} \vee \frac{1}{\sqrt{p}} \right)\right) \, ,
\end{align*}
where for step $(1)$ we used that $\left\|\widehat{\bGamma}_\bx(\tau,h)\right\|_2 \leq 
p\left[ \left|\bGamma_\bx(h) \right|_{\infty} + \left|\widehat{\bGamma}_\bx(\tau,h) - \bGamma_\bx(h) \right|_{\infty} \right]
$, and that $\left|\bGamma_\bx(h) \right|_{\infty}$ is constant i.e $O(1)$. And step $(2)$, has used that $\psi_{n,p} \rightarrow 0$ as $n,p \rightarrow \infty$. With similar steps we can then show that
\begin{align*}
    V = O_P \left( \sqrt{r} \left(\psi_{n,p} \vee \frac{1}{\sqrt{p}} \right)\right) \, .
\end{align*}
Now, since the $O_P$ rates for $III$ and $V$ are the same, and are larger than the $O_P$ rate for $IV$, and $II$, by returning to the decomposition \eqref{eq:common_var_error_decomp}, we obtain the bound
\begin{gather*}
    \left|\widehat{\bGamma}_{\bchi}(h) - \bGamma_{\bchi}(h)\right|_{\infty} = O_P \left( \psi_{n,p} \vee \frac{1}{\sqrt{p}} \right) \, , \, \text{and hence} \\
     \left|\widetilde{\bGamma}_{\bxii}(h) - \bGamma_{\bxii}(h)\right|_{\infty} = O_P\left( \psi_{n,p} \vee \frac{1}{\sqrt{p}} \right) \, .    
\end{gather*}
\end{proof}

\begin{lemma} \label{lem:eig_est_rows_bound}
\begin{gather*}
    \max_{1 \leq i \leq p} \left| \bphi_i^\top \left(\widehat{\mathbf{E}}_{\bx}  \widehat{\mathbf{E}}_{\bx}^\top -\mathbf{E}_{\bx} \mathbf{E}_{\bx}^\top\right) \right|_2 = O_P\left( \frac{1}{\sqrt{p}}\left(\psi_{n,p} \vee \frac{1}{\sqrt{p}} \right)\right)  \, .
\end{gather*}
\end{lemma}

\begin{proof}
By the triangle inequality we have
\begin{align*}
     \max_{1 \leq i \leq p} \left| \bphi_i^\top \left(\widehat{\mathbf{E}}_{\bx}  \widehat{\mathbf{E}}_{\bx}^\top -\mathbf{E}_{\bx} \mathbf{E}_{\bx}^\top\right) \right|_2 &\leq 
     \max_{1 \leq i \leq p} \left| \bphi_i^\top \left(\widehat{\mathbf{E}}_{\bx}  \widehat{\mathbf{E}}_{\bx}^\top -\mathbf{E}_{\bchi} \mathbf{E}_{\bchi}^\top\right) \right|_2 + \max_{1 \leq i \leq p} \left| \bphi_i^\top \left(\mathbf{E}_{\bchi}  \mathbf{E}_{\bchi}^\top -\mathbf{E}_{\bx} \mathbf{E}_{\bx}^\top\right) \right|_2 \\
     &= I + II \, .
\end{align*}
Now from \eqref{eq:eigen_product_row_bound} we have
\begin{align*}
    I = O_P\left( \frac{1}{\sqrt{p}}\left(\psi_{n,p} \vee \frac{1}{\sqrt{p}} \right)\right).
\end{align*}
Then for any orthogonal matrix $\bO$ we have that
\begin{align*}
     II \leq \left\|
     \left(\mathbf{E}_{\bchi}  \mathbf{E}_{\bchi}^\top -\mathbf{E}_{\bx} \mathbf{E}_{\bx}^\top\right) \right\|_2 \leq 2 \left\|
     \mathbf{E}_{\bchi} -\mathbf{E}_{\bx} \bO \right\|_2 \, ,
\end{align*}
and by Theorem 2 from \citet[][]{Yu2015} we have,
\begin{align*}
    \left\|
     \mathbf{E}_{\bchi} -\mathbf{E}_{\bx} \bO \right\|_2 \leq \frac{\left\|\bGamma_{\bx} - \bGamma_{\bchi}\right\|_2}{\mu_{{\bchi},r}} = O_P\left(\frac{1}{p}\right) \, ,
\end{align*}
putting together the bounds on $I$ and $II$ the conclusion follows.
\end{proof}

\subsection{Proof of \Cref{prop:FNETS_trunc_consistency}}

Throughout, we suppose \Cref{assump:VAR} and that $\mathcal{E}_{n, p}$ holds.
We denote the estimands of $\widehat{\bGamma}, \widehat{\bgamma}$, by 
\begin{align*}
\bGamma = \begin{bmatrix} \bGamma_{\bxii}(0) & \dots & \bGamma_{{\bxii}}(d-1) \\
\vdots & \ddots & \vdots \\
\bGamma_{{\bxii}}(d - 1)^\top & \dots & \bGamma_{{\bxii}}(0) \end{bmatrix} \text{ \ and \ }
\bgamma = \begin{bmatrix}
\bGamma_{{\bxii}}(1) \\
\bGamma_{{\bxii}}(2) \\
\vdots \\
\bGamma_{{\bxii}}(d)
\end{bmatrix},  
\end{align*}
respectively.
For all $\mathbf{v} \in \mathbb{R}^{dp}$, we have
\begin{align}
    \bbv^\top \widehat{\bGamma} \bbv  &= |\bbv^\top\bGamma\bbv - \bbv^\top(\bGamma -\widehat{\bGamma} )\bbv|
    \geq |\bbv^\top\bGamma\bbv| - |\bbv^\top(\widehat{\bGamma} -\bGamma )\bbv| \nonumber \\
    &\geq \Lambda_{\min }(\bGamma)|\bbv|_2^2 - |\widehat{\bGamma} - \bGamma|_{\infty}|\bbv|_1^2 \nonumber \\
    &\geq \Lambda_{\min }(\bGamma)|\bbv|_2^2 - C \left(\psi_{n,p} \vee p^{-1/2}\right)|\bbv|_1^2 \, .
    \nonumber
\end{align}
Further, by Proposition 2.3 in \citet{Basu2015}, we have that, for any $\ell \geq 1$
\begin{align*}
2 \pi \mathfrak{m}\left(\boldsymbol{\Sigma}_{\mathbf{x}}\right) \leq \Lambda_{\min }\left(\Upsilon_{\ell}^X\right) \leq \Lambda_{\max }\left(\Upsilon_{\ell}^X\right) \leq 2 \pi \mathcal{M}\left(\boldsymbol{\Sigma}_{\mathbf{x}}\right) \, ,
\end{align*}
where $\Upsilon_{\ell}^X = \operatorname{Cov}\left( \operatorname{vec}\left(\boldsymbol{\mathcal{Z}}^\top \right), \operatorname{vec}\left(\boldsymbol{\mathcal{Z}}^\top \right)\right)$, with $\boldsymbol{\mathcal{Z}} = \left[\bX_{\ell}, \dots , \bX_1\right]^\top$. Therefore, for $\ell=d$ we have $\Lambda_{\min}(\bGamma) \geq 2 \pi \mathfrak{m}\left(\boldsymbol{\Sigma}_{\mathbf{x}}\right)$. Furthermore, by \eqref{eq:gamma:xi:bound}, we have that $\Lambda_{\min}(\bGamma) \geq 2 \pi \mathfrak{m}\left(\boldsymbol{\Sigma}_{\mathbf{x}}\right) > 2\pi m_{\bxii} > 0$, where $m_{\bxii}$ is defined in \eqref{eq:gamma:xi:bound}.
Hence we have
\begin{align}
    \bbv^\top \widehat{\bGamma} \bbv \geq 2\pi m_{\bxii} |\bbv|_2^2 - C \left(\psi_{n,p} \vee p^{-1/2}\right)|\bbv|_1^2.
\label{eq:re}
\end{align}
Next, recall the notation in \eqref{eq:robust_var_lasso}, with $\widehat{\mathbb{A}} = \left[\widehat{\bbeta}_{1}, \dots, \widehat{\bbeta}_{p}\right]^\top$, we now additionally denote the rows of the true parameters by $\mathbb{A} = \left[\bbeta^{*}_{1}, \dots, \bbeta^{*}_{p}\right]^\top$. 
Then by the Yule-Walker equation, we have $\bgamma_{(j)} = \bGamma\bbeta^{*}_{j}$, such that
\begin{align*}
    |\widehat{\bgamma} _{(j)} - \widehat{\bGamma} \bbeta^{*}_{j}|_{\infty} &\leq |\widehat{\bgamma} _{(j)} - {\bgamma}_{(j)}|_{\infty} + |\bGamma\bbeta^{*}_{j} - \widehat{\bGamma} \bbeta^{*}_{j}|_{\infty} \nonumber \\
    &\leq |\widehat{\bgamma} _{(j)} - {\bgamma}_{(j)}|_{\infty} + |\bGamma -\widehat{\bGamma} |_{\infty} |\bbeta^{*}_{j}|_1 \nonumber \\
    &\leq C(|\bbeta^{*}_{j}|_1 + 1) \left(\psi_{n,p} \vee \frac{1}{\sqrt{p}}\right)  \, ,
\end{align*}
on $\mathcal{E}_{n, p}$, and thus
\begin{align}
    |\widehat{\bgamma} _{(j)} - \widehat{\bGamma} \bbeta^{*}_{j}|_{\infty} &\leq  C (\|\mathbb{A}\|_{\infty} + 1)\left(\psi_{n,p} \vee \frac{1}{\sqrt{p}}\right) \quad \forall \, 1 \leq j \leq p \, . \label{eq:db_bound}
\end{align}

Equipped with \eqref{eq:re} and~\eqref{eq:db_bound}, the remainder of the proof follows that of Proposition 4.1 from \citet{Basu2015}.
By definition of $\widehat{\bbeta}$ we have that
\begin{gather*}
    -2 \widehat{\bbeta}_{j}^\top \widehat{\bgamma} _{(j)} + \widehat{\bbeta}_{j}^\top \widehat{\bGamma} \widehat{\bbeta}_{j} + \lambda |\widehat{\bbeta}_{j}|_1 \leq -2\bbeta^\top\widehat{\bgamma} _{(j)} + \bbeta^\top \widehat{\bGamma} \bbeta + \lambda |\bbeta|_1 \quad
\end{gather*}
for any $\bbeta \in \mathbb{R}^{pd}$.
Setting $\bbeta = \bbeta^*_{j}$ and rearranging the above inequality, we get
\begin{gather*}
    \bv^\top \widehat{\bGamma} \bv \leq 2 \bv^\top \left(\widehat{\bgamma} _{(j)} - \widehat{\bGamma} \bbeta^*_{j}\right) + \lambda \left\{|\bbeta^*_{j}|_1 - |\bbeta^*_{j} + \bv|_1\right\} \, ,
\end{gather*}
where $\bv = \widehat{\bbeta}_{j} - \bbeta^*_{j}$.

Now let $S$ denote the index of non-zero entries of $\bbeta^*$, and $\bv_S$ the sub-vector of $\bv$ containing the elements indexed by the set $S$, we have
\begin{align}
    \bv^\top \widehat{\bGamma}  \bv 
    &\leq 2 |\bv|_1 \left|\widehat{\bgamma} _{(j)} - \widehat{\bGamma} \bbeta^*_{j}\right|_{\infty} + \lambda \left\{|\bv_S|_1 - |\bv_{S^c}|_1\right\} \nonumber\\
    &\leq  |\bv|_1 \frac{\lambda}{2}  + \lambda \left\{|\bv_S|_1 - |\bv_{S^c}|_1\right\} \label{eq:using_db_bound}\\
    &\leq  \frac{\lambda}{2}\left(|\bv_S|_1 + |\bv_{S^c}|_1 \right) + \lambda \left\{|\bv_S|_1 - |\bv_{S^c}|_1\right\} \nonumber\\
    &\leq \frac{\lambda}{2}\left(3|\bv_S|_1 - |\bv_{S^c}|_1\right) \label{eq:sparse_supp_bound} \\
    &\leq \frac{\lambda}{2}\left(4|\bv_S|_1 - |\bv|_1\right) \leq 2\lambda |\bv|_1 \nonumber \, ,
\end{align}
where \Cref{eq:using_db_bound} follows from by our choice of $\lambda$, and the deviation bound in \eqref{eq:db_bound}. In summary, we have that $\bv^\top \widehat{\bGamma}  \bv \leq 2\lambda |\bv|_1$, and from \eqref{eq:sparse_supp_bound} we have
\begin{gather}
        |\bv_{S^c}|_1 \leq 3|\bv_S|_1 \nonumber\\
        \implies |\bv|_1 \leq 4|\bv_S|_1 \leq 4\sqrt{s_{0,j}} |\bv|_2 \label{eq:v1_v2_bound}
\end{gather}
Then from~\eqref{eq:re}, a lower bound for $\bv^\top \widehat{\bGamma}  \bv$ can be obtained as
\begin{align}
    \bv^\top \widehat{\bGamma}  \bv &\geq 2\pi m_{\bxii} |\bv|_2^2 - C \left(\psi_{n,p} \vee p^{-1/2}\right) |\bv|_1^2 \nonumber\\
    &\geq  2\pi m_{\bxii} |\bv|_2^2 - C \left(\psi_{n,p} \vee p^{-1/2}\right) 16 s_{\text {\upshape max}} |\bv|_2^2 \label{eq:low_bound_quad_tol_sparse} \\
    &\geq \pi m_{\bxii}|\bv|_2^2 \, ,\label{eq:low_bound_quad}
\end{align}
where \eqref{eq:low_bound_quad_tol_sparse} follows from \eqref{eq:v1_v2_bound}, and \eqref{eq:low_bound_quad} from the assumption that $s_{\text {\upshape max}} C \left(\psi_{n,p} \vee p^{-1/2}\right) \leq \pi m_{\bxii}/16$. Now combining the lower and upper bounds we have
\begin{gather*}
    0 < \frac{\alpha}{2}|\bv|_2^2 \leq \bv^\top \widehat{\bGamma}  \bv \leq 2\lambda |\bv|_1 \leq 8\lambda \sqrt{s_{0,j}} |\bv|_2 \, ,
\end{gather*}
which then gives us
\begin{align*}
    |\widehat{\bbeta}_{(j)} - \bbeta^{*}_{j}|_2 \leq \frac{8\lambda \sqrt{s_{0, j}}}{\pi m_{\bxii}} \,  \quad \text{and} \quad   |\widehat{\bbeta}_{(j)} - \bbeta^{*}_{j}|_1 \leq \frac{32\lambda s_{0, j}}{\pi m_{\bxii}} \, ,
\end{align*}
and hence the concluding bounds
\begin{align*}
    |\widehat{\mathbb{A}} - \mathbb{A}|_2 \leq \frac{8\lambda \sqrt{s_{0}}}{\pi m_{\bxii}} \,  \quad \text{and} \quad   |\widehat{\mathbb{A}} - \mathbb{A}|_1 \leq \frac{32\lambda s_{0}}{\pi m_{\bxii}} \, .
\end{align*}

\end{document}